%% file: soliton_fc_perturbation.tex
\documentclass{IEEEtran}
\usepackage{url}
\usepackage{amsmath}
\usepackage{amssymb}
\usepackage{amsthm}
\usepackage{authblk}
\usepackage{upgreek}
\usepackage[ruled,algonl,lined]{algorithm2e}
\usepackage{cite}
\usepackage[dvipsnames]{xcolor}
\newtheorem{theorem}{Theorem}

\newtheorem*{corollary*}{Corollary}

\newtheorem*{remark}{Remark}

\usepackage[normalem]{ulem}
\DeclareMathAlphabet{\mathbit}{OML}{cmr}{bx}{it}
\DeclareMathAlphabet{\mathsf}{OT1}{cmss}{m}{n}
\DeclareMathAlphabet{\mathbsf}{OT1}{cmss}{bx}{it}
\newcommand{\inC}[1]{\ensuremath{\in\mathbb{C}^{#1}}}
\newcommand{\inR}[1]{\ensuremath{\in\mathbb{R}^{#1}}}
\newcommand{\inset}[2]{\ensuremath{\in \left\{#1,\ldots,#2\right\}}}
\newcommand{\Real}[1]{\ensuremath{\Re\left\{#1\right\}}}

\newcommand\diff[1]{\ensuremath{\:\mathrm{d}#1}}
\newcommand\pderiv[2]{\ensuremath{\frac{\partial#1}{\partial#2}}}
\newcommand\pderivk[3]{\ensuremath{\frac{\partial^{#3}#1}{\partial{#2}^{#3}}}}
\newlength{\figurewidth}
\newlength{\figureheight}

\usepackage{tikz}
\usetikzlibrary{arrows,decorations,backgrounds,shadows,plotmarks,calc, positioning}
\usetikzlibrary{arrows.meta}

\usepackage{pgfplots}
\pgfplotsset{compat=newest}
\pgfplotsset{plot coordinates/math parser=false}
\pgfplotsset{every axis/.append style={font=\footnotesize}}
\pgfplotsset{
	ylabel right/.style={
		after end axis/.append code={
			\node [rotate=90, anchor=north] at (rel axis cs:1,0.5) {#1};
		}   
	}
}

 \usepackage{fancyhdr}
 \fancypagestyle{first}{%
 \chead{\small (This work has been submitted to the IEEE/OSA Journal of Lightwave Technology for possible publication. Copyright may be transferred without notice, after which this version may no longer be accessible.)}
 \rhead{\thepage}
 \fancyfoot{}
 }

 \fancypagestyle{others}{%
 \rhead{\thepage}
 \fancyfoot{}
 }
\title{Statistics of the Nonlinear Discrete Spectrum of a Noisy Pulse}
\author{Francisco Javier Garc\'ia-G\'omez and Vahid Aref
\IEEEcompsocitemizethanks{
	\IEEEcompsocthanksitem Date of current version \today. J. Garc\'ia-G\'omez is with the Institute for Communications Engineering (LNT), Technical University of Munich, Germany (e-mail: javier.garcia@tum.de). His work was supported by the German Research Foundation under Grant KR 3517/8-1.
	\IEEEcompsocthanksitem V. Aref is with Nokia Bell Labs, Stuttgart 70435, Germany (e-mail: vahid.aref@nokia-bell-labs.com). 
	}
}
\begin{document}
\maketitle


\begin{abstract}
In the presence of additive Gaussian noise, the statistics of the nonlinear Fourier transform (NFT) of a pulse are not yet completely known in closed form. In this paper, we propose a novel approach to study this problem. Our contributions are twofold: first, we extend the existing Fourier Collocation (FC) method to compute the whole discrete spectrum (eigenvalues and spectral amplitudes). We show numerically that the accuracy of FC is comparable to the state-of-the-art NFT algorithms. 
Second, we apply perturbation theory of linear operators to derive analytic expressions for the joint statistics of the eigenvalues and the spectral amplitudes when a pulse is contaminated by additive Gaussian noise. Our analytic expressions closely match the empirical statistics obtained through simulations. 
\end{abstract}


 \thispagestyle{first}
 \pagestyle{others}

\section{Introduction}
The Nonlinear Fourier Transform (NFT), or Inverse Scattering Transform~\cite{ablowitz1981ist, mansoor_parts} has been proposed as an alternative for system design in an attempt to overcome the capacity peak reported in~\cite{essiambre_limits} of linear transmission systems over the nonlinear optical channel. An overview on the NFT and its application for optical communications is given in~\cite{turitsyn2017overview}.

Communication using the NFT (continuous or discrete spectrum) has been demonstrated numerically and experimentally, e.g. \cite{turitsyn2017overview, aref2018contdisc, hari_multieigenvalue, civelli2018dualpol,gui2017alternative, le2018gbps}. Despite some promising results, the effect of channel noise on the NFT is not yet well understood.
Recently, a general method has been developed
in \cite{sander_statistics} to numerically
compute the statistics of the spectral coefficients (not the eigenvalues) of a signal with additive white Gaussian noise (AWGN). This method, however, requires knowledge of the time-domain signal.
{In the case of propagation along an optical fiber with \textit{distributed} AWGN along the fiber, the statistics of the NFT of a first-order soliton are well known~\cite{gordon86soliton, zhang_perturbation, derevyanko_pdf}.
Recently, the statistics of the eigenvalues of an arbitrary pulse with the same propagation model were derived~\cite{prati2018some}.}
{All these results were obtained based on some perturbative methods.}

This work builds up on our previous paper~\cite{garcia2018statistics} and makes a twofold contribution. First, we extend the existing Fourier Collocation (FC) method~\cite[Sec. 2.4.3]{yang_nlse} to compute the complete discrete spectrum (eigenvalues and spectral amplitudes) of arbitrary pulses. We apply a proper windowing and truncation to overcome the ringing problem caused by non-periodic boundary conditions of NFT. Our simulations show that our method achieves a comparable accuracy to the best existing methods when the number of time samples is low.

Second, we apply the perturbation theory of matrix eigenvalues~\cite{kato_eigenvalues} to our extended method to derive the statistics of the discrete spectrum when a pulse is contaminated by Gaussian noise. Note that the eigenvalue perturbation theory is also used in~\cite[Part III, Sec. IV-A]{mansoor_parts} {and~\cite{prati2018some} to find the statistics} of the eigenvalues.
{Our method is novel
 in two aspects: we analyze the statistics in the frequency domain, and we provide a single method to compute both eigenvalues and spectral amplitudes, thus allowing computation of cross-correlations between the two.}
We show that
our analytic expressions for the statistics of the discrete spectrum closely match the statistics obtained through Monte-Carlo simulations.

The paper is organized as follows. In Sec.~\ref{sec:preliminaries}, we briefly overview the NFT and the multi-soliton pulses.
In Sec.~\ref{sec:fc}, we describe the FC method and extend it to compute also the discrete spectral amplitudes by proper windowing and truncation. 
Applying a first-order perturbation method,
%first-order perturbation theory, 
we derive   
the statistics of the discrete spectrum for an arbitrary pulse in Sec.~\ref{sec:perturbation}.
We validate in Sec.~\ref{sec:numerical} both our extended FC method and our analytic expressions for the statistics of the discrete spectrum through simulations with different pulses. Sec.~\ref{sec:conclusion} concludes the paper.

\textit{Notation:} Bold lowercase letters $\mathbf{x}$ denote vectors, and bold uppercase letters $\mathbf{X}$ denote matrices. $\mathbf{X}^{\mathrm{T}}$ and $\mathbf{X}^{\mathrm{H}}$ are respectively the transpose and conjugate transpose of $\mathbf{X}$. $(\mathbf{x}, \mathbf{y})$ is the horizontal concatenation of vectors $\mathbf{x}$ and $\mathbf{y}$. $\Re x$ is the real part of $x$, and $\Im x$ is its imaginary part.

\section{System model}\label{sec:preliminaries}

Consider the complex envelope $A(Z, \tau)$ of an electrical field propagating along an optical fiber, where $Z$ is distance and $\tau$ is time. The propagation is modeled by
the Nonlinear Schr\"odinger Equation (NLSE)~\cite[Eq. (2.3.46)]{agrawal_nfo}:
\begin{equation}
\pderiv{A(Z, \tau)}{Z}= -j\frac{\beta_2}{2}\pderivk{A(Z, \tau)}{T}{2} +j\gamma\left|A(Z, \tau)\right|^2 A(Z, \tau) %
%\pderiv{}{Z}Q= -j\frac{\beta_2}{2}\pderivk{}{T}{2}Q+j\gamma\left|Q\right|^2 Q
\label{eq:nlse_analog}
\end{equation}
where $\beta_2$ is the chromatic dispersion
parameter, and $\gamma$ is the nonlinear coefficient. We neglect attenuation in~\eqref{eq:nlse_analog} assuming that it is compensated by distributed amplification. %The noise term $N(Z, T)$ is the formal derivative of a band-limited Wiener process and is circularly symmetric Gaussian with auto-correlation function
%\begin{equation}
%\mathrm{E}\left[N(Z, T)N^*(Z', T')\right] B_{\mathrm{noise}}\sinc\left(B_{\mathrm{noise}}\left(T-T'\right)\right)\delta(Z-Z')
%\end{equation}
%where $\sinc(x)\triangleq\sin(\pi x)/(\pi x)$, the function $\delta(\cdot)$ is a Dirac delta, and $B_{\mathrm{noise}}$is the noise bandwidth. 
% By applying the following change of variables:
% \begin{equation}
% \tau=T_0t,\hspace{8pt} 
% Z=2\frac{T_0^2}{\left|\beta_2\right|}z, 
% \hspace{8pt}
% A(Z, \tau)=\frac{1}{T_0}\sqrt{\frac{\left|\beta_2\right|}{\gamma}}q(z, t)
% \label{eq:normalize}
% \end{equation}
% (where $T_0$ is a free parameter), the NLSE~\eqref{eq:nlse_analog} is normalized to
With proper normalization~\cite{mansoor_parts}, the NLSE can be transformed to
\begin{equation}
\pderiv{}{z}q(z, t)=j\pderivk{}{t}{2}q(z, t) +j2\left|q(z, t)\right|^2 q(z, t) %\nonumber \\
%&+n(z, t).
\label{eq:nlse}
\end{equation}
where $q$ is the normalized signal, $t$ is the normalized time, and $z$ is the normalized distance.
%\begin{align}
%\pderiv{}{z}q(z, t)&=-j\sign\left(\beta_2\right)\pderivk{}{t}{2}q(z, t) +j2\left|q(z, t)\right|^2 q(z, t) \nonumber \\
%&+n(z, t).
%\label{eq:nlse}
%\end{align}

\subsection{Description in Nonlinear Spectrum} 
The NFT is calculated by solving the Zakharov-Shabat system (ZSS)~\cite{ablowitz1981ist, mansoor_parts}
\begin{equation}
\left(\begin{matrix}
-\frac{\partial}{\partial t} & q(t) \\ q^*(t) & \frac{\partial}{\partial t}
\end{matrix}\right)\left(\begin{matrix}
v_1(t,\lambda) \\ v_2(t,\lambda)
\end{matrix}\right)=j\lambda \left(\begin{matrix}
v_1(t,\lambda) \\ v_2(t,\lambda)
\end{matrix}\right)
\label{eq:zs}
\end{equation}
with the boundary condition
\begin{equation}
\mathbf{v}(t, \lambda)\to\left(\begin{matrix}
1\\0
\end{matrix}\right)e^{-j\lambda t}, \quad t\to-\infty
\label{eq:boundary}
\end{equation}
{where  $\mathbf{v}(t,\lambda)=(v_1(t, \lambda) v_2(t,\lambda))^{\mathrm{T}}$ is the \textit{Jost solution}.} % We define the \textit{adjoint} of a vector $v=(v_1, v_2)^{\mathrm{T}}$ as $\tilde{v}=(v_2^*, -v_1^*)^{\mathrm{T}}$.
The \textit{spectral coefficients} $a(\lambda)$ and $b(\lambda)$ are given by
\begin{subequations}
	\begin{align}
	a(\lambda)&=\lim\limits_{t\to\infty} v_1(t, \lambda)e^{j\lambda t} \\
	b(\lambda)&=\lim\limits_{t\to\infty} v_2(t, \lambda)e^{-j\lambda t}.
	\end{align}
	\label{eq:ab_lim}
\end{subequations}
The NFT of the signal $q(t)$ is made up of two spectra:
\begin{itemize}
	\item the \textit{continuous spectrum} $Q(\xi)=\frac{b(\xi)}{a(\xi)}$, for $\xi\in\mathbb{R}$;
	\item the \textit{discrete spectrum} $Q_k=\frac{b(\lambda_k)}{a_\lambda(\lambda_k)}$, for the $K$ eigenvalues $\left\{\lambda_k\in\mathbb{C}^+\colon a(\lambda_k)=0\right\}$
\end{itemize}
where $a_{\lambda}=\mathrm{d}a/\mathrm{d}\lambda$ and $\mathbb{C}^+=\{\lambda\in\mathbb{C}\colon\Im{\lambda}>0\}$.

In recent works~\cite{wahls2017generation, le2018gbps, gui2017alternative}, the use of {$b_k=b(\lambda_k)$ instead of $Q_k$ for data modulation in the discrete spectrum has been shown to achieve better results}. In terms of degrees of freedom, both approaches are equivalent, as given $\lambda_k$ {and $b(\xi)$ for $\xi\in\mathbb{R}$, one can compute $a(\lambda)$ as~\cite[Ch. I, Eq. (6.23)]{faddeev2007hamiltonian}
\begin{equation}
    a(\lambda)=\exp\left[\frac{1}{2\pi j}\int\limits_{-\infty}^{\infty}\frac{\log (1-|b(\xi)|^2)}{\xi-\lambda}\diff{\xi}\right]\prod_{k=1}^{K}\frac{\lambda-\lambda_k}{\lambda-\lambda_k^*}.
    \label{eq:afromcont}
\end{equation}
For this reason, in this paper we obtain the statistics of $\lambda_k$ and $b_k$, but not of $Q_k$.}
 %a(\xi)$ for $\xi\in\mathbb{R}$  
%\begin{equation}
%a(\xi)=\sqrt{1-|b(\xi)|^2}\exp \left(-\frac{j}{2\pi}\int_{-\infty}^{\infty}\frac{\log\left(1-|b(\mu)|^2\right)}{\mu-\xi}\diff{\mu}\right)
%\end{equation}
%and $a(\lambda)$ for $\lambda\inC{+}$ as
%\begin{equation}
%a(\lambda)= \left(\prod_{k=1}^{K}\frac{\lambda-\lambda_k}{\lambda-\lambda_k^*}\right)\exp \left(-\frac{j}{2\pi}\int\limits_{-\infty}^{\infty}\frac{\log\left(1-|b(\mu)|^2\right)}{\mu-\lambda}\diff{\mu}\right)
%\end{equation}
%
%Two additional solutions $\tilde{v}^1(t, \lambda^*)$ and $\tilde{v}^2(t, \lambda^*)$ of~\eqref{eq:zs} are calculated by solving $v_t(t, \lambda^*)=P(t, \lambda^*)v(t, \lambda^*)$ using boundary conditions adjoint to~\eqref{eq:bc}, and taking the adjoint of the solutions. The four \textit{canonical eigenvectors} $v^1(t, \lambda)$, $v^2(t, \lambda)$, $\tilde{v}^1(t, \lambda)$, and $\tilde{v}^2(t, \lambda)$ satisfy%
%\begin{subequations}%
%	\begin{align}
%	v^2(t, \lambda)&=a(\lambda)\tilde{v}^1(t, \lambda^*)+b(\lambda)v^1(t, \lambda) \\
%	\tilde{v}^2(t, \lambda)&=-b^*(\lambda^*)\tilde{v}^1(t, \lambda^*)+a^*(\lambda^*)v^1(t, \lambda)
%	\end{align}%
%	\label{eq:ab}%
%\end{subequations}%
%where $a(\lambda)$ and $b(\lambda)$ do not depend on $t$.  To compute the NFT, the following relations are useful
The usefulness of the NFT lies in the fact that, given a signal $q(z, t)$ propagating according to the NLSE~\eqref{eq:nlse}, the evolution of its NFT in $z$ is multiplicative:
% \begin{subequations}
	\begin{align}
	Q(z, \xi)&=Q(0, \xi)e^{4j\xi^2 z} &
	\lambda_k(z)&=\lambda_k(0) \nonumber \\
	b_k(z)&=b_k(0)e^{4j\lambda_k^2 z} &
	a(z, \lambda)&=a(0, \lambda). 
	\label{eq:nft_z}
	\end{align}
% \end{subequations}
We skip index $z$ in the sequel for simplicity.
Many numerical algorithms have been developed to compute NFT (e.g.,~\cite{mansoor_parts,turitsyn2017overview}). 
As a benchmark, we use in this paper the forward-backward iterations~\cite{hari2016bidierctional, aref_control_detection} combined with a recently proposed algorithm in \cite{chimmalgi2018cfqm} with a sixth-order Commutator-Free Quasi-Magnus (CFQM) integrator $\mathrm{CF}_4^{[6]}$. {Its sixth-order accuracy is the best we have found in the literature}. We replace the trapezoidal integration of~\cite{aref_control_detection} with the CFQM to improve the performance. We call this algorithm FB-CFQM.

\subsection{Multi-soliton Pulses}
A \textit{multi-soliton pulse} has no continuous spectrum, i.e., $b(\xi)=0$. In Algorithm~\ref{algo:darboux}, we provide pseudo-code that uses the Darboux transform~\cite{matveev1991darboux} to construct a time-domain multi-soliton pulse and its Jost solutions from the discrete spectrum. Some other Inverse NFT algorithms are reviewed in \cite{turitsyn2017overview}.
\begin{algorithm}[t]
	%	\SetKwInOut{Input}{Input}
	%	\SetKwInOut{Output}{Output}
	%\dontprintsemicolon
	%	\Input{$K$, $(\lambda_k,Q_d(\lambda_k))$ for $k=1,\dots,K$.}
	%	\Output{$K-$soliton $q^{(K)}(t)$ and Jost Solutions $v_k^{(K)}(t), k=1,\dots,K$.}
% 	\BlankLine
	%\Begin{
	%		\tcc{initialization}
	\tcc{initialize the Jost solutions $\mathbf{v}_k^{(0)}(t)$. The superscript $^{(i)}$ indicates the algorithm iteration number}
	\For{$k\leftarrow 1$ \KwTo $K$}{
		%$b(\lambda_i)=\frac{Q_d(\lambda_i)}{\lambda_i-\lambda_i^*}\prod_{k=1,k\ne i}^K \frac{\lambda_i-\lambda_k}{\lambda_i-\lambda_k^*}$\;
		$\mathbf{v}^{(0)}_k(t)=(e^{-j\lambda_k t},-b_k e^{j\lambda_k t})^{\mathrm{T}}$\;
	}
	$q^{(0)} (t)= 0$\;
	\tcc{iteratively add $(\lambda_i,b_i)$}
	\For{$i\leftarrow 1$ \KwTo $K$}{
		$(f_1,f_2)= \mathbf{v}_i^{(i-1)}(t)$\;
		\tcc{update signal}
		$q^{(i)}(t)= q^{(i-1)}(t)-2j(\lambda_i-\lambda_i^*)\frac{f_2^*(t)f_1(t)}{|f_1(t)|^2+|f_2(t)|^2}$\;%\tcc*[l]{signal update}
		
		\tcc{update $\mathbf{v}_i^{(i)}(t)$}
		$C=b_i \prod_{k=1}^{i-1}\left(\lambda_i-\lambda_k\right) \prod_{k=i+1}^{K}1/\left(\lambda_i-\lambda_k^*\right)$\;%\tcc*[l]{compensation of previous and forthcoming evolution}
		$\displaystyle \mathbf{v}_i^{(i)}(t)=\frac{C}{|f_1(t)|^2+|f_2(t)|^2}\left(\begin{matrix}-f_2^*(t)\\f_1^*(t)\end{matrix}\right)$\;
		
		\tcc{update $\mathbf{v}_k^{(i)}(t)\triangleq(v_{k,1}^{(i)}(t), v_{k,2}^{(i)}(t))^T$}
		\For{$k\leftarrow 1$ \KwTo $K$; $k\ne i$}{
			$v_{k,1}^{(i)}(t)= \left(\lambda_k -\lambda_i^*-\frac{(\lambda_i-\lambda_i^*)|f_1(t)|^2}{|f_1(t)|^2+|f_2(t)|^2}
			\right) v_{k,1}^{(i-1)}(t) -\frac{(\lambda_i-\lambda_i^*)f_2^*(t)f_1(t)}{|f_1(t)|^2+|f_2(t)|^2}v_{k,2}^{(i-1)}(t)$\;
			$v_{k,2}^{(i)}(t)= 
			-\frac{(\lambda_i-\lambda_i^*)f_2(t)f_1^*(t)}{|f_1(t)|^2+|f_2(t)|^2}v_{k,1}^{(i-1)}(t)
			+\left(\lambda_k -\lambda_i+\frac{(\lambda_i-\lambda_i^*)|f_1(t)|^2}{|f_1(t)|^2+|f_2(t)|^2}
			\right)v_{k,2}^{(i-1)}(t)$\;
			%				$v_{k}^{(i)}(t)=\left(v_{k,1}^{(i)}(t),v_{k,2}^{(i)}(t)\right)^{\mathrm{T}}$\;
		}
		%			\tcc{update of the Jost solutions of forthcoming eigenvalues}
		%			\For{$k\leftarrow i+1$ \KwTo $K$}{
		%				$v_{k,1}^{(i)}(t)= \left(\lambda_k -\lambda_i^*-\frac{(\lambda_i-\lambda_i^*)|f_1(t)|^2}{|f_1(t)|^2+|f_2(t)|^2}
		%				\right)v_{k,1}^{(i-1)}(t) -\frac{(\lambda_i-\lambda_i^*)f_2^*(t)f_1(t)}{|f_1(t)|^2+|f_2(t)|^2}v_{k,2}^{(i-1)}(t)$\;
		%				$v_{k,2}^{(i)}(t)= 
		%				-\frac{(\lambda_i-\lambda_i^*)f_2(t)f_1^*(t)}{|f_1(t)|^2+|f_2(t)|^2}v_{k,1}^{(i-1)}(t)
		%				+\left(\lambda_k -\lambda_i+\frac{(\lambda_i-\lambda_i^*)|f_1(t)|^2}{|f_1(t)|^2+|f_2(t)|^2}
		%				\right)v_{k,2}^{(i-1)}(t)$\;
		%				$v_{k}^{(i)}(t)=\left(v_{k,1}^{(i)}(t),v_{k,2}^{(i)}(t)\right)^{\mathrm{T}}$\;
		%				
		%			}
	}
	{\textbf{Output:} $q(t)=q^{(K)}(t)$ and $\mathbf{v}(t, \lambda_i)=\mathbf{v}_i^{(K)}(t)$}
	%		$q(t)= q^{(K)}(t)$\tcc*[l]{The desired multi-soliton pulse.}
	%		\For{$i\leftarrow 1$ \KwTo $K$}{
	%			$v(t;\lambda_k)= v_{k}^{(K)}(t)$\tcc*[l]{The desired Jost solutions of eigenvalue $\lambda_k$}
	%			$v(t;\lambda_k^*)=\left(v_{k,2}^{(K)}(t),-v_{k,1}^{(K)}(t)\right)^{\mathrm{H}} $\tcc*[l]{The desired Jost solutions of eigenvalue $\lambda_k^*$ ($H$ denotes Hermitian transform)}
	%		}
	
	%}	
	\caption{{Darboux Transform to compute $K-$soliton $q(t)$ and its Jost Solutions
	$\mathbf{v}_k(t)\triangleq \mathbf{v}(t, \lambda_k)$, $1\leq k\leq K$ from discrete spectrum $\{(\lambda_k,b_k)\}_{k=1}^K$.}}
	\label{algo:darboux}
\end{algorithm}

\section{The Fourier Collocation (FC) Method}\label{sec:fc}
The FC method~\cite[Sec. 2.4.3]{yang_nlse}, or \textit{spectral method}~\cite[Part II]{mansoor_parts}
finds the
discrete eigenvalues of the NFT of a signal $q(t)$. 
This is done by setting up a matrix eigenvalue problem in the linear frequency domain.
Assume that $q(t)$ is only nonzero in a finite time interval\footnote{A pulse $q(t)$ can be confined in $[-T/2, T/2]$ by a time-shift $t_0$. The spectrum of $q(t-t_0)$ has the same $\lambda_k$ but the $b$-coefficients are $b_k\exp(-j\lambda_k t_0)$.} 
$[-T/2, T/2]$. In this case, we can trivially compute $\mathbf{v}(t,\lambda)$ for $t\notin[-T/2, T/2]$ in terms of $\mathbf{v}(\pm T/2, \lambda)$. We need only to find $\mathbf{v}(t,\lambda)$ for $t\in[-T/2, T/2]$.  

Assume that the periodic extensions with period $T$ of $q(t)$ and $v_i(t,\lambda_k)$ for $t\in[-T/2, T/2]$ are band-limited to the frequency band $[-N/T, N/T]$, where $N$ is an integer. {By performing these periodic extensions}, we can express $q(t)$ and $v_i(t,\lambda_k)$ for $t\in[-T/2, T/2]$ and $i\in\{1,2\}$ by a Fourier series of $M=2N+1$ terms
\begin{equation}
q(t)=\!\sum_{n=-N}^{N}\!c[n] e^{jn\frac{2\pi}{T}t}, \hspace{3pt} v_{i}(t,\lambda_k)=\!\!\sum_{n=-N}^{N}\!\psi_{k,i}[n]e^{jn\frac{2\pi}{T}t}.
\label{eq:fs}
\end{equation}

The Fourier coefficients can  be then computed as a discrete Fourier transform (DFT) of the sampled pulses.
Define $q[m]=q(t_m)$ where $t_m=m\frac{T}{M},m\inset{-N}{N}$. Hence,
\begin{equation}
    c[n]=\frac{1}{M}\sum_{m=-N}^{N} q[m] e^{-j\frac{2\pi}{M}mn}.
\end{equation}
In a similar manner, $\psi_{k,i}[n]$ can be obtained from $v_{k, i}[m]=v_{i}(t_m, \lambda_k)$ for $i\in\{1,2\}$.
Let $\boldsymbol{\uppsi}_k=(\psi_{k,1}[-N],\dots,\psi_{k,1}[N],\psi_{k,2}[-N],\dots,\psi_{k,2}[N])^{ T}$.
Substituting~\eqref{eq:fs} into~\eqref{eq:zs} yields
\begin{equation}
%\mathbf{L}\boldsymbol{\uppsi}_\lambda=\lambda\boldsymbol{\uppsi}_\lambda
\mathbf{L}\boldsymbol{\uppsi}_k=\lambda_k\boldsymbol{\uppsi}_k
\label{eq:fc}
\end{equation}
where
\begin{equation}
\mathbf{L}=\left(\begin{matrix}\boldsymbol{\Omega} & \boldsymbol{\Gamma} \\ -\boldsymbol{\Gamma}^{\mathrm{H}} & -\boldsymbol{\Omega} \end{matrix}\right),
\label{eq:L}
\end{equation}
$\boldsymbol{\Omega}=-\frac{2\pi}{T}\mathrm{diag}\left(-N,\ldots,N\right)$, and $\mathbf{\Gamma}\inC{M\times M}$ is a Toeplitz matrix whose first column is 
$-j(c[0],\dots,c[N], 0,\dots, 0)^{\mathrm{T}}$
and whose first row is 
$-j(c[0],\dots,c[-N]\,\dots,0)$.

A solution $(\lambda_k,\boldsymbol{\uppsi}_k)$ of \eqref{eq:fc} corresponds to the Fourier coefficients of a solution $\mathbf{v}(t,\lambda_k)$ of  the ZSS~\eqref{eq:zs}. Accordingly,
the eigenvalues of $\mathbf{L}$ include the eigenvalues $\lambda_k$ of the discrete spectrum, their complex conjugates, and $M-2K$ spurious eigenvalues
{which are usually observed with rather small imaginary parts.}
%{which tend to be close to the real axis}.

\subsection*{Computation of Spectral Amplitudes $b_k$}
We extend the FC method to also compute the spectral amplitudes $b_k$. %
% {Consider the \textit{canonical eigenvectors} $\mathbf{v}$, $\overline{\mathbf{v}}$, $\mathbf{x}$, $\overline{\mathbf{x}}$, solutions of~\eqref{eq:zs} with boundary conditions given by~\eqref{eq:boundary} and
% \begin{subequations}
% \begin{align}
%     \overline{\mathbf{v}}(t, \lambda)\to \left(0, -e^{j\lambda t}\right), & \quad t \to -\infty \\
%     \mathbf{x}(t, \lambda)\to \left(0, e^{j\lambda t}\right), & \quad t \to \infty \\
%     \overline{\mathbf{x}}(t, \lambda)\to \left(e^{-j\lambda t}, 0\right), & \quad t \to \infty 
% \end{align}
% \end{subequations}
% Then we have~\cite[P. I]{mansoor_parts}:
% \begin{subequations}
% \begin{align}
%     \mathbf{v}(t, \lambda) & =a(\lambda)\overline{\mathbf{x}}(t, \lambda)+b(\lambda)\mathbf{x}(t, \lambda) \\
%     \overline{\mathbf{v}}(t, \lambda) & =-b^*(\lambda^*)\overline{\mathbf{x}}(t, \lambda)+a^*(\lambda^*)\mathbf{x}(t, \lambda). \label{eq:v_bar}
% \end{align}
% \label{eq:scatteting}
% \end{subequations}
% % From~\eqref{eq:afromcont}, we have $a*(\lambda^*)\to\infty$ when $\lambda\to\lambda_k$. As $\mathbf{v}(t, \lambda_k)=b(\lambda_k)\mathbf{x}(t, \lambda)$ is nonzero for almost all $t$ (because it is an analytic function of $t$), Eq.~\eqref{eq:v_bar} implies that at least one of $\overline{\mathbf{x}}(t, \lambda)$ and $\overline{\mathbf{v}}(t, \lambda)$ is unbounded for $\lambda\to\lambda_k$ and \textit{almost all} $t$. As 
% %
%  }
{%Furthermore, the eigenspace of $\lambda_k$ has dimension $2$ with basis $\{\mathbf{v}(t, \lambda_k), \overline{\mathbf{v}}(t, \lambda_k)\}$. 
The IDFT of any $\lambda_k$-eigenvector $\boldsymbol{\uppsi}_k$ of \eqref{eq:fc} is an eigenvector of~\eqref{eq:zs}. As any multiple of an eigenvector is also an eigenvector, we have
\begin{equation}
    \sum_{n=-N}^{N}\!\psi_{k,i}[n]e^{jn\frac{2\pi}{T}t}=G_k v_i(t, \lambda_k)%+\overline{G}_k \overline{v}_i(t, \lambda_k)
    \label{eq:eigenspace}
\end{equation}
for some constant $G_k$. 
%From~\eqref{eq:afromcont}, we have $a^*(\lambda^*)\to\infty$ as $\lambda\to\lambda_k$. Using~\eqref{eq:v_bar}, this implies that both $\overline{\mathbf{v}}(t, \lambda)$ and $\overline{\mathbf{x}}(t, \lambda)$ are unbounded. To ensure consistency of~\eqref{eq:eigenspace}, we assume $\overline{G}_k=0$ in~\eqref{eq:eigenspace}. 
Setting $v_{1}(-T/2,\lambda_k)=\exp(-j\lambda_k (-T/2))$ in~\eqref{eq:eigenspace} to satisfy~\eqref{eq:boundary}, we obtain}
\begin{equation}
G_k=e^{j\lambda_k \left(-\frac{T}{2}\right)}\sum_{n=-N}^{N}\psi_{k, 1}[n]e^{jn\frac{2\pi}{T}\left(-\frac{T}{2}\right)}.
\label{eq:G_k}
\end{equation}
{Note that \eqref{eq:boundary} imposes two boundary conditions. We observed in our simulations that the other condition
$v_2(-T/2, \lambda_k)\approx 0$ in~\eqref{eq:eigenspace} was always numerically satisfied, as expected.
}
% \del{
% We observed in our simulations that the eigenvalues of the discretized FC problem~\eqref{eq:fc} had one-dimensional eigenspaces and that the condition $v_2(-T/2, \lambda_k)\approx 0$ in~\eqref{eq:eigenspace} was always satisfied.}
{ From~\eqref{eq:ab_lim}, we have $b_k=v_{2}(T/2,\lambda_k)\exp(-j\lambda_k T/2)$. Substituting this in~\eqref{eq:eigenspace}, we have}
\begin{equation}
b_k=\frac{e^{-j\lambda_k \frac{T}{2}}}{G_k}\sum_{n=-N}^{N}\psi_{k,2}[n]e^{jn\frac{2\pi}{T}\frac{T}{2}}.
\label{eq:ab_fromev}
\end{equation}
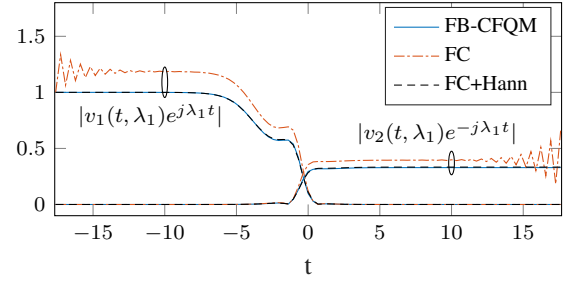
\begin{figure}[tbp]\centering
	\setlength{\figurewidth}{0.80\columnwidth}
	\setlength{\figureheight}{0.4\figurewidth}
	\input{Graphs/ringing_a.tex}
	\caption{Computation of $|v_1(t, \lambda_1)e^{j\lambda_1 t}|$ and $|v_2(t, \lambda_1)e^{-j\lambda_1 t}|$ for a $2$-soliton ($T=35.34$). Note the ringing artifact of FC at the edges where $G_k$~\eqref{eq:G_k} and $b_k$~\eqref{eq:ab_fromev} are calculated.}% $t=-T/2$ for $v_1(t, \lambda_1)e^{j\lambda_1 t}$ (where $G_k$ is calculated) and $t=T/2$ for $v_2(t, \lambda_1)e^{-j\lambda_1 t}$ for $t=T/2$ (where $b_k$ is calculated).} 
	\label{fig:ringing}
\end{figure}

The above quantities are based on the assumption that the periodic extensions of $v_i(t,\lambda_k)$ for $i\in\{1,2\}$ have a bandwidth smaller than $2N/T$. However, this assumption may not be satisfied if $v_i(-T/2,\lambda_k)\neq v_i(T/2,\lambda_k)$ {(note that, for $T\to\infty$ we have $v_i(-T/2,\lambda_k)= v_i(T/2,\lambda_k)\to 0$, but this holds only approximately for finite $T$)}.
This may cause undesirable ripples on $v_i(t,\lambda_k)$ around $t=\pm T/2$ when $v_i(t,\lambda_k)$ is obtained from \eqref{eq:fs}.
The ripples have a severe impact on the estimation of $G_k=v_1(-T/2, \lambda_k)e^{-j\lambda_k T/2}$ and $b_k=v_2(T/2,\lambda_k)e^{-j\lambda_k T/2}/G_k$. For example, consider a 2-soliton with $\lambda_1=0.6j,\lambda_2=0.3j$
and $b_1=b_2=\frac{1}{3}j$. The signal is truncated to $t\in[-17.67, 17.67]$ and sampled with
$M=103$. We computed
$v_1(t, \lambda_1)e^{j\lambda_1 t}$ and $v_2(t, \lambda_1)e^{-j\lambda_1 t}$ using FB-CFQM and FC.
The results in Fig.~\ref{fig:ringing}
show that the FC generates large ripples {causing further an incorrect estimate of $G_k$}.
%{and improperly normalizes the curves due to a wrong calculation of $G_k$}. 
We apply two techniques to mitigate the effect of ripples:
\begin{itemize}
	\item \textit{Frequency-domain windowing}: we apply windowing functions $w_1[n]$ and $w_2[n]$ respectively to $\boldsymbol{\uppsi}_{k, 1}$ and $\boldsymbol{\uppsi}_{k, 2}$, the two halves of $\boldsymbol{\uppsi}_k$.
	In our simulations, a Hann windowing function~\cite{blackman1958hann} gave promising results{, as shown by the curves labeled ``FC+Hann'' in Fig.~\ref{fig:ringing}}.
	
	\item \textit{Tail truncation}: we neglect part of the tails of $v_{1}(t, \lambda_k)$ and $v_{k}(t, \lambda_k)$ before computing $G_k$  and $b_k$. This means that we replace $T/2$ with $T_k<T/2$ in~\eqref{eq:G_k}, and~\eqref{eq:ab_fromev}.
\end{itemize}
{In Fig.~\ref{fig:ringing}, frequency-domain windowing seems to completely remove the ringing. For more complicated pulses, such as a $5$-soliton, tail truncation becomes also necessary}. Define the frequency-shifted windows ${u}_{k,1}[n]$ and ${u}_{k,2}[n]$ as
\begin{align}% n in the exponents are missing
    u_{k,1}[n] & =w_1[n]e^{-jn\frac{2\pi}{T}T_k} \\
    u_{k,2}[n] & =w_2[n]e^{jn\frac{2\pi}{T}T_k}.
\end{align}

Let $\mathbf{u}_{k,i}=(u_{k,i}[-N],\ldots,u_{k,i}[N])^{\mathrm{T}}$, with $i\in\{1,2\}$. Then, from~\eqref{eq:G_k} and~\eqref{eq:ab_fromev}, $b_k$ becomes
% \begin{align}
% G_k & =\sum_{n=-N}^{N}w_1[n]\psi_{k,1}[n]e^{j\left[n\frac{2\pi}{T}+\lambda_k\right](-T_k)}
% \label{eq:G_k_win} \\
% b_k & =\frac{1}{G_k}\sum_{n=-N}^{N}w_2[n]\psi_{k,2}[n]e^{j\left[n\frac{2\pi}{T}-\lambda_k\right]T_k}.
% \label{eq:ab_fromev_win}
% \end{align}
\begin{equation}
b_k=\frac{\sum_{n=-N}^{N}w_2[n]\psi_{k,2}[n]e^{jn\frac{2\pi}{T}T_k}}{\sum_{n=-N}^{N}w_1[n]\psi_{k,1}[n]e^{-jn\frac{2\pi}{T}T_k}}=\frac{\mathbf{u}_{k,2}^{\mathrm{T}}\boldsymbol{\uppsi}_{k,2}}{\mathbf{u}_{k,1}^{\mathrm{T}}\boldsymbol{\uppsi}_{k,1}}.
%{\defeq\frac{B_k}{A_k}}
\label{eq:ab_fromev_win}
\end{equation}
%where $B_k=\mathbf{u}_{k,2}^{\mathrm{T}}\boldsymbol{\uppsi}_{k,2}$ and $A_k=\mathbf{u}_{k,1}^{\mathrm{T}}\boldsymbol{\uppsi}_{k,1}$. 
The choice of $T_k$ should avoid the ripples in Fig.~\ref{fig:ringing} while staying close enough to the limit value. We choose $T_k=\min(t_{\exp}, t_{5\%})$, with:%
\begin{itemize}
	\item $t_{\exp}=12/\Im{\lambda_k}$ (such that $e^{\pm j\lambda_k t}$ does not become larger than $e^{12}\approx 1.63\cdot{10}^5$).
	\item $t_{5\%}=T/2-0.05T$ ($5\%$ removed tail). 
\end{itemize}
In our simulations, this heuristic choice works for signals that have most of their energy inside the interval $[-t_{\exp}, t_{\exp}]$, which seems to be the case for, at least, solitons of order up to $5$ with $e^{-\Im{\lambda_k}} <|b_k|< e^{\Im{\lambda_k}}$. This heuristic is based on the condition $|b_k|=1$ for symmetric solitons, which have good energy confinement in time~\cite{span2017symmetric}. For signals known to have high energy outside $[-t_{\exp}, t_{\exp}]$, the value of $t_{\exp}$ can be increased at the cost of accuracy in $b_k$. Note that windowing and truncation are done after obtaining $\lambda_k$ from~\eqref{eq:fc}. Therefore, these two techniques affect only the computation of $b_k$.

\section{Statistics of the Discrete Spectrum}\label{sec:perturbation}

In this section, we derive the second-order statistics 
of the discrete spectrum of a pulse contaminated by additive Gaussian noise.
For this purpose, we apply the well-established perturbation theory of linear operators to the FC method.

Consider again the signal $q(t)$ of pulse duration $T$ contaminated by zero-mean additive Gaussian noise $\sigma\tilde{q}(t)$ that is wide-sense stationary (WSS), i.e., $\mathbb{E}\left[\tilde{q}(t)\tilde{q}^*(t+\tau)\right]=r_q(\tau)$ does not depend on $t$. %with arbitrary power spectral density (PSD) $\sigma^2R_{\mathrm{cont}}(\Omega)$, where $\Omega$ is the angular frequency corresponding to continuous time. % with power spectral density $N_0$.
 Here, $\sigma^2$ is chosen such that $\mathbb{E}[|\tilde{q}(t)|^2]=1$. Let us first assume that $\sigma\tilde{q}(t)$ has bandwidth $\mathcal{B}$ and a constant power spectral density (PSD) $N_0$ inside the band. This implies $\sigma^2=N_0\mathcal{B}$. If the signal is sampled at the Nyquist rate, i.e., $\mathcal{B}=M/T$, then we have for $m\inset{-N}{N}$
\begin{equation}
\hat{q}[m]=q[m]+\sigma\tilde{q}[m]
\label{eq:q_pert}
\end{equation}
where $\sigma\tilde{q}[m]$ are the zero-mean noise samples. % with zero mean  and PSD $\sigma^2 R(\omega)=\sigma^2\frac{M}{T}R_{\mathrm{cont}}\left(\omega\frac{M}{T}\right)$. % and $E(|\tilde{q}[m]|^2)=N_0\frac{M}{T}$.
 The DFT coefficients of the noisy pulse are
\begin{align}
    \hat{c}[n]&=\frac{1}{M}\sum_{m=-N}^{N}q[m] e^{-j\frac{2\pi}{M}mn} + \frac{\sigma}{M}\sum_{m=-N}^{N}\tilde{q}[m] e^{-j\frac{2\pi}{M}mn} \nonumber \\
    &=c[n] + \sigma\tilde{c}[n]
\end{align}
where the choice of $\sigma$ implies that $\mathbb{E}[\sum_n|\tilde{c}[n]|^2]=1$. %The DFT coefficients $\sigma\tilde{c}[n]$ of $\sigma\tilde{q}[m]$ have a diagonal covariance matrix $\sigma^2\overline{\mathbf{R}}_{\tilde{\mathbf{c}}}$ whose diagonal entries are $\frac{\sigma^2}{T}R_{\mathrm{cont}}(n\frac{2\pi}{T})$ for $n\inset{-N}{N}$. 
Define 
$\tilde{\mathbf{c}}=\left(\tilde{c}[-N],\ldots,\tilde{c}[N]\right)^{\mathrm{T}}$.
%Let $\tilde{\mathbf{c}}$ be the vector that stacks the $\tilde{c}[n]$. 
The covariance matrix of the vector $(\Re\tilde{\mathbf{c}}^{\mathrm{T}}, \Im\tilde{\mathbf{c}}^{\mathrm{T}})^{\mathrm{T}}$  is
\begin{equation}
    \mathbf{R}_{\tilde{\mathbf{c}}}\triangleq\left[\begin{matrix}\mathbb{E}\left[\Re{\tilde{\mathbf{c}}}\Re{\tilde{\mathbf{c}}}^{\mathrm{T}}\right] & \mathbb{E}\left[\Re{\tilde{\mathbf{c}}}\Im{\tilde{\mathbf{c}}}^{\mathrm{T}}\right] \\ \mathbb{E}\left[\Im{\tilde{\mathbf{c}}}\Re{\tilde{\mathbf{c}}}^{\mathrm{T}}\right] & \mathbb{E}\left[\Im{\tilde{\mathbf{c}}}\Im{\tilde{\mathbf{c}}}^{\mathrm{T}}\right] \end{matrix}\right]=\frac{1}{2M}\mathbf{1}
    \label{eq:R_c}
\end{equation}
% \begin{equation}
%     \mathbf{R}_{\tilde{\mathbf{c}}}=\left[\begin{matrix}\mathbf{R}_{\tilde{\mathbf{c}},11} & \mathbf{R}_{\tilde{\mathbf{c}},12} \\ \mathbf{R}_{\tilde{\mathbf{c}},21} & \mathbf{R}_{\tilde{\mathbf{c}},22} \end{matrix}\right]=\left[\begin{matrix}\mathbb{E}\left[\Re{\tilde{\mathbf{c}}}\Re{\tilde{\mathbf{c}}}^{\mathrm{T}}\right] & \mathbb{E}\left[\Re{\tilde{\mathbf{c}}}\Im{\tilde{\mathbf{c}}}^{\mathrm{T}}\right] \\ \mathbb{E}\left[\Im{\tilde{\mathbf{c}}}\Re{\tilde{\mathbf{c}}}^{\mathrm{T}}\right] & \mathbb{E}\left[\Im{\tilde{\mathbf{c}}}\Im{\tilde{\mathbf{c}}}^{\mathrm{T}}\right] \end{matrix}\right]
%     \label{eq:R_c}
% \end{equation}
where $\mathbf{1}$ is the identity matrix. This is the case of additive white Gaussian noise (AWGN).
%This is the AWGN case. 
 %If the noise is oversampled ($\mathcal{B}<M/T$), we have $\sigma^2=N_0\mathcal{B}$, and $\mathbf{R}_{\tilde{\mathbf{c}}}=\mathrm{diag}([\mathbf{r}_{\tilde{\mathbf{c}}}, \mathbf{r}_{\tilde{\mathbf{c}}}])$, where $\mathbf{r}_{\tilde{\mathbf{c}}}$ is a vector whose first and last $(M-\mathcal{B}T)/2$ entries are $0$, and the rest are $1/(2\mathcal{B}T)$. 
In a general case with colored noise, $(\Re\tilde{\mathbf{c}}^{\mathrm{T}}, \Im\tilde{\mathbf{c}}^{\mathrm{T}})^{\mathrm{T}}=\mathbf{G}\mathbf{w}$, where $\mathbf{w}$ is a vector of $2M$ real-valued, i.i.d. Gaussian variables with variance $1/(2M)$ and $\mathbf{G}\inR{2M\times 2M}$ satisfies $\mathrm{tr}(\mathbf{G}\mathbf{G}^{\mathrm{T}})=2M$. Then we have
\begin{equation}
    \mathbf{R}_{\tilde{\mathbf{c}}}=\frac{1}{2M}\mathbf{G}\mathbf{G}^{\mathrm{T}}.
\end{equation}

% For circularly symmetric (CS) noise, we have $\mathbf{R}_{\tilde{\mathbf{c}},11}=\mathbf{R}_{\tilde{\mathbf{c}},22}=\frac{1}{2}\overline{\mathbf{R}}_{\tilde{\mathbf{c}}}$ and $\mathbf{R}_{\tilde{\mathbf{c}},11}=\mathbf{R}_{\tilde{\mathbf{c}},22}=\mathbf{0}$. For CS AWGN of PSD $N_0$, we have $\sigma^2=N_0\frac{M}{T}$ and $\mathbf{R}_{\tilde{\mathbf{c}}}=\frac{1}{2M}\mathbf{1}$, where $\mathbf{1}$ is an identity matrix. For band-limited CS noise with PSD $N_0$ inside the band $|\Omega|<\pi\mathcal{B}$, we have $\sigma^2=N_0\mathcal{B}$, and $\mathbf{R}_{\tilde{\mathbf{c}}}=\mathrm{diag}([\mathbf{r}_{\tilde{\mathbf{c}}}, \mathbf{r}_{\tilde{\mathbf{c}}}])$, where $\mathbf{r}_{\tilde{\mathbf{c}}}$ is a vector whose first and last $(M-\mathcal{B}T)/2$ entries are $0$, and the rest are $1/(2\mathcal{B}T)$.

%Since DFT is a unitary transform, $\tilde{c}[n]$ are independent and identically distributed (i.i.d) complex Gaussian random variables
%with zero mean and $\sigma^2=E(|\tilde{c}[n]|^2)=\frac{N_0}{T}$~\cite{tse2005wireless}. 
% Let $\sigma^2=\frac{1}{M}\mathrm{tr}(\mathbf{R}_{\tilde{\mathbf{c}}})$. We define the normalized noise $g[n]=\frac{1}{\sigma}\tilde{c}[n]$, such that $\mathbb{E}\left[\|\mathbf{g}\|^2\right]=M$.

\subsection{Perturbation of the Discrete Spectrum}

Let $\{(\hat{\lambda}_k, \hat{b}_k)\}$ denote the discrete spectrum of the noisy signal. Using the FC method, they are obtained from the solutions of the eigenvalue problem,
\begin{equation}
\hat{\mathbf{L}}\hat{\boldsymbol{\uppsi}}_k=\hat{\lambda}_k\hat{\boldsymbol{\uppsi}}_k
\label{eq:fc_pert}
\end{equation}
where
\begin{equation}
    \hat{\mathbf{L}}=\mathbf{L}+\sigma\tilde{\mathbf{L}},
\label{eq:L_decompose}
\end{equation}
where $\mathbf{L}$, given in~\eqref{eq:L}, corresponds to the noiseless pulse and 
\begin{equation}
\tilde{\mathbf{L}}=\left(\begin{matrix}
\mathbf{0} & \tilde{\boldsymbol{\Gamma}} \\
-\tilde{\boldsymbol{\Gamma}}^{\mathrm{H}} & \mathbf{0}
\end{matrix}\right)
\label{eq:L_prime}
\end{equation}
where $\tilde{\mathbf{\Gamma}}\inC{M\times M}$ is a Toeplitz matrix whose first column is $-j(\tilde{c}[0]\ \dots\ \tilde{c}[N]\ 0\ \dots\ 0)^{\mathrm{T}}$
%$-j(\tilde{c}_0\ \dots\ \tilde{c}_N\ 0\ \dots\ 0)^{\mathrm{T}}$ 
and whose first row is $-j(\tilde{c}[0]\ \dots\ \tilde{c}[-N]\ 0\ \dots\ 0)$.
%$-j(\tilde{c}_0\ \dots\ \tilde{c}_{-N}\ 0\ \dots\ 0)$. 
When $\sigma^2$ is relatively small, $(\hat{\lambda}_k, \hat{b}_k)$ can be approximated by first-order perturbations
\begin{align}
    \hat{\lambda}_k&=\lambda_k + \sigma\tilde{\lambda}_k + \mathcal{O}(\sigma^2)\label{eq:lam_per}\\
    \hat{\boldsymbol{\uppsi}}_k&=\boldsymbol{\uppsi}_k + \sigma\tilde{\boldsymbol{\uppsi}}_k+\mathcal{O}(\sigma^2)\label{eq:psi_per}\\
    \hat{b}_k &= b_k + \sigma \tilde{b}_k  + \mathcal{O}(\sigma^2). \label{eq:b_per}
\end{align}
The perturbation analysis of eigenvalues and eigenvectors is a mature topic. 
A detailed analysis is given in~\cite[Sections II.1 and II.2]{kato_eigenvalues}. 
We enclose some relevant first-order results as Theorem~\ref{thm:1}. 
We first define the left eigenvector $\boldsymbol{\upphi}_k$ as
\begin{equation}
\boldsymbol{\upphi}_k^{\mathrm{H}}\mathbf{L}=\lambda_k\boldsymbol{\upphi}_k^{\mathrm{H}}.
\label{eq:left}
\end{equation}
For our operator $\mathbf{L}$, $\boldsymbol{\upphi}_k$ is equal to the flipped, conjugated version of the right eigenvector $\boldsymbol{\uppsi}_k$, i.e.
\begin{equation}
\boldsymbol{\upphi}_k=\left(\psi_{k,2}^*[N],\dots,\psi_{k,2}^*[-N],\psi_{k,1}^*[N],\dots,\psi_{k,1}^*[-N]\right)^{\mathrm{T}}.
\label{eq:flipped}
\end{equation}
The reason is that $\boldsymbol{\upphi}_k$ corresponds to the (right) eigenvector of $\lambda_k^*$ for the signal $-q(t)$.
This can be seen by taking the conjugate transpose from both sides of \eqref{eq:left}. Moreover,
if $(\lambda_k, \left(v_{k,1}\ \ v_{k,2}\right)^{\mathrm{T}})$ is a solution of \eqref{eq:zs} for $q(t)$, then
$(\lambda_k^*,(v_{k,2}^*\ \ v_{k,1}^*)^{\mathrm{T}})$ is a solution of \eqref{eq:zs} for $-q(t)$. 
Combining these two properties concludes \eqref{eq:flipped}.

\begin{theorem}\label{thm:1} Consider \eqref{eq:L_decompose} with the condition that $\|\sigma\tilde{\mathbf{L}}\|_{\mathrm{F}}\ll\|\mathbf{L}\|_{\mathrm{F}}$ ($\|\cdot\|_{\mathrm{F}}$ is the Frobenius norm). We define for every eigenvalue and eigenvectors
$(\lambda_k,\boldsymbol{\uppsi}_k,\boldsymbol{\upphi}_k)$ of $\mathbf{L}$
\begin{align}
    \mathbf{P}_k&=\frac{1}{\boldsymbol{\upphi}_k^{\mathrm{H}}\boldsymbol{\uppsi}_k}\boldsymbol{\uppsi}_k \boldsymbol{\upphi}_k^{\mathrm{H}}\label{eq:eigproj_eig2}\\
\mathbf{S}_k&=\left(\mathbf{L}-\lambda_k\mathbf{1}-\mathbf{P}_k\right)^{-1}\left(\mathbf{1}-\mathbf{P}_k\right).
\label{eq:S}
\end{align}
The matrix $\mathbf{P}_k$ is called the \textit{eigenprojector} of $\lambda_k$
and $\mathbf{S}_k$ is the \textit{Drazin inverse}~\cite{rothblum_expansions} of  $\mathbf{L}-\lambda_k \mathbf{1}$.
If $\lambda_k$ has algebraic multiplicity 1,
then its first-order perturbations \eqref{eq:lam_per} and \eqref{eq:psi_per} are given by,
\begin{align}
    \label{eq:lambda_pert2}
\tilde{\lambda}_k &=\frac{1}{\boldsymbol{\upphi}_k^{\mathrm{H}}\boldsymbol{\uppsi}_k} \boldsymbol{\upphi}_k^{\mathrm{H}} \tilde{\mathbf{L}}\boldsymbol{\uppsi}_k\\
\tilde{\boldsymbol{\uppsi}}_k &=-\mathbf{S}_k\tilde{\mathbf{L}}\boldsymbol{\uppsi}_k.
\label{eq:psi_hat}
\end{align}
\end{theorem}
\begin{proof} Under the conditions that $\|\sigma\tilde{\mathbf{L}}\|_{\mathrm{F}}\ll\|\mathbf{L}\|_{\mathrm{F}}$ and $\lambda_k$ has multiplicity 1, it is shown
in \cite[Ch. II, Eqs. (2.21) and (2.33)]{kato_eigenvalues}
\begin{equation}
\tilde{\lambda}_k=\mathrm{tr}\left(\tilde{\mathbf{L}}\mathbf{P}_k\right)
\label{eq:lambda_pert}
\end{equation}
and $\mathbf{P}_k$ is given based on a contour integral~\cite[Ch. I, Eq. (5.22)]{kato_eigenvalues} and further simplified in \cite{rothblum1976eigenprojector} to \eqref{eq:eigproj_eig2}. Substituting~\eqref{eq:eigproj_eig2} in~\eqref{eq:lambda_pert} and using $\mathrm{tr}(\tilde{\mathbf{L}} \boldsymbol{\uppsi}_k \boldsymbol{\upphi}_k^{\mathrm{H}} )=\mathrm{tr}(\boldsymbol{\upphi}_k^{\mathrm{H}} \tilde{\mathbf{L}} \boldsymbol{\uppsi}_k)$ results in \eqref{eq:lambda_pert2}. The proof of \eqref{eq:psi_hat} is also given in~\cite[Ch. II, Eq. (4.23)]{kato_eigenvalues}.
\end{proof}
{\begin{remark}
We assumed that all $\lambda_k$ have multiplicity 1, which is generally assumed when considering NFT for communication. Eigenvalues with higher multiplicities are possible~\cite{olmedilla1987multiple} and could be used for communication~\cite{garcia2018communication}. Noise would make these eigenvalues split~\cite{kato_eigenvalues}, complicating the analysis.
\end{remark}
} 

%The perturbation  needs some further steps. 
{To compute the perturbation $\tilde{b}_k$, we apply a Taylor expansion around $\sigma=0$ to the perturbed version of~\eqref{eq:ab_fromev_win}:}
\begin{align}
    \hat{b}_k&=\frac{\mathbf{u}_{k,2}^{\mathrm{T}}\boldsymbol{\uppsi}_{k,2} +\sigma \mathbf{u}_{k,2}^{\mathrm{T}}\tilde{\boldsymbol{\uppsi}}_{k,2}+\mathcal{O}(\sigma^2) }{\mathbf{u}_{k,1}^{\mathrm{T}}\boldsymbol{\uppsi}_{k,1}+\sigma\mathbf{u}_{k,1}^{\mathrm{T}}\tilde{\boldsymbol{\uppsi}}_{k,1}+\mathcal{O}(\sigma^2)}\\
    &= \frac{\mathbf{u}_{k,2}^{\mathrm{T}}\boldsymbol{\uppsi}_{k,2}}{\mathbf{u}_{k,1}^{\mathrm{T}}\boldsymbol{\uppsi}_{k,1}}
    \left(1 + 
    \sigma \frac{\mathbf{u}_{k,2}^{\mathrm{T}}\tilde{\boldsymbol{\uppsi}}_{k,2}}{\mathbf{u}_{k,2}^{\mathrm{T}}\boldsymbol{\uppsi}_{k,2}}-\sigma
    \frac{\mathbf{u}_{k,1}^{\mathrm{T}}\tilde{\boldsymbol{\uppsi}}_{k,1}}{\mathbf{u}_{k,1}^{\mathrm{T}}\boldsymbol{\uppsi}_{k,1}}\right) + \mathcal{O}(\sigma^2).\nonumber
\end{align}
Using~\eqref{eq:b_per}, $\tilde{b}_k$ is given in terms of $b_k$ and $\mathbf{u}_{k,1}^{\mathrm{T}}\boldsymbol{\uppsi}_{k,1}$ as
\begin{equation}
    \tilde{b}_k = \frac{1}{\mathbf{u}_{k,1}^{\mathrm{T}}\boldsymbol{\uppsi}_{k,1}}\left(
    \mathbf{u}_{k,2}^{\mathrm{T}}\tilde{\boldsymbol{\uppsi}}_{k,2}-
    b_k\mathbf{u}_{k,1}^{\mathrm{T}}\tilde{\boldsymbol{\uppsi}}_{k,1}
    \right)
    \label{eq:b_k_tilde}
\end{equation}
or equivalently,
\begin{equation}
    \tilde{b}_k = \frac{1}{\mathbf{u}_{k,1}^{\mathrm{T}}\boldsymbol{\uppsi}_{k,1}}
    \left(b_k\mathbf{u}_{k,1}^{\mathrm{T}}, -\mathbf{u}_{k,2}^{\mathrm{T}}\right) \mathbf{S}_k\tilde{\mathbf{L}}\boldsymbol{\uppsi}_k.
    \label{eq:b_k_tilde2}
\end{equation}
Both $\tilde{\lambda}_k$ and $\tilde{b}_k$ depend on $\tilde{\mathbf{L}}\boldsymbol{\uppsi}_k$.
From \eqref{eq:L_prime}
we can write 
\begin{equation}
    \tilde{\mathbf{L}}\boldsymbol{\uppsi}_k=
    \left(\begin{matrix}
    \tilde{\mathbf{\Gamma}}\boldsymbol{\uppsi}_{k, 2}\\
    -\tilde{\mathbf{\Gamma}}^{\mathrm{H}}\boldsymbol{\uppsi}_{k, 1}
    \end{matrix}\right).
    \label{eq:Lpsi_0}
\end{equation}
The first vector $\tilde{\mathbf{\Gamma}}\boldsymbol{\uppsi}_{k, 2}$
is 
the convolution of $-j\tilde{c}[n]$ and $\uppsi_{k, 2}[n]$. 
From the commutative property of convolution, i.e. $-j\tilde{c}[n]*\uppsi_{k, 2}[n]=-j
\uppsi_{k, 2}[n]*\tilde{c}[n]$, it can be reordered as
\begin{equation}
\tilde{\mathbf{\Gamma}}\boldsymbol{\uppsi}_{k, 2} = -j\mathbf{J}_{k,2}\tilde{\mathbf{c}}
\end{equation}
where $\mathbf{J}_{k,2}\inC{M\times M}$ is a Toeplitz matrix whose first column is $(\psi_{k,2}[0]\ \dots\ \psi_{k,2}[N]\ 0\ \dots\ 0)^{\mathrm{T}}$ and whose first row is $(\psi_{k,2}[0]\ \dots\ \psi_{k,2}[-N]\ 0\ \dots\ 0)$. Similarly,
\begin{equation}
-\tilde{\mathbf{\Gamma}}^{\mathrm{H}}\boldsymbol{\uppsi}_{k, 1} = -j\mathbf{J}_{k,1}\boldsymbol{\Pi}\tilde{\mathbf{c}}^*
\end{equation}
where $\mathbf{J}_{k,1}$ is defined similarly to $\mathbf{J}_{k,2}$, and $\boldsymbol{\Pi}$ is an order-reversing matrix, i.e., an anti-diagonal matrix with the anti-diagonal elements equal to one. Equation~\eqref{eq:Lpsi_0} becomes
\begin{equation}
     \tilde{\mathbf{L}}\boldsymbol{\uppsi}_k=-j
    \left(\begin{matrix}
    \mathbf{J}_{k,2}\tilde{\mathbf{c}}\\
    \mathbf{J}_{k,1}\boldsymbol{\Pi}\tilde{\mathbf{c}}^*
    \end{matrix}\right)=\boldsymbol{\Sigma}_k
    \left(\begin{matrix}
    \Re{\tilde{\mathbf{c}}}\\
    \Im{\tilde{\mathbf{c}}}
    \end{matrix}\right)
    \label{eq:Lpsi}
\end{equation}
with
\begin{equation}
    \boldsymbol{\Sigma}_k=\left(\begin{matrix}
    -j\mathbf{J}_{k,2}&& \mathbf{J}_{k,2}\\
    -j\mathbf{J}_{k,1}\boldsymbol{\Pi}&& -\mathbf{J}_{k,1}\boldsymbol{\Pi}
    \end{matrix}\right).
\end{equation}

\subsection{Statistics of $\tilde{\lambda}_k$}\label{sec:stat_lambda}

Consider the perturbation term $\tilde{\lambda}_k$ in \eqref{eq:lambda_pert2}. Using \eqref{eq:Lpsi},
\begin{equation}
\tilde{\lambda}_k =\frac{1}{\boldsymbol{\upphi}_k^{\mathrm{H}}\boldsymbol{\uppsi}_k} 
\boldsymbol{\upphi}_k^{\mathrm{H}} \boldsymbol{\Sigma}_k \left(\begin{matrix}
    \Re{\tilde{\mathbf{c}}}\\
    \Im{\tilde{\mathbf{c}}}
    \end{matrix}\right)=\mathbf{d}_{k}\left(\begin{matrix}
    \Re{\tilde{\mathbf{c}}}\\
    \Im{\tilde{\mathbf{c}}}
    \end{matrix}\right)
\end{equation}
{where $\mathbf{d_k}$ is a horizontal vector defined as}
\begin{equation}
    \mathbf{d}_{k} = \frac{1}{\boldsymbol{\upphi}_k^{\mathrm{H}}\boldsymbol{\uppsi}_k} \boldsymbol{\upphi}_k^{\mathrm{H}} \boldsymbol{\Sigma}_k =
    \frac{1}{\boldsymbol{\uppsi}_k^{\mathrm{T}}\mathbf{\Pi}\boldsymbol{\uppsi}_k} \boldsymbol{\uppsi}_k^{\mathrm{T}}\mathbf{\Pi} \boldsymbol{\Sigma}_k.
\end{equation}
% depends only on the elements of ${\uppsi}_{k}$. Moreover,
% \begin{equation}
% \label{eq:normalization}
% \boldsymbol{\upphi}_k^{\mathrm{H}}\boldsymbol{\uppsi}_k=2\sum_{n=-N}^N \psi_{k,1}[ n]\psi_{k,2}[ -n].
% \end{equation}
{The second equality uses $\boldsymbol{\upphi}_k^{\mathrm{H}}=\boldsymbol{\uppsi}_{k}^{\mathrm{T}}\mathbf{\Pi}$ (see~\eqref{eq:flipped})}. Let $\tilde{\boldsymbol{\lambda}}=(\tilde{\lambda}_1,\dots,\tilde{\lambda}_K)^{\mathrm{T}}$, and let $\mathbf{D}$ be a matrix whose rows are the $\mathbf{d}_{k}$. Then we have $\tilde{\boldsymbol{\lambda}} = \mathbf{D} \left(\begin{matrix}
    \Re{\tilde{\mathbf{c}}}\\
    \Im{\tilde{\mathbf{c}}}
    \end{matrix}\right)$, or
% \begin{equation}
%     \tilde{\boldsymbol{\lambda}} = \mathbf{D} \left(\begin{matrix}
%     \Re{\tilde{\mathbf{c}}}\\
%     \Im{\tilde{\mathbf{c}}}
%     \end{matrix}\right)
% \end{equation}
% where . Equivalently,
\begin{equation}
    \left(\begin{matrix}\Re{\tilde{\boldsymbol{\lambda}}}\\ \Im{\tilde{\boldsymbol{\lambda}}}\end{matrix}\right)=
    \left(\begin{matrix}\Re{\mathbf{D}}\\ \Im{\mathbf{D}}\end{matrix}\right)
\left(\begin{matrix}
\Re{\tilde{\mathbf{c}}} \\ \Im{\tilde{\mathbf{c}}}
\end{matrix}\right).
\label{eq:l_pert_linear}
\end{equation}
{The matrix $\mathbf{D}$ contains the normalized autocorrelation functions of the FC eigenvectors $\boldsymbol{\uppsi}_k$ (see~\cite{garcia2018statistics} for an alternative formulation). In time domain, these autocorrelation functions become the squared Jost solutions, which is in accordance with Eq. (47) of~\cite{prati2018some}. The dependence of the perturbation of the eigenvalues on the squared Jost solutions is well known~\cite{kaup1976perturbation}.}

%From~\eqref{eq:l_pert_linear}, $\Re{\tilde{\boldsymbol{\lambda}}}$ and $\Im{\tilde{\boldsymbol{\lambda}}}$ are obtained from linear combinations of
%$\Re{\tilde{\mathbf{c}}}$ and $\Im{\tilde{\mathbf{c}}}$.
From~\eqref{eq:l_pert_linear}, since $(\Re{\tilde{\mathbf{c}}}^{\mathrm{T}},\Im{\tilde{\mathbf{c}}}^{\mathrm{T}})^{\mathrm{T}}$ has a jointly Gaussian distribution,
$(\Re{\tilde{\boldsymbol{\lambda}}}^{\mathrm{T}},\Im{\tilde{\boldsymbol{\lambda}}}^{\mathrm{T}})^{\mathrm{T}}$ has a jointly Gaussian distribution
with zero mean and covariance matrix,
\begin{equation}
\mathbf{C}_{\tilde{\boldsymbol{\lambda}}}= \left(\begin{matrix}\Re{\mathbf{D}}\\ \Im{\mathbf{D}}\end{matrix}\right)
\mathbf{R}_{\tilde{\mathbf{c}}}
\left(\Re{\mathbf{D}^{\mathrm{T}}}\\ \Im{\mathbf{D}^{\mathrm{T}}}\right).
\label{eq:C_ll}
\end{equation}

\subsection{Statistics of $\tilde{b}_k$ coefficients}

Like $\tilde{\lambda}_k$, $\tilde{b}_k$ is a linear combination of $\Re{\tilde{c}[n]}$ and $\Im{\tilde{c}[n]}$. To see this more clearly,
define
\begin{equation*}
    \mathbf{h}_{k}=\frac{1}{\mathbf{u}_{k,1}^{\mathrm{T}}\boldsymbol{\uppsi}_{k,1}}
    \left(b_k\mathbf{u}_{k,1}^{\mathrm{T}}, -\mathbf{u}_{k,2}^{\mathrm{T}}\right) \mathbf{S}_k\boldsymbol{\Sigma}_k.
\end{equation*}
Using \eqref{eq:b_k_tilde2}, we have simply,
\begin{equation}
    \tilde{b}_k = \mathbf{h}_{k} \left(\begin{matrix}
    \Re{\tilde{\mathbf{c}}}\\
    \Im{\tilde{\mathbf{c}}}
    \end{matrix}\right).
\end{equation}
Let $\tilde{\mathbf{b}}=(\tilde{b}_1,\ldots,\tilde{b}_K)^{\mathrm{T}}$, and let $\mathbf{H}$ be a matrix whose rows are the $\mathbf{h}_{k}$ Then we have $\tilde{\mathbf{b}} = \mathbf{H}\left(\begin{matrix}
\Re{\tilde{\mathbf{c}}} \\ \Im{\tilde{\mathbf{c}}}
\end{matrix}\right)$, or
% \begin{equation}
%     \tilde{\mathbf{b}} = \mathbf{H}\left(\begin{matrix}
% \Re{\tilde{\mathbf{c}}} \\ \Im{\tilde{\mathbf{c}}}
% \end{matrix}\right)
% \end{equation}
% where  Equivalently,
\begin{equation}
    \left(\begin{matrix}\Re{\tilde{\mathbf{b}}}\\ \Im{\tilde{\mathbf{b}}}\end{matrix}\right)=
    \left(\begin{matrix}\Re{\mathbf{H}}\\ \Im{\mathbf{H}}\end{matrix}\right)
\left(\begin{matrix}
\Re{\tilde{\mathbf{c}}} \\ \Im{\tilde{\mathbf{c}}}
\end{matrix}\right).
\label{eq:b_pert_linear}
\end{equation}
Thus, $(\Re{\tilde{\mathbf{b}}}^{\mathrm{T}},\Im{\tilde{\mathbf{b}}}^{\mathrm{T}})^{\mathrm{T}}$ has a jointly Gaussian distribution
with zero mean and covariance matrix,
\begin{equation}
\mathbf{C}_{\tilde{\mathbf{b}}}=
\left(\begin{matrix}\Re{\mathbf{H}}\\ \Im{\mathbf{H}}\end{matrix}\right)
\mathbf{R}_{\tilde{\mathbf{c}}}\left(\Re{\mathbf{H}^{\mathrm{T}}},\Im{\mathbf{H}^{\mathrm{T}}}\right).
\label{eq:C_bb}
\end{equation}
% \begin{corollary}\label{thm:C_bb}
% Up to the first-order approximation in $\sigma$, the spectral coefficients
% $\hat{b}_k$ have a joint Gaussian distribution with covariance matrix $\sigma^2\mathbf{C}_{\tilde{\mathbf{b}}}$. The mean of each $\hat{b}_k$ is $b_k$. 
% \end{corollary}

{\subsection{Cross-statistics of $\tilde{\lambda}_k$ and $\tilde{b}_k$}
From~\eqref{eq:l_pert_linear} and~\eqref{eq:b_pert_linear}, we obtain the cross-covariance matrix $\mathbf{C}_{\tilde{\boldsymbol{\lambda}}\tilde{\mathbf{b}}}\triangleq\mathbb{E}[(\Re{\tilde{\boldsymbol{\lambda}}}^{\mathrm{T}},\Im{\tilde{\boldsymbol{\lambda}}}^{\mathrm{T}})^{\mathrm{T}}(\Re{\tilde{\mathbf{b}}}^{\mathrm{T}},\Im{\tilde{\mathbf{b}}}^{\mathrm{T}})]$:
\begin{equation}
    \mathbf{C}_{\tilde{\boldsymbol{\lambda}}\tilde{\mathbf{b}}}=\left(\begin{matrix}\Re{\mathbf{D}}\\ \Im{\mathbf{D}}\end{matrix}\right)
\mathbf{R}_{\tilde{\mathbf{c}}}\left(\Re{\mathbf{H}^{\mathrm{T}}},\Im{\mathbf{H}^{\mathrm{T}}}\right).
\label{eq:C_lb}
\end{equation}\begin{corollary*} \label{thm:C_ll}
Up to the first-order approximation in $\sigma$,
the eigenvalues $\hat{\lambda}_k$ of the noisy signal have a joint Gaussian distribution with means $\lambda_k$ and covariance matrix $\sigma^2\mathbf{C}_{\tilde{\boldsymbol{\lambda}}}$. The spectral coefficients
$\hat{b}_k$ have a joint Gaussian distribution with means $b_k$ covariance matrix $\sigma^2\mathbf{C}_{\tilde{\mathbf{b}}}$. The cross-covariances between the $\hat{\lambda}_k$ and the $\hat{b}_k$ are $\sigma^2 \mathbf{C}_{\tilde{\boldsymbol{\lambda}}\tilde{\mathbf{b}}}$.
\end{corollary*}
}

\section{Numerical validation}\label{sec:numerical}

\begin{figure*}[!t]\centering
	\setlength{\figurewidth}{0.29\textwidth}
	\setlength{\figureheight}{0.5\figurewidth}
	\input{Graphs/pert_2sol_l.tex}
	\input{Graphs/pert_5sol_l.tex}
	\input{Graphs/pert_sech_l.tex}
	\input{Graphs/pert_2sol_b.tex}
	\input{Graphs/pert_5sol_b.tex}
	\input{Graphs/pert_sech_b.tex}
	\caption{Error in the computation of the discrete spectrum of different pulses. First column: error in $\lambda_k$ (a) and $b_k$ (d) for a $2$-soliton. Second column: errors for a $5$-soliton, with legend: \protect\tikz\protect\draw [only marks, mark=*, color=red, mark options={scale=0.5}] plot (0, 0); \protect\tikz\protect\draw [only marks, mark=o, color=red] plot (0, 0); ($\lambda_1$, $b_1$), \protect\tikz\protect\draw [only marks, mark=square*, color=mycolor1, mark options={scale=0.5}] plot (0, 0); \protect\tikz\protect\draw [only marks, mark=square, color=mycolor1] plot (0, 0); ($\lambda_2$, $b_2$), \protect\tikz\protect\draw [only marks, mark=triangle*, color=mycolor2, mark options={rotate=90, solid, scale=0.5, mycolor2}] plot (0, 0); \protect\tikz\protect\draw [only marks, mark=triangle, color=mycolor2, mark options={rotate=90, solid, mycolor2}] plot (0, 0); ($\lambda_3$, $b_3$), \protect\tikz\protect\draw [only marks, mark=triangle*, color=mycolor3, mark options={rotate=180, scale=0.5, solid, mycolor3}] plot (0, 0); \protect\tikz\protect\draw [only marks, mark=triangle, color=mycolor3, mark options={rotate=180, solid, mycolor3}] plot (0, 0); ($\lambda_4$, $b_4$), \protect\tikz\protect\draw [only marks, mark=pentagon*, color=mycolor4, mark options={scale=0.5}] plot (0, 0); \protect\tikz\protect\draw [only marks, mark=pentagon, color=mycolor4] plot (0, 0); ($\lambda_5$, $b_5$). Third column: errors for $q(t)=2.2\mathrm{sech}(t)$.}
	\label{fig:pert}
\end{figure*}

\subsection{Accuracy of FC in the noiseless case}

We measured the accuracy of the computation of $\{\lambda_k,b_k\}$
%$\lambda_k$ and $b(\lambda_k)$ 
using~\eqref{eq:fc} and~\eqref{eq:ab_fromev_win} for three different pulses:
\begin{itemize}
\item a 2-soliton with $\lambda_k=\left[0.6j, 0.3j\right]$ and $b_k=\left[\frac{1}{3}j, \frac{1}{3}j\right]$
\item a 5-soliton with $\lambda_k=\left[1.5j, 1.2j, 0.9j, 0.6j, 0.3j\right]$ and, respectively, $b_k=[0.8855 + 0.1109j,  -1.4293 - 0.6778j,   1.0701 + 0.2486j, -0.0965 + 1.0385j,   0.3345 + 0.8551j]$
\item the pulse 
\begin{equation}
q(t)=2.2\mathrm{sech}(t)
\end{equation}
which has known nonzero continuous and discrete spectra~\cite{satsuma2018sech}. The discrete spectrum is $\lambda_k=[1.7, 0.7]$ and $b_k=[-1, 1]$.
\end{itemize}
We used the Darboux algorithm {with double precision floating-point accuracy} to generate the solitons, and then computed the discrete spectrum of the pulses using FB-CFQM and FC. For our FC method, we used {tail truncation and} Hann windows $w_1[n]$ and $w_2[n]$ (see~\eqref{eq:ab_fromev_win}). %\del{As a benchmark, we use an NFT algorithm that combines the forward-backward iterations in [16] with a sixth-order Commutator-Free Quasi-Magnus (CFQM) integrator $CF_4^{[6]}$, proposed in~\cite{chimmalgi2018cfqm}. We replace the second-order trapezoidal integration of~\cite{aref_control_detection} with CFQM to improve the performance further. We denote this benchmark algorithm by FB-CFQM.}
We used a time interval with fixed length $T$, and varied the number of time samples by varying the sampling period. We used $T=35.3$ for the 2-soliton and the 5-soliton, and $T=24.0$ for the sech pulse. We compared the results of FB-CFQM and FC in terms of error:
\begin{equation}
\mathrm{err}(\lambda_k, \hat{\lambda}_k)\triangleq\left(|\hat{\lambda}_k-\lambda_k|\right)/\left|\lambda_k\right|
\end{equation}
where $\hat{\lambda}_k$ is the eigenvalue obtained using FB-CFQM or FC, and $\lambda_k$ is the original eigenvalue. $\mathrm{err}(b_k, \hat{b}_k)$ is defined similarly.
The results in Fig.~\ref{fig:pert} show that
Our FC method
provides rather precise
estimates of $b_k$
(especially for $2$-soliton and sech pulse, resp. Fig.~\ref{fig:pert} (d) and (f)). However,
FB-CFQM has overall more precise estimates, especially when the number of samples is large.

%\begin{figure}[tbp]\centering
%	\setlength{\figurewidth}{0.7\columnwidth}
%	\setlength{\figureheight}{0.618\figurewidth}
%	\input{Graphs/pert_2sol_l.tex}
%	\input{Graphs/pert_2sol_b.tex}
%	\caption{Error in the computation of the discrete spectrum of a 2-soliton}
%	\label{fig:pert_2sol}
%\end{figure}
%\begin{figure}[tbp]\centering
%\setlength{\figurewidth}{0.35\columnwidth}
%\setlength{\figureheight}{0.618\figurewidth}
%\input{Graphs/pert_5sol_l.tex}
%\input{Graphs/pert_5sol_b.tex}
%\caption{Error in the computation of the discrete spectrum of a 5-soliton}
%\label{fig:pert_5sol}
%\end{figure}

\subsection{Accuracy of the FC perturbation analysis in the noisy case}
To validate our proposed closed-form expressions for the covariance matrix of the eigenvalues~\eqref{eq:C_ll} and spectral coefficients~\eqref{eq:C_bb}, we simulated the same three pulses of the previous section, using $M=365$ samples for the 2-soliton, $M=909$ for the 5-soliton, and $M=329$ for the sech pulse. {Using higher values of $M$ did not have any significant effect on the resulting covariances}. We then added AWGN $\sigma\tilde{q}[m]$ to the three pulses. We define the signal-to-noise ratio (SNR) as:
{\begin{equation}
\mathrm{SNR}=\frac{1}{M\sigma^2}\sum_{m=-N}^{N}\left|q[m]\right|^2.
\end{equation}}
For each value of SNR, we generated $1024$ realizations of noise and computed the discrete spectrum of the noisy pulse using the FC and FB-CFQM. From this data, we estimated the covariance matrix of $\lambda_k$ and of $b_k$ for the FC and FB-CFQM case, and compared it with our analytic covariance matrices given by~\eqref{eq:C_ll} and~\eqref{eq:C_bb}. {For clarity of presentation,}
we only plot the following entries
\begin{align}
&\sigma_{\lambda_k}^2 \triangleq\sigma^2\mathrm{E}\left[(\Re\tilde{\lambda}_k)^2\right]+\sigma^2\mathrm{E}\left[(\Im\tilde{\lambda}_k)^2\right] \\
&\left|\Re{\sigma_{\lambda_1\lambda_2}}\right|\triangleq\sigma^2\left|\mathrm{E}\left[\Re{\tilde{\lambda}_1}\Re{\tilde{\lambda}_2}+\Im{\tilde{\lambda}_1}\Im{\tilde{\lambda}_2}\right]\right|
\end{align}
%\begin{itemize}
%	\item For all pulses: $\sigma_{\lambda_k}^2=\mathrm{E}\left[\Real{\tilde{\lambda}_k}^2\right]+\mathrm{E}\left[\Imag{\tilde{\lambda}_k}^2\right]$
%	\item For the 2-soliton and the sech pulse: $\left|\Real{\sigma_{\lambda_1\lambda_2}}\right|=\left|\mathrm{E}\left[\Real{\tilde{\lambda}_1}\Real{\tilde{\lambda}_2}+\Imag{\tilde{\lambda}_1}\Imag{\tilde{\lambda}_2}\right]\right|$
%\end{itemize}
{Similar quantities are defined and plotted for $b_k$}. {We also compare the variances $\sigma_{\lambda_k}^2$ of our method with
the ones obtained from the time-domain method in ~\cite{mansoor_parts}. Both methods obtain fairly the same $\sigma_{\lambda_k}^2$
as shown in Fig.~\ref{fig:noise}.  } 
%{For the variances $\sigma_{\lambda_k}^2$, we also compare to the analytic variances obtained using
%Yousefi's method {~\cite{mansoor_parts}. The results in Fig.~\ref{fig:noise} show that Yousefi's expressions coincide with ours}. 
For the 2-soliton, we set the SNR to $10$ dB and used $b_k=[\frac{1}{3}j,\frac{1}{3}j\exp(j\alpha)]$. In the first column of Fig.~\ref{fig:noise}, we plot $\sigma_{\lambda_k}^2$ and $\left|\Re{\sigma_{\lambda_1\lambda_2}}\right|$ as a function of the phase difference $\alpha$ between the spectral coefficients. As we already reported in~\cite{garcia2018statistics}, we see that the eigenvalues of 2-solitons where $b_1$ and $b_2$ have equal phase (or equivalently, $Q_d(\lambda_1)$ and $Q_d(\lambda_2)$ have \textit{opposite} phase) are less robust to noise.

In the second and third columns of Fig.~\ref{fig:noise}, we vary the SNR and plot $\sigma_{\lambda_k}^2$ for the 5-soliton, and both $\sigma_{\lambda_k}^2$ and $\left|\Re{\sigma_{\lambda_1\lambda_2}}\right|$ for the sech pulse. Our analytic expressions for the covariances (solid lines) very accurately predict the numerical covariances obtained using FC and FB-CFQM. {For very low SNR, the computation of the Jost solutions using FC becomes numerically unstable, and the ringing in Fig.~\ref{fig:ringing} increases considerably. This makes the estimation of $b_1$ with FC challenging, especially for the sech pulse. The analytic covariance matrix is computed from the noiseless pulse and is thus not affected by this}. Our analytic covariances are in all cases very close to the lowest numerical ones (either FB-CFQM or FC).

{Fig.~\ref{fig:C_big} shows the full covariance matrix obtained with FC
\begin{equation}
\sigma^2\mathbf{C}\triangleq\sigma^2\left(\begin{matrix}\mathbf{C}_{\tilde{\mathbf{\lambda}}} & \mathbf{C}_{\tilde{\mathbf{\lambda}}\tilde{\mathbf{b}}} \\ \mathbf{C}_{\tilde{\mathbf{\lambda}}\tilde{\mathbf{b}}}^{\mathrm{T}} & \mathbf{C}_{\tilde{\mathbf{b}}} \end{matrix}\right)
\end{equation}
for the $2$-soliton with SNR=$10$ dB and $\alpha=\pi$ (i.e., $b_k=[\frac{1}{3}j,-\frac{1}{3}j]$. The covariance matrices obtained with FB-CFQM and with our analytic formulas \eqref{eq:C_ll}, \eqref{eq:C_bb}, \eqref{eq:C_lb} were very close to the FC result. Namely, the errors between them were
\begin{subequations}
\begin{align}
    & \|\mathbf{C}^{(\text{FC})}-\mathbf{C}^{(\text{analytic})}\|_{\mathrm{F}}^2/\|\mathbf{C}^{(\text{FC})}\|_{\mathrm{F}}^2  \approx 0.0035 \\
    & \|\mathbf{C}^{(\text{FC})}-\mathbf{C}^{(\text{FB-CFQM})}\|_{\mathrm{F}}^2/\|\mathbf{C}^{(\text{FC})}\|_{\mathrm{F}}^2 \approx 0.0039 \\
    & \|\mathbf{C}^{(\text{FB-CFQM})}-\mathbf{C}^{(\text{analytic})}\|_{\mathrm{F}}^2/\|\mathbf{C}^{(\text{analytic})}\|_{\mathrm{F}}^2 \approx 0.0032.
\end{align}
\end{subequations}
We observe in Fig.~\ref{fig:C_big} that all real parts $(\Re \lambda_k, \Re b_k)$ are correlated with each other, and so are the imaginary parts $(\Im \lambda_k, \Im b_k)$, but there is little correlation between any real part and any imaginary part. Similarly to $1$-solitons, the covariances of the $\Re\lambda_k$ are smaller than those of the $\Im\lambda_k$.

}

\begin{figure*}[!t]\centering
	\setlength{\figurewidth}{0.29\textwidth}
	\setlength{\figureheight}{0.5\figurewidth}
	\input{Graphs/noise_2sol_phi_l.tex}
	\input{Graphs/noise_5sol_l.tex}
	\input{Graphs/noise_sech_l.tex}
	\input{Graphs/noise_2sol_phi_b.tex}
	\input{Graphs/noise_5sol_b.tex}
	\input{Graphs/noise_sech_b.tex}
	\begin{tikzpicture}
	\begin{axis}[
	width=0.8\textwidth,
	hide axis,
	xmin=10,
	xmax=50,
	ymin=0,
	ymax=0.4,
	legend columns=4,
	legend style={draw=white!15!black,legend cell align=left, /tikz/every even column/.append style={column sep=1em}},
	]
	\addlegendimage{only marks, semithick, gray, mark=diamond*, mark options={solid, scale=0.5, gray, fill=gray}}
	\addlegendentry{(filled markers) FB-CFQM}
	\addlegendimage{only marks, gray, mark=diamond, mark options={solid, gray}}
	\addlegendentry{(empty markers) FC}
	\addlegendimage{solid, gray}
	\addlegendentry{Analytic (this work)}
	\addlegendimage{densely dotted, thick, gray}
	\addlegendentry{Analytic~\cite{mansoor_parts}}
	\end{axis}
	\end{tikzpicture}
	\caption{Statistics of the discrete spectrum of different noisy pulses. The first column gives the analytic and numerical second-order moments of $\lambda_k$ (a) and $b_k$ (d) of a noisy 2-soliton with SNR=10 dB, as a function of the phase difference $\alpha$ between the $b_1$ and $b_2$. The second column gives the moments of $\lambda_k$ (b) and $b_k$ (e) of a $5$-soliton as a function of SNR. The third column (c) and (f) gives the moments for $q(t)=2.2\mathrm{sech}(t)$ as a function of SNR.}
	\label{fig:noise}
\end{figure*}
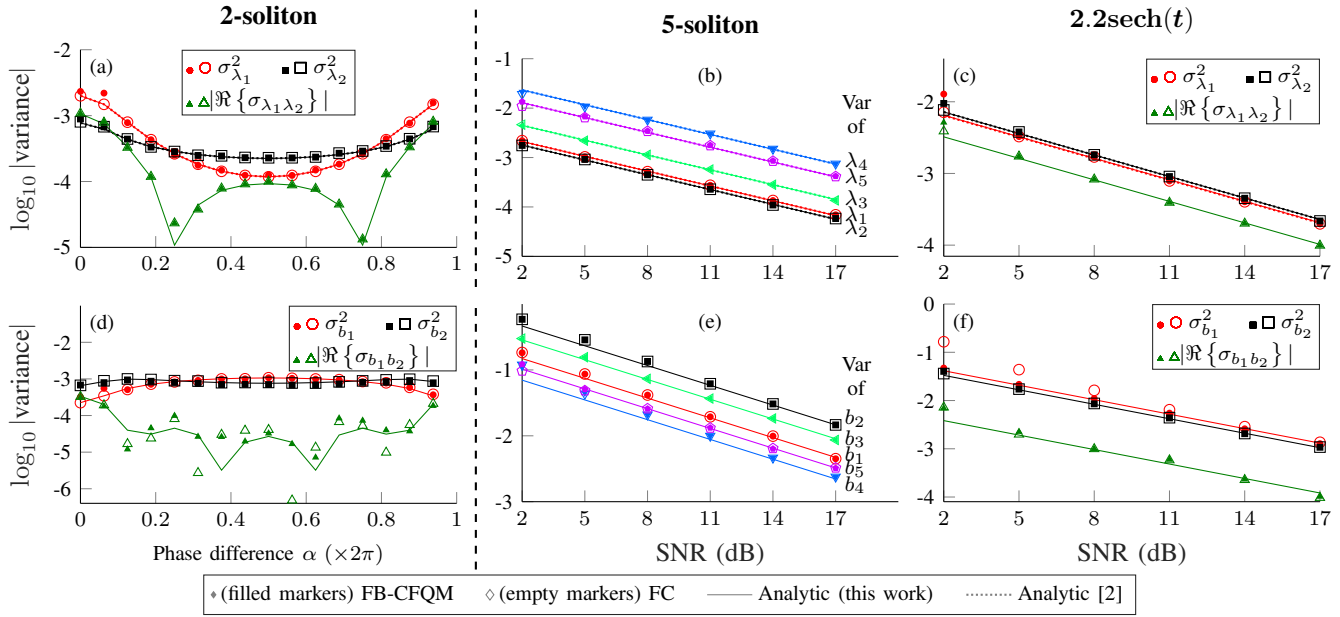

\begin{figure}[tbp]\centering
    \setlength{\figurewidth}{0.55\columnwidth}
    \setlength{\figureheight}{\figurewidth}
    \input{Graphs/C_big.tex}
    \caption{Covariance matrix of $\lambda_k$ and $b_k$ for a $2$-soliton with $\lambda_k=\left[0.6j, 0.3j\right]$ and $b_k=\left[\frac{1}{3}j, -\frac{1}{3}j\right]$ (SNR=$10$ dB). Note that the color scale is logarithmic.}
    \label{fig:C_big}
\end{figure}
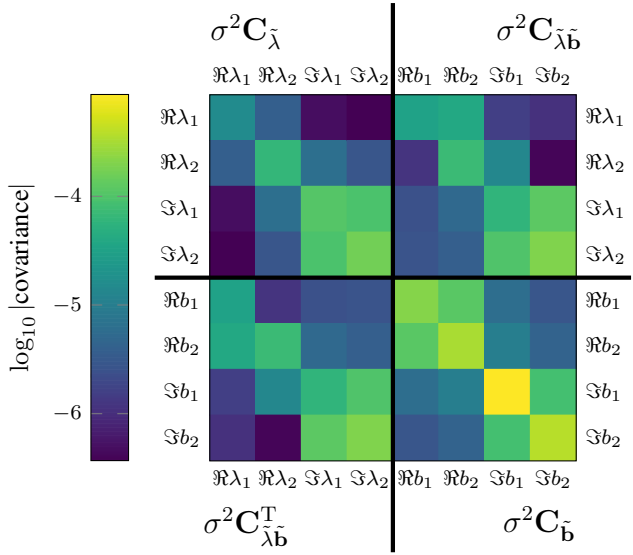

\section{Conclusion}\label{sec:conclusion}
We extended the Fourier Collocation (FC) method~\cite[Section 2.4.3]{yang_nlse} to compute the full discrete spectrum ($\lambda_k$ and $b(\lambda_k)$) of an arbitrary pulse. We showed numerically that our extended FC method 
estimates the discrete spectrum rather precisely with small number of samples.
In comparison to the state-of-the-art FB-CFQM algorithm, both algorithms have a comparable 
estimation error for small number of samples. 
However, the FB-CFQM estimation errors decrease monotonically in number of samples while the FC estimation errors saturates to an error floor at some number of samples.

We applied perturbation theory of linear operators~\cite{kato_eigenvalues} to our method and derived analytic expressions for the second-order statistics of the discrete spectrum (eigenvalues and spectral amplitudes) of a pulse contaminated with additive white Gaussian noise. Our simulations show that our expressions very accurately predict the numerical statistics. Our expressions, though involved, are much faster than Monte-Carlo simulations and could be used to design better NFT transmission systems by avoiding combinations of spectral parameters that are less robust to noise.

{This work assumes an AWGN channel. In optical fiber, usually a distributed noise model is assumed, where AWGN noise is added incrementally along the fiber. Our work can be extended to the distributed model by integrating our ($z$-dependent) analytic covariances along the $z$ variable, similarly to~\cite{prati2018some}. This is left for future work.}

\section{Acknowledgment}
The authors would like to thank Shrinivas Chimmalgi 
for pointing out numerical mistakes in evaluation of 
the FB-CFQM algorithm in the earlier version.

\bibliographystyle{IEEEtran}
\bibliography{soliton_fc_perturbation.bib}

\end{document}

%% file: Graphs/ringing_a.tex
\definecolor{mycolor1}{rgb}{0.00000,0.44700,0.74100}%
\definecolor{mycolor2}{rgb}{0.85000,0.32500,0.09800}%
\definecolor{mycolor3}{rgb}{0,0,0}%
\begin{tikzpicture}

\begin{axis}[%
width=0.951\figurewidth,
height=\figureheight,
at={(0\figurewidth,0\figureheight)},
scale only axis,
xmin=-17.6743316604409,
xmax=17.6743316604409,
xlabel style={font=\color{white!15!black}},
xlabel={t},
ymin=-0.1,
ymax=1.8,
axis background/.style={fill=white},
legend style={legend cell align=left, align=left, draw=white!15!black}
]
\addplot [color=mycolor1]
  table[row sep=crcr]{%
-17.6095905188276	1\\
-17.2643044302231	0.999999993468569\\
-16.9190183416187	0.999999975587148\\
-16.5737322530142	0.999999940604576\\
-16.2284461644098	0.999999881225165\\
-15.8831600758053	0.999999786131359\\
-15.5378739872008	0.999999637971754\\
-15.1925878985964	0.99999941030722\\
-14.8473018099919	0.999999062982718\\
-14.5020157213874	0.999998535119238\\
-14.156729632783	0.999997734507534\\
-13.8114435441785	0.999996521561587\\
-13.4661574555741	0.999994685047567\\
-13.1208713669696	0.999991905381576\\
-12.7755852783651	0.999987699142996\\
-12.4302991897607	0.999981335213962\\
-12.0850131011562	0.99997170808096\\
-11.7397270125517	0.999957146502557\\
-11.3944409239473	0.999935124738878\\
-11.0491548353428	0.999901827050815\\
-10.7038687467384	0.999851491565863\\
-10.3585826581339	0.999775423042114\\
-10.0132965695294	0.999660510094307\\
-9.66801048092496	0.999487003551243\\
-9.3227243923205	0.999225198895446\\
-8.97743830371604	0.998830505311616\\
-8.63215221511157	0.998236165178755\\
-8.28686612650711	0.997342606620943\\
-7.94158003790265	0.996002088441509\\
-7.59629394929818	0.993997015656235\\
-7.25100786069372	0.991010293193091\\
-6.90572177208926	0.986586865767021\\
-6.56043568348479	0.980088207305177\\
-6.21514959488033	0.97064777466405\\
-5.86986350627587	0.957147724654791\\
-5.52457741767141	0.938256814494752\\
-5.17929132906694	0.912591344279807\\
-4.83400524046248	0.87906447435461\\
-4.48871915185802	0.837433085851036\\
-4.14343306325356	0.788902342301203\\
-3.79814697464909	0.736462954217017\\
-3.45286088604463	0.684613256336065\\
-3.10757479744017	0.638427911792494\\
-2.7622887088357	0.60239182264863\\
-2.41700262023124	0.579518656571177\\
-2.07171653162678	0.570713806855205\\
-1.72643044302231	0.573121988417106\\
-1.38114435441785	0.574201910550576\\
-1.03585826581339	0.539524848713947\\
-0.690572177208926	0.42296466777098\\
-0.345286088604463	0.244395625126475\\
0	0.0970995956796994\\
0.345286088604463	0.0293494378476729\\
0.690572177208926	0.00184638051399234\\
1.03585826581339	0.0053014752259762\\
1.38114435441785	0.00565689988853786\\
1.72643044302231	0.00430028086311\\
2.07171653162678	0.00287732793354223\\
2.41700262023124	0.00179956858294161\\
2.7622887088357	0.00107949383936034\\
3.10757479744017	0.000629604635489181\\
3.45286088604463	0.000359984575858088\\
3.79814697464909	0.000202864575004436\\
4.14343306325356	0.000113095707965435\\
4.48871915185802	6.25390532155871e-05\\
4.83400524046248	3.43682764428769e-05\\
5.17929132906694	1.87968234693275e-05\\
5.52457741767141	1.02422571074705e-05\\
5.86986350627587	5.56472039908447e-06\\
6.21514959488033	3.01645498702308e-06\\
6.56043568348479	1.63216498152308e-06\\
6.90572177208926	8.81874121840992e-07\\
7.25100786069372	4.75938659958922e-07\\
7.59629394929818	2.56623724511882e-07\\
7.94158003790265	1.38268377265518e-07\\
8.28686612650711	7.4454645777157e-08\\
8.63215221511157	4.00731585924542e-08\\
8.97743830371604	2.15599708447599e-08\\
9.3227243923205	1.15959854574723e-08\\
9.66801048092496	6.2353069701056e-09\\
10.0132965695294	3.3521191783485e-09\\
10.3585826581339	1.80181082639055e-09\\
10.7038687467384	9.68368061859506e-10\\
11.0491548353428	5.20383959386662e-10\\
11.3944409239473	2.79619730809737e-10\\
11.7397270125517	1.50237477578275e-10\\
12.0850131011562	8.07158976570935e-11\\
12.4302991897607	4.33621782876202e-11\\
12.7755852783651	2.32932981306288e-11\\
13.1208713669696	1.25114755657227e-11\\
13.4661574555741	6.71925921425618e-12\\
13.8114435441785	3.60765744519723e-12\\
14.156729632783	1.9361343209651e-12\\
14.5020157213874	1.038226099309e-12\\
14.8473018099919	5.55896027427708e-13\\
15.1925878985964	2.96805874941513e-13\\
15.5378739872008	1.57633595743456e-13\\
15.8831600758053	8.28767918797261e-14\\
16.2284461644098	4.27212413789637e-14\\
16.5737322530142	2.11518494801854e-14\\
16.9190183416187	9.56599277774508e-15\\
17.2643044302231	3.3427502691375e-15\\
17.6095905188276	0\\
};
\addlegendentry{FB-CFQM}

\addplot [color=mycolor2, dash pattern=on 4pt off 1pt on 1pt off 1pt]
  table[row sep=crcr]{%
-17.6095905188276	0.999999999998983\\
-17.2643044302231	1.33683622767274\\
-16.9190183416187	1.06325912056106\\
-16.5737322530142	1.28573154980334\\
-16.2284461644098	1.10471493422612\\
-15.8831600758053	1.25218247118694\\
-15.5378739872008	1.13197742703381\\
-15.1925878985964	1.23008117763345\\
-14.8473018099919	1.1499677109691\\
-14.5020157213874	1.2154705672602\\
-14.156729632783	1.16187946348456\\
-13.8114435441785	1.20577739148502\\
-13.4661574555741	1.16979194148105\\
-13.1208713669696	1.19932170552482\\
-12.7755852783651	1.17506214680374\\
-12.4302991897607	1.19500115119708\\
-12.0850131011562	1.17857571629951\\
-11.7397270125517	1.19208559215161\\
-11.3944409239473	1.18090586572421\\
-11.0491548353428	1.19007985155329\\
-10.7038687467384	1.18240954822167\\
-10.3585826581339	1.18862446254006\\
-10.0132965695294	1.18327457770623\\
-9.66801048092496	1.18740800198236\\
-9.3227243923205	1.18351576576758\\
-8.97743830371604	1.1860581651889\\
-8.63215221511157	1.18289546342255\\
-8.28686612650711	1.18395656633078\\
-7.94158003790265	1.18070431835123\\
-7.59629394929818	1.17987467416479\\
-7.25100786069372	1.17527164062797\\
-6.90572177208926	1.17125490236724\\
-6.56043568348479	1.16300578452673\\
-6.21514959488033	1.15294865774203\\
-5.86986350627587	1.13691118295126\\
-5.52457741767141	1.11573503971557\\
-5.17929132906694	1.08573505564343\\
-4.83400524046248	1.04724419910619\\
-4.48871915185802	0.998426186374378\\
-4.14343306325356	0.941751314430526\\
-3.79814697464909	0.87968115543783\\
-3.45286088604463	0.818371543808281\\
-3.10757479744017	0.763264960049443\\
-2.7622887088357	0.720332919424632\\
-2.41700262023124	0.692873062330309\\
-2.07171653162678	0.68251042221885\\
-1.72643044302231	0.686586656590551\\
-1.38114435441785	0.692581944110414\\
-1.03585826581339	0.662637486384406\\
-0.690572177208926	0.529598543040974\\
-0.345286088604463	0.303087578082513\\
0	0.118496301552207\\
0.345286088604463	0.0308444908915312\\
0.690572177208926	0.000448536303480061\\
1.03585826581339	0.00689049970144706\\
1.38114435441785	0.00693881047018633\\
1.72643044302231	0.00514764768132709\\
2.07171653162678	0.00342766182327397\\
2.41700262023124	0.00212105242526304\\
2.7622887088357	0.00127511829110557\\
3.10757479744017	0.000736425468103393\\
3.45286088604463	0.00042413040741412\\
3.79814697464909	0.000235837702272021\\
4.14343306325356	0.000133441754000624\\
4.48871915185802	7.21571233028867e-05\\
4.83400524046248	4.08131827716505e-05\\
5.17929132906694	2.14134615513058e-05\\
5.52457741767141	1.23467808710789e-05\\
5.86986350627587	6.18289263879676e-06\\
6.21514959488033	3.75236895186984e-06\\
6.56043568348479	1.7197073214134e-06\\
6.90572177208926	1.16865214019261e-06\\
7.25100786069372	4.44105142089342e-07\\
7.59629394929818	3.84120917798452e-07\\
7.94158003790265	9.35020013437426e-08\\
8.28686612650711	1.38570910957761e-07\\
8.63215221511157	4.87758599518144e-09\\
8.97743830371604	5.6857750599501e-08\\
9.3227243923205	1.26417183939518e-08\\
9.66801048092496	2.67585393535952e-08\\
10.0132965695294	1.26539392916203e-08\\
10.3585826581339	1.40628627398725e-08\\
10.7038687467384	9.50853682324973e-09\\
11.0491548353428	7.90358866946357e-09\\
11.3944409239473	6.62818146281051e-09\\
11.7397270125517	4.56116679474841e-09\\
12.0850131011562	4.5551800067692e-09\\
12.4302991897607	2.609337082681e-09\\
12.7755852783651	3.17131078490721e-09\\
13.1208713669696	1.41999314078946e-09\\
13.4661574555741	2.27249418310221e-09\\
13.8114435441785	6.77834854111605e-10\\
14.156729632783	1.69413525180045e-09\\
14.5020157213874	2.08324672198341e-10\\
14.8473018099919	1.32263479863617e-09\\
15.1925878985964	9.12855564550602e-11\\
15.5378739872008	1.08365464064923e-09\\
15.8831600758053	2.83659464968764e-10\\
16.2284461644098	9.29486562022724e-10\\
16.5737322530142	4.07787979466212e-10\\
16.9190183416187	8.2969641524429e-10\\
17.2643044302231	4.88224419089387e-10\\
17.6095905188276	7.64875665225058e-10\\
};
\addlegendentry{FC}

\addplot [color=mycolor3, densely dashed]
  table[row sep=crcr]{%
-17.6095905188276	0.999999999999164\\
-17.2643044302231	1.0000474650766\\
-16.9190183416187	1.00003963827484\\
-16.5737322530142	1.00008069030589\\
-16.2284461644098	1.00006943845863\\
-15.8831600758053	1.00010145712116\\
-15.5378739872008	1.00008995207592\\
-15.1925878985964	1.00011420217297\\
-14.8473018099919	1.00010368876834\\
-14.5020157213874	1.00012152240754\\
-14.156729632783	1.00011216880321\\
-13.8114435441785	1.00012457227272\\
-13.4661574555741	1.00011576626357\\
-13.1208713669696	1.00012297331128\\
-12.7755852783651	1.00011313734261\\
-12.4302991897607	1.00011395824652\\
-12.0850131011562	1.00009967142773\\
-11.7397270125517	1.00009000271263\\
-11.3944409239473	1.00006376534716\\
-11.0491548353428	1.00003316435866\\
-10.7038687467384	0.999978196267899\\
-10.3585826581339	0.999902043459915\\
-10.0132965695294	0.999780436709832\\
-9.66801048092496	0.999602156287055\\
-9.3227243923205	0.999328169582613\\
-8.97743830371604	0.998919345981131\\
-8.63215221511157	0.998300012128286\\
-8.28686612650711	0.997372581963237\\
-7.94158003790265	0.995978949359389\\
-7.59629394929818	0.993898874557798\\
-7.25100786069372	0.990801376741319\\
-6.90572177208926	0.986223539798902\\
-6.56043568348479	0.97951057401843\\
-6.21514959488033	0.969790446577418\\
-5.86986350627587	0.955945061633416\\
-5.52457741767141	0.936677068110613\\
-5.17929132906694	0.910675624680056\\
-4.83400524046248	0.876990933274089\\
-4.48871915185802	0.835549335446089\\
-4.14343306325356	0.787709606848904\\
-3.79814697464909	0.736495851210411\\
-3.45286088604463	0.686290184938185\\
-3.10757479744017	0.64193008647028\\
-2.7622887088357	0.607671983979946\\
-2.41700262023124	0.586311213194373\\
-2.07171653162678	0.578392183310731\\
-1.72643044302231	0.579823606947824\\
-1.38114435441785	0.575126352460485\\
-1.03585826581339	0.529155477629183\\
-0.690572177208926	0.410414924609877\\
-0.345286088604463	0.246443937981668\\
0	0.108868890757841\\
0.345286088604463	0.033207949973414\\
0.690572177208926	0.00371369886877755\\
1.03585826581339	0.00455294134613409\\
1.38114435441785	0.00537416825076696\\
1.72643044302231	0.00420012406811837\\
2.07171653162678	0.0028459966055256\\
2.41700262023124	0.00179305211600142\\
2.7622887088357	0.00108062514724537\\
3.10757479744017	0.000632372525851376\\
3.45286088604463	0.000362414878652638\\
3.79814697464909	0.000204609193716968\\
4.14343306325356	0.000114214652307324\\
4.48871915185802	6.32272505721649e-05\\
4.83400524046248	3.47707598272962e-05\\
5.17929132906694	1.90304788502587e-05\\
5.52457741767141	1.03729373223871e-05\\
5.86986350627587	5.63855732217783e-06\\
6.21514959488033	3.05649063345948e-06\\
6.56043568348479	1.654429629147e-06\\
6.90572177208926	8.93504455366493e-07\\
7.25100786069372	4.82265991615683e-07\\
7.59629394929818	2.59657778399145e-07\\
7.94158003790265	1.39781969317468e-07\\
8.28686612650711	7.49445964533801e-08\\
8.63215221511157	4.01486240093554e-08\\
8.97743830371604	2.13098696702339e-08\\
9.3227243923205	1.12426824639534e-08\\
9.66801048092496	5.77372211015688e-09\\
10.0132965695294	2.87009348047096e-09\\
10.3585826581339	1.28127312872438e-09\\
10.7038687467384	4.46980272223608e-10\\
11.0491548353428	1.58826535174737e-11\\
11.3944409239473	2.54224666007654e-10\\
11.7397270125517	3.89854605853502e-10\\
12.0850131011562	4.57292418149466e-10\\
12.4302991897607	4.9745143361625e-10\\
12.7755852783651	5.16216137803904e-10\\
13.1208713669696	5.28309353004484e-10\\
13.4661574555741	5.33380854229597e-10\\
13.8114435441785	5.37111580182185e-10\\
14.156729632783	5.38416403747968e-10\\
14.5020157213874	5.39599630819711e-10\\
14.8473018099919	5.3991131825133e-10\\
15.1925878985964	5.40292546661592e-10\\
15.5378739872008	5.40363339349739e-10\\
15.8831600758053	5.40481398746588e-10\\
16.2284461644098	5.40503447565461e-10\\
16.5737322530142	5.40531833042692e-10\\
16.9190183416187	5.40548055267915e-10\\
17.2643044302231	5.40545859136642e-10\\
17.6095905188276	5.40563727396567e-10\\
};
\addlegendentry{FC+Hann}
\addplot [color=mycolor1]
  table[row sep=crcr]{%
-17.6095905188276	0\\
-17.2643044302231	9.14791834864635e-14\\
-16.9190183416187	2.61787184220867e-13\\
-16.5737322530142	5.78850825133441e-13\\
-16.2284461644098	1.16912819499705e-12\\
-15.8831600758053	2.26804224705164e-12\\
-15.5378739872008	4.31386842159148e-12\\
-15.1925878985964	8.12251378228433e-12\\
-14.8473018099919	1.52128753808685e-11\\
-14.5020157213874	2.84124925227946e-11\\
-14.156729632783	5.29849284666774e-11\\
-13.8114435441785	9.8728234046697e-11\\
-13.4661574555741	1.83880751013224e-10\\
-13.1208713669696	3.42390405125894e-10\\
-12.7755852783651	6.37442895592694e-10\\
-12.4302991897607	1.18663519096105e-09\\
-12.0850131011562	2.20881279452198e-09\\
-11.7397270125517	4.11119970944987e-09\\
-11.3944409239473	7.65144881033288e-09\\
-11.0491548353428	1.42389493709619e-08\\
-10.7038687467384	2.64948526362436e-08\\
-10.3585826581339	4.92924452057701e-08\\
-10.0132965695294	9.16885765172688e-08\\
-9.66801048092496	1.70506017449002e-07\\
-9.3227243923205	3.16969522891747e-07\\
-8.97743830371604	5.88977036636647e-07\\
-8.63215221511157	1.09373513058389e-06\\
-8.28686612650711	2.02936592170895e-06\\
-7.94158003790265	3.76100279492131e-06\\
-7.59629394929818	6.95895896849699e-06\\
-7.25100786069372	1.28469781819649e-05\\
-6.90572177208926	2.36414658132076e-05\\
-6.56043568348479	4.3311220044391e-05\\
-6.21514959488033	7.88480409877663e-05\\
-5.86986350627587	0.000142286767488515\\
-5.52457741767141	0.000253675438208559\\
-5.17929132906694	0.000444923097135759\\
-4.83400524046248	0.000763751428935378\\
-4.48871915185802	0.00127582663189638\\
-4.14343306325356	0.00206211345041481\\
-3.79814697464909	0.00320871909282287\\
-3.45286088604463	0.00478800206364243\\
-3.10757479744017	0.00682808528470803\\
-2.7622887088357	0.00925225698516697\\
-2.41700262023124	0.0117204915002906\\
-2.07171653162678	0.0131766251643125\\
-1.72643044302231	0.0105931167901951\\
-1.38114435441785	0.00400446436414932\\
-1.03585826581339	0.0468339745955299\\
-0.690572177208926	0.130080735070706\\
-0.345286088604463	0.226687909639929\\
0	0.294240832514692\\
0.345286088604463	0.313544583417435\\
0.690572177208926	0.318296386345852\\
1.03585826581339	0.318300227652725\\
1.38114435441785	0.318394105642374\\
1.72643044302231	0.319339760113673\\
2.07171653162678	0.320831543261139\\
2.41700262023124	0.322468692551855\\
2.7622887088357	0.323986785788129\\
3.10757479744017	0.325263200928603\\
3.45286088604463	0.32627039124241\\
3.79814697464909	0.327031060622399\\
4.14343306325356	0.327587761874896\\
4.48871915185802	0.327985862707552\\
4.83400524046248	0.328265647210636\\
5.17929132906694	0.328459700213031\\
5.52457741767141	0.328592932087799\\
5.86986350627587	0.328683687939946\\
6.21514959488033	0.32874513002273\\
6.56043568348479	0.328786525402602\\
6.90572177208926	0.328814307939883\\
7.25100786069372	0.328832897438795\\
7.59629394929818	0.328845305579805\\
7.94158003790265	0.328853571662771\\
8.28686612650711	0.328859069776098\\
8.63215221511157	0.328862722200545\\
8.97743830371604	0.32886514606536\\
9.3227243923205	0.328866753302704\\
9.66801048092496	0.328867818338583\\
10.0132965695294	0.328868523706785\\
10.3585826581339	0.328868990666147\\
10.7038687467384	0.328869299688178\\
11.0491548353428	0.328869504132747\\
11.3944409239473	0.328869639358814\\
11.7397270125517	0.328869728784458\\
12.0850131011562	0.328869787912622\\
12.4302991897607	0.328869827002744\\
12.7755852783651	0.328869852842289\\
13.1208713669696	0.328869869920663\\
13.4661574555741	0.328869881206678\\
13.8114435441785	0.32886988866325\\
14.156729632783	0.328869893588002\\
14.5020157213874	0.328869896838603\\
14.8473018099919	0.328869898981801\\
15.1925878985964	0.328869900391985\\
15.5378739872008	0.328869901316332\\
15.8831600758053	0.328869901917879\\
16.2284461644098	0.328869902303972\\
16.5737322530142	0.32886990254506\\
16.9190183416187	0.328869902687093\\
17.2643044302231	0.328869902759694\\
17.6095905188276	0.328869902786212\\
};
% \addlegendentry{FB-CFQM}

\addplot [color=mycolor2, dash pattern=on 4pt off 1pt on 1pt off 1pt]
  table[row sep=crcr]{%
-17.6095905188276	3.81943048337436e-11\\
-17.2643044302231	4.30266677000933e-10\\
-16.9190183416187	5.46045607075626e-11\\
-16.5737322530142	5.44579477498653e-10\\
-16.2284461644098	1.98444187483587e-10\\
-15.8831600758053	7.20328808198525e-10\\
-15.5378739872008	4.23382711483867e-10\\
-15.1925878985964	9.89570870461354e-10\\
-14.8473018099919	7.80363098333703e-10\\
-14.5020157213874	1.39667483931836e-09\\
-14.156729632783	1.36219064799616e-09\\
-13.8114435441785	1.99018048028083e-09\\
-13.4661574555741	2.35810623114029e-09\\
-13.1208713669696	2.77158098321857e-09\\
-12.7755852783651	4.21358155907217e-09\\
-12.4302991897607	3.48061123572871e-09\\
-12.0850131011562	8.1361631013538e-09\\
-11.7397270125517	2.79088498328567e-09\\
-11.3944409239473	1.77687327319031e-08\\
-11.0491548353428	4.57389504046992e-09\\
-10.7038687467384	4.48494759184569e-08\\
-10.3585826581339	3.79764129147437e-08\\
-10.0132965695294	1.28554024743666e-07\\
-9.66801048092496	1.66131371372937e-07\\
-9.3227243923205	4.01817925029386e-07\\
-8.97743830371604	6.28932848633291e-07\\
-8.63215221511157	1.31842733184146e-06\\
-8.28686612650711	2.25509324548817e-06\\
-7.94158003790265	4.42539324463233e-06\\
-7.59629394929818	7.87639333399057e-06\\
-7.25100786069372	1.49496346623699e-05\\
-6.90572177208926	2.70053322371221e-05\\
-6.56043568348479	5.01803573319874e-05\\
-6.21514959488033	9.0576390362356e-05\\
-5.86986350627587	0.000164803423572549\\
-5.52457741767141	0.000292895010482779\\
-5.17929132906694	0.000516857960672217\\
-4.83400524046248	0.000887363973077121\\
-4.48871915185802	0.00149077189942196\\
-4.14343306325356	0.00241430637988768\\
-3.79814697464909	0.00377709211558553\\
-3.45286088604463	0.00564954711691208\\
-3.10757479744017	0.0080977773996984\\
-2.7622887088357	0.0110064347684954\\
-2.41700262023124	0.014060012824892\\
-2.07171653162678	0.0160245353533878\\
-1.72643044302231	0.0136695142864799\\
-1.38114435441785	0.00223269826259877\\
-1.03585826581339	0.0515897592869123\\
-0.690572177208926	0.155747599930313\\
-0.345286088604463	0.280556183318568\\
0	0.355827704218754\\
0.345286088604463	0.379770758408232\\
0.690572177208926	0.383281851402579\\
1.03585826581339	0.382828704424116\\
1.38114435441785	0.38282368373473\\
1.72643044302231	0.384115332563276\\
2.07171653162678	0.385813081284377\\
2.41700262023124	0.387924377986733\\
2.7622887088357	0.389555402958731\\
3.10757479744017	0.391256012480594\\
3.45286088604463	0.392204221196609\\
3.79814697464909	0.393359339842581\\
4.14343306325356	0.393692267972888\\
4.48871915185802	0.39452072453386\\
4.83400524046248	0.394413631718644\\
5.17929132906694	0.395144100799515\\
5.52457741767141	0.394701076910041\\
5.86986350627587	0.395510402993053\\
6.21514959488033	0.394747075199733\\
6.56043568348479	0.395783907267388\\
6.90572177208926	0.394640426230391\\
7.25100786069372	0.396061133575447\\
7.59629394929818	0.394406518263364\\
7.94158003790265	0.39640966933\\
8.28686612650711	0.394030395934493\\
8.63215221511157	0.396893469187055\\
8.97743830371604	0.393464927834198\\
9.3227243923205	0.397590828738236\\
9.66801048092496	0.392628349956574\\
10.0132965695294	0.398611007231402\\
10.3585826581339	0.391393112428018\\
10.7038687467384	0.400114312985744\\
11.0491548353428	0.389565386098075\\
11.3944409239473	0.402340281933867\\
11.7397270125517	0.386852174716692\\
12.0850131011562	0.405649623228753\\
12.4302991897607	0.382810257598787\\
12.7755852783651	0.410588003655352\\
13.1208713669696	0.376767490407546\\
13.4661574555741	0.417983890518914\\
13.8114435441785	0.367701464955677\\
14.156729632783	0.429099381750414\\
14.5020157213874	0.354051953223883\\
14.8473018099919	0.44586362849374\\
15.1925878985964	0.333430057491792\\
15.5378739872008	0.471235418791456\\
15.8831600758053	0.302165565524179\\
16.2284461644098	0.509768527742608\\
16.5737322530142	0.254599824428477\\
16.9190183416187	0.568496800806605\\
17.2643044302231	0.181976441340438\\
17.6095905188276	0.658327905492591\\
};
% \addlegendentry{FC}

\addplot [color=mycolor3, densely dashed]
  table[row sep=crcr]{%
-17.6095905188276	1.80151522413674e-10\\
-17.2643044302231	1.80068945593036e-10\\
-16.9190183416187	1.79896440348291e-10\\
-16.5737322530142	1.79575087749913e-10\\
-16.2284461644098	1.78998880542895e-10\\
-15.8831600758053	1.77870741146515e-10\\
-15.5378739872008	1.75874682897328e-10\\
-15.1925878985964	1.71980244839354e-10\\
-14.8473018099919	1.65013208995925e-10\\
-14.5020157213874	1.5160499195107e-10\\
-14.156729632783	1.27297463874529e-10\\
-13.8114435441785	8.10812769509244e-11\\
-13.4661574555741	3.55487627312195e-12\\
-13.1208713669696	1.63117305269272e-10\\
-12.7755852783651	4.57266050260772e-10\\
-12.4302991897607	1.00884135863893e-09\\
-12.0850131011562	2.02972289515551e-09\\
-11.7397270125517	3.9377420493601e-09\\
-11.3944409239473	7.47713908544307e-09\\
-11.0491548353428	1.40787561941877e-08\\
-10.7038687467384	2.63385345823657e-08\\
-10.3585826581339	4.91735253920842e-08\\
-10.0132965695294	9.15942201410237e-08\\
-9.66801048092496	1.70513311227001e-07\\
-9.3227243923205	3.17069790323167e-07\\
-8.97743830371604	5.89340437044892e-07\\
-8.63215221511157	1.0943470981752e-06\\
-8.28686612650711	2.03048911176699e-06\\
-7.94158003790265	3.76225926835618e-06\\
-7.59629394929818	6.95970461181422e-06\\
-7.25100786069372	1.28429265868215e-05\\
-6.90572177208926	2.36214935693109e-05\\
-6.56043568348479	4.32404015791252e-05\\
-6.21514959488033	7.86373531202359e-05\\
-5.86986350627587	0.000141704960798255\\
-5.52457741767141	0.000252182677315309\\
-5.17929132906694	0.000441317039297024\\
-4.83400524046248	0.000755641346511305\\
-4.48871915185802	0.00125892450198743\\
-4.14343306325356	0.00202983350645169\\
-3.79814697464909	0.00315229744672311\\
-3.45286088604463	0.00469736362962129\\
-3.10757479744017	0.00669098224655559\\
-2.7622887088357	0.0090480801862795\\
-2.41700262023124	0.0113970572781913\\
-2.07171653162678	0.0125909839758588\\
-1.72643044302231	0.00939360588686788\\
-1.38114435441785	0.00625303107881328\\
-1.03585826581339	0.048576904316347\\
-0.690572177208926	0.125664959306164\\
-0.345286088604463	0.216987103879972\\
0	0.284502990959729\\
0.345286088604463	0.314575495894577\\
0.690572177208926	0.322204061599283\\
1.03585826581339	0.322943183117545\\
1.38114435441785	0.32304566874055\\
1.72643044302231	0.323868352110599\\
2.07171653162678	0.32526278949888\\
2.41700262023124	0.326854661105763\\
2.7622887088357	0.328364473192127\\
3.10757479744017	0.329653860561166\\
3.45286088604463	0.330681861692642\\
3.79814697464909	0.331464486597895\\
4.14343306325356	0.332040411005733\\
4.48871915185802	0.332454150848722\\
4.83400524046248	0.332745819004951\\
5.17929132906694	0.332948685523929\\
5.52457741767141	0.333088197789351\\
5.86986350627587	0.333183422719379\\
6.21514959488033	0.333247913592961\\
6.56043568348479	0.33329146873763\\
6.90572177208926	0.333320627688428\\
7.25100786069372	0.333340267094848\\
7.59629394929818	0.333353209004896\\
7.94158003790265	0.333362045456039\\
8.28686612650711	0.333367623225662\\
8.63215221511157	0.333371690834138\\
8.97743830371604	0.333373891467859\\
9.3227243923205	0.333375939452968\\
9.66801048092496	0.333376501789986\\
10.0132965695294	0.333377798291476\\
10.3585826581339	0.333377428275576\\
10.7038687467384	0.333378594172311\\
11.0491548353428	0.333377510879385\\
11.3944409239473	0.333378901038733\\
11.7397270125517	0.333377072536964\\
12.0850131011562	0.33337895104498\\
12.4302991897607	0.33337618851785\\
12.7755852783651	0.333378810271711\\
13.1208713669696	0.333374798491418\\
13.4661574555741	0.333378444753752\\
13.8114435441785	0.333372747365192\\
14.156729632783	0.33337772971906\\
14.5020157213874	0.333369794328303\\
14.8473018099919	0.333376419744158\\
15.1925878985964	0.333365610958044\\
15.5378739872008	0.333374078195624\\
15.8831600758053	0.333359787588765\\
16.2284461644098	0.33336995173754\\
16.5737322530142	0.333351893714543\\
16.9190183416187	0.33336276867982\\
17.2643044302231	0.333342021041943\\
17.6095905188276	0.333345001063307\\
};
% \addlegendentry{FC+Hann}

\draw (-10, 1.09) ellipse (0.2 and 0.136);
\node at (-11, 0.98) [anchor=north] {$|v_1(t,\lambda_1)e^{j\lambda_1 t}|$};

\draw (10, 0.37) ellipse (0.2 and 0.105);
\node at (9, 0.45) [anchor=south] {$|v_2(t,\lambda_1)e^{-j\lambda_1 t}|$};

\end{axis}

% \begin{axis}[%
% width=1.227\figurewidth,
% height=1.227\figureheight,
% at={(-0.16\figurewidth,-0.135\figureheight)},
% scale only axis,
% xmin=0,
% xmax=1,
% ymin=0,
% ymax=1,
% axis line style={draw=none},
% ticks=none,
% axis x line*=bottom,
% axis y line*=left,
% legend style={legend cell align=left, align=left, draw=white!15!black}
% ]
% \end{axis}
\end{tikzpicture}%

%% file: Graphs/pert_2sol_l.tex
\definecolor{mycolor1}{rgb}{0.00000,0,0}%
\begin{tikzpicture}
\begin{axis}[%
width=0.951\figurewidth,
height=\figureheight,
at={(0\figurewidth,0\figureheight)},
scale only axis,
xmode=log,
xmin=47,
xmax=1455,
xminorticks=true,
xlabel style={font=\color{white!15!black}},
ymode=log,
ymin=1e-13,
ymax=500,
ytick={1e0,1e-1, 1e-2, 1e-3, 1e-4, 1e-5, 1e-6, 1e-7, 1e-8, 1e-9, 1e-10, 1e-11, 1e-12},
yticklabels={0,, -2,, -4,, -6, ,-8,,-10,,-12},
yminorticks=true,
ylabel style={font=\color{white!15!black}},
ylabel={$\log_{10}\mathrm{err}(\lambda_k, \hat{\lambda}_k)$},
axis background/.style={fill=white},
title style={font=\bfseries},
title={2-soliton},
axis x line*=bottom,
axis y line*=left,
legend columns=3,
legend style={at={(1, 1)}, anchor=north east, inner sep=0cm, legend cell align=right, align=right, draw=white!15!black, /tikz/every even column/.append style={column sep=0.8em}}
]
\addlegendimage{empty legend}
\addlegendentry{FB-CFQM:}
\addplot [color=red, semithick, densely dotted, mark=*, mark options={solid, red, scale=0.5, fill=red}]
  table[row sep=crcr]{%
47	0.020278527842532\\
65	0.00245942872461222\\
91	4.31321643322694e-05\\
129	1.35158323772488e-07\\
183	3.15801155270249e-08\\
257	3.97273068537974e-09\\
365	4.97058740877571e-10\\
515	6.10072166494474e-11\\
727	6.20832711754891e-12\\
1029	7.57962329211516e-13\\
1455	1.49397668461831e-12\\
};
\addlegendentry{$\lambda_1$}
\addplot [color=mycolor1, semithick, densely dotted, mark=square*, mark options={solid, mycolor1, scale=0.5, fill=mycolor1}]
table[row sep=crcr]{%
	47	0.0170641316484428\\
	65	0.00151372110450172\\
	91	2.09154808554646e-05\\
	129	8.83483598492988e-08\\
	183	7.67605988053469e-08\\
	257	1.03045949174809e-07\\
	365	9.53922991833607e-08\\
	515	1.00251490126044e-07\\
	727	1.0410748523995e-07\\
	1029	1.02196782333229e-07\\
	1455	1.02482839436927e-07\\
};
\addlegendentry{$\lambda_{2}$}
\addlegendimage{empty legend}
\addlegendentry{FC:}
\addplot [color=red, dashed, mark=o, mark options={solid, red}]
  table[row sep=crcr]{%
  	47	0.0315454470251619\\
  	65	0.00131960065289316\\
  	91	5.7915652998212e-05\\
  	129	8.7552735061891e-08\\
  	183	6.55786517682226e-09\\
  	257	7.62117369215917e-09\\
  	365	6.96421531554575e-09\\
  	515	7.31354458476242e-09\\
  	727	7.60025320465489e-09\\
  	1029	7.46581037764476e-09\\
  	1455	7.49056835262477e-09\\
};
\addlegendentry{$\lambda{}_\text{1}$}
\addplot [color=mycolor1, dashed, mark=square, mark options={solid, mycolor1}]
  table[row sep=crcr]{%
47	0.0301776082978034\\
65	0.00050989373216126\\
91	0.00025955873325656\\
129	0.000290666512694863\\
183	0.000281212991953643\\
257	0.00030290778980091\\
365	0.000289538569497131\\
515	0.000296722438599499\\
727	0.000302491772154097\\
1029	0.000299799206667413\\
1455	0.000300293310147396\\
};
\addlegendentry{$\lambda{}_\text{2}$}

\node at (rel axis cs:0,1) [anchor=north west] {(a)};
\end{axis}

\end{tikzpicture}%

%% file: Graphs/pert_5sol_l.tex
\definecolor{mycolor1}{rgb}{0,0,0}%
\definecolor{mycolor2}{rgb}{0.00000,1.00000,0.40000}%
\definecolor{mycolor3}{rgb}{0.00000,0.40000,1.00000}%
\definecolor{mycolor4}{rgb}{0.80000,0.00000,1.00000}%
\begin{tikzpicture}
\begin{axis}[%
width=0.951\figurewidth,
height=\figureheight,
at={(0\figurewidth,0\figureheight)},
scale only axis,
xmode=log,
xmin=100,
xmax=2000,
xminorticks=true,
ymode=log,
ymin=5e-15,
ymax=2000,
ytick={1,1e-1,1e-2,1e-3,1e-4,1e-5,1e-6,1e-7,1e-8,1e-9,1e-10,1e-11,1e-12,1e-13,1e-14,1e-15},
yticklabels={0,,,-3,,,-6,,,-9,,,-12,,,-15},
yminorticks=true,
axis background/.style={fill=white},
title style={font=\bfseries},
title={5-soliton},
axis x line*=bottom,
axis y line*=left,
legend style={at={(1, 1)}, anchor=north east, inner sep=0cm, legend cell align=left, align=left, draw=white!15!black}
]
\addlegendimage{color=gray, densely dotted, semithick, mark=diamond*, mark options={scale=0.5, solid}}
\addlegendentry{(filled markers) FB-CFQM}
\addlegendimage{color=gray, dashed, mark=diamond, mark options={solid}}
\addlegendentry{(empty markers) FC}
\addplot[color=red, densely dotted, semithick, mark=*, mark options={solid, red, scale=0.5, fill=red}, forget plot]
  table[row sep=crcr]{%
115	3.59714301098506e-07\\
161	1.68504722159193e-08\\
229	2.13287897246556e-09\\
323	2.67810043404333e-10\\
455	3.35534747574399e-11\\
643	4.19828810376588e-12\\
909	5.26650008911562e-13\\
1285	6.53406918677819e-14\\
1817	2.12031223851959e-14\\
};
%\addlegendentry{$\lambda{}_\text{1}\text{, FB}$}
\addplot [color=red, dashed, mark=o, mark options={solid, red}, forget plot]
  table[row sep=crcr]{%
115	3.40903977222692e-07\\
161	1.73982689802995e-11\\
229	2.42855921744963e-14\\
323	6.52237106170216e-15\\
455	1.72203344534074e-14\\
643	3.10654224556567e-14\\
909	1.44935367672543e-14\\
1285	5.01830398204209e-14\\
1817	6.9398910692318e-14\\
};
%\addlegendentry{$\lambda{}_\text{1}\text{, FC}$}
\addplot [color=mycolor1, densely dotted, semithick, mark=square*, mark options={solid, mycolor1, scale=0.5, fill=mycolor1}, forget plot]
  table[row sep=crcr]{%
115	4.57758293957226e-08\\
161	1.09983210774604e-08\\
229	1.40305981966268e-09\\
323	1.76050245737263e-10\\
455	2.2053288474067e-11\\
643	2.75695387722694e-12\\
909	3.49315410652211e-13\\
1285	4.08545557327535e-14\\
1817	4.69373097595814e-14\\
};
%\addlegendentry{$\lambda{}_\text{2}\text{, FB}$}
\addplot [color=mycolor1, dashed, mark=square, mark options={solid, mycolor1}, forget plot]
  table[row sep=crcr]{%
115	2.43596047801459e-08\\
161	1.61201607901793e-10\\
229	2.32209163928348e-14\\
323	1.07337509342059e-14\\
455	3.09077728608804e-14\\
643	2.76040918220463e-14\\
909	1.13993787727786e-14\\
1285	4.39374348632517e-14\\
1817	4.58444776435741e-14\\
};
%\addlegendentry{$\lambda{}_\text{2}\text{, FC}$}
\addplot [color=mycolor2, densely dotted, semithick, mark=triangle*, mark options={rotate=90, solid, mycolor2, scale=0.5, fill=mycolor2}, forget plot]
  table[row sep=crcr]{%
115	4.73035147568665e-07\\
161	3.75924650243728e-08\\
229	4.81936141134341e-09\\
323	6.04639437187241e-10\\
455	7.57110608835062e-11\\
643	9.47419810876789e-12\\
909	1.17717935925208e-12\\
1285	1.48440091737042e-13\\
1817	2.83468117982797e-14\\
};
%\addlegendentry{$\lambda{}_\text{3}\text{, FB}$}
\addplot [color=mycolor2, dashed, mark=triangle, mark options={rotate=90, solid, mycolor2}, forget plot]
  table[row sep=crcr]{%
115	5.751087980222e-07\\
161	6.10242598205142e-10\\
229	1.3921586446543e-11\\
323	1.5086453769782e-11\\
455	1.70999121711231e-11\\
643	1.74884800530168e-11\\
909	1.77027584634324e-11\\
1285	1.79235672778094e-11\\
1817	1.80798644926486e-11\\
};
%\addlegendentry{$\lambda{}_\text{3}\text{, FC}$}
\addplot [color=mycolor3, densely dotted, semithick, mark=triangle*, mark options={solid, rotate=180,  scale=0.5, fill=mycolor3, mycolor3}, forget plot]
  table[row sep=crcr]{%
115	1.10526091897166e-06\\
161	7.1390220779368e-08\\
229	9.04809946150587e-09\\
323	1.10177740580911e-09\\
455	1.06796398430505e-10\\
643	4.68530799183131e-11\\
909	5.75935840965365e-11\\
1285	5.99510692899951e-11\\
1817	5.9622012496219e-11\\
};
%\addlegendentry{$\lambda{}_\text{4}\text{, FB}$}
\addplot [color=mycolor3, dashed, mark=triangle, mark options={solid, rotate=180, mycolor3}, forget plot]
  table[row sep=crcr]{%
115	1.89847724078048e-06\\
161	2.54610184574057e-07\\
229	2.24730556049078e-07\\
323	2.37762155210947e-07\\
455	2.5846838994348e-07\\
643	2.6248006909238e-07\\
909	2.64564523554007e-07\\
1285	2.66842669381169e-07\\
1817	2.67666383261692e-07\\
};
%\addlegendentry{$\lambda{}_\text{4}\text{, FC}$}
\addplot [color=mycolor4, densely dotted, semithick, mark=pentagon*, mark options={solid, mycolor4, scale=0.5, fill=mycolor4}, forget plot]
  table[row sep=crcr]{%
115	3.19487016786424e-06\\
161	7.43759266220672e-07\\
229	5.12461325832797e-07\\
323	5.19538422003752e-07\\
455	5.61905645899594e-07\\
643	5.70321453703397e-07\\
909	5.74823887945125e-07\\
1285	5.79773461283588e-07\\
1817	5.81563758081915e-07\\
};
%\addlegendentry{$\lambda{}_\text{5}\text{, FB}$}
\addplot [color=mycolor4, dashed, mark=pentagon, mark options={solid, mycolor4}, forget plot]
  table[row sep=crcr]{%
115	0.00131811382186988\\
161	0.0014669269703154\\
229	0.00138031701997924\\
323	0.0014195809457875\\
455	0.0014797985038519\\
643	0.00149118007662596\\
909	0.00149705942185681\\
1285	0.00150345808926237\\
1817	0.00150576496348938\\
};
%\addlegendentry{$\lambda{}_\text{5}\text{, FC}$}
\node at (rel axis cs:0, 1) [anchor=north west] {(b)};
\end{axis}
\end{tikzpicture}%

%% file: Graphs/pert_sech_l.tex
\definecolor{mycolor1}{rgb}{0,0,0}%
\begin{tikzpicture}
\begin{axis}[%
width=0.951\figurewidth,
height=\figureheight,
at={(0\figurewidth,0\figureheight)},
scale only axis,
xmode=log,
xmin=33,
xmax=1025,
xminorticks=true,
ymode=log,
ymin=5.84138943398947e-16,
ymax=5,
ytick={1,1e-1,1e-2,1e-3,1e-4,1e-5,1e-6,1e-7,1e-8,1e-9,1e-10,1e-11,1e-12,1e-13,1e-14,1e-15},
yticklabels={0,,,-3,,,-6,,,-9,,,-12,,,-15},
yminorticks=true,
axis background/.style={fill=white},
title style={font=\bfseries},
title={$\mathbf{2.2sech(\boldsymbol{t})}$},
axis x line*=bottom,
axis y line*=left,
legend columns=2,
transpose legend,
legend style={at={(1, 1)}, anchor=north east, inner sep=0cm, legend cell align=right, align=right, draw=white!15!black, /tikz/every even column/.append style={column sep=0.8em}}
]
\addlegendimage{empty legend}
\addlegendentry{FB-CFQM:}
\addlegendimage{empty legend}
\addlegendentry{FC:}
\addplot [color=red, densely dotted, semithick, mark=*, mark options={solid, scale=0.5, fill=red, red}]
  table[row sep=crcr]{%
33	0.000303517426895664\\
47	3.46450370297913e-06\\
65	1.12070861983513e-07\\
91	1.48470036361685e-08\\
129	1.86811760627403e-09\\
181	2.34284584678999e-10\\
257	2.93347134712885e-11\\
363	3.66782488081865e-12\\
513	4.60618950467794e-13\\
725	6.26025783268261e-14\\
1025	7.81458401547499e-15\\
};
\addlegendentry{$\lambda{}_{1}$}
\addplot [color=red, dashed, mark=o, mark options={solid, red}]
  table[row sep=crcr]{%
33	0.000107310213768961\\
47	1.38698136453077e-06\\
65	1.81759744510196e-09\\
91	9.6528435070606e-14\\
129	1.45063605389197e-15\\
181	1.01247202753286e-14\\
257	6.3627174560715e-15\\
363	1.30649886639436e-14\\
513	1.32479154414516e-14\\
725	3.8488384339513e-14\\
1025	3.74950162488004e-14\\
};
\addlegendentry{$\lambda{}_{1}$}
\addplot [color=mycolor1, densely dotted, semithick, mark=square*, mark options={solid, scale=0.5, fill=mycolor1, mycolor1}]
  table[row sep=crcr]{%
33	0.000490544042639911\\
47	8.93166763276991e-06\\
65	1.89225784259399e-07\\
91	2.60661167758447e-08\\
129	3.27730232210009e-09\\
181	4.06059580544906e-10\\
257	4.65980511113594e-11\\
363	2.35279123202374e-12\\
513	4.90353963076232e-12\\
725	5.69761885157503e-12\\
1025	5.79885518444742e-12\\
};
\addlegendentry{$\lambda{}_{2}$}
\addplot [color=mycolor1, dashed, mark=square, mark options={solid, mycolor1}]
  table[row sep=crcr]{%
33	0.00359324344340955\\
47	5.33824261238015e-05\\
65	4.80460777512499e-07\\
91	4.76247642433591e-07\\
129	4.57484098223152e-07\\
181	5.22496624696548e-07\\
257	4.88473046138844e-07\\
363	4.98833989194734e-07\\
513	5.04745997696534e-07\\
725	5.10528242379595e-07\\
1025	5.13084640964477e-07\\
};
\addlegendentry{$\lambda{}_{2}$}
%\legend{} % this hides the legend!
\node at (rel axis cs:0,1) [anchor=north west] {(c)};
\end{axis}
\end{tikzpicture}%

%% file: Graphs/pert_2sol_b.tex
\definecolor{mycolor1}{rgb}{0,0,0}%
\begin{tikzpicture}
\begin{axis}[%
width=0.951\figurewidth,
height=\figureheight,
at={(0\figurewidth,0\figureheight)},
scale only axis,
xmode=log,
xmin=47,
xmax=1455,
xminorticks=true,
xlabel style={font=\color{white!15!black}},
xlabel={Number of time samples, $M$},
ymode=log,
ymin=5e-8,
ymax=1,
ytick={1,1e-1,1e-2,1e-3,1e-4,1e-5,1e-6,1e-7},
yticklabels={0,-1,-2,-3,-4,-5,-6,-7},
yminorticks=true,
ylabel style={font=\color{white!15!black}},
ylabel={$\log_{10}\mathrm{err}(b_k, \hat{b}_k)$},
axis background/.style={fill=white},
axis x line*=bottom,
axis y line*=left,
legend columns=3,
legend style={at={(1,1)}, anchor=north east, inner sep=0cm, legend cell align=right, align=right, draw=white!15!black, /tikz/every even column/.append style={column sep=0.8em}}
]
\addlegendimage{empty legend}
\addlegendentry{FB-CFQM:}
\addplot [color=red, densely dotted, semithick, mark=*, mark options={solid, red, scale=0.5, fill=red}]
  table[row sep=crcr]{%
47	0.00944894024987011\\
65	0.000795069107140034\\
91	2.75548227083842e-05\\
129	1.09442050192029e-07\\
183	9.70547579849931e-08\\
257	1.03509812681567e-07\\
365	9.34785035613765e-08\\
515	9.78804892633773e-08\\
727	1.01605791122444e-07\\
1029	9.97437369392056e-08\\
1455	9.99132834800243e-08\\
};
\addlegendentry{${b}_{1}$}
\addplot [color=mycolor1, densely dotted, semithick, mark=square*, mark options={solid, mycolor1, scale=0.5, fill=mycolor1}]
table[row sep=crcr]{%
	47	0.00404423318141334\\
	65	0.000437184352039702\\
	91	1.54218655967513e-05\\
	129	8.90910044162412e-07\\
	183	8.31324620823271e-07\\
	257	9.48574023247926e-07\\
	365	8.68583952837636e-07\\
	515	9.08196195095479e-07\\
	727	9.40722693909752e-07\\
	1029	9.24179440309206e-07\\
	1455	9.2649473043549e-07\\
};
\addlegendentry{${b}_{2}$}
\addlegendimage{empty legend}
\addlegendentry{FC:}
\addplot [color=red, dashed, mark=o, mark options={solid, red}]
  table[row sep=crcr]{%
  	47	0.071076624021813\\
  	65	0.00570920973596306\\
  	91	3.08042874958005e-05\\
  	129	6.23641514968054e-07\\
  	183	3.13437093028326e-07\\
  	257	3.30124873642525e-07\\
  	365	2.96082857928736e-07\\
  	515	3.04473166695295e-07\\
  	727	3.1185602296805e-07\\
  	1029	3.04450040034365e-07\\
  	1455	3.0379864979056e-07\\
};
\addlegendentry{${b}_{1}$}
\addplot [color=mycolor1, dashed, mark=square, mark options={solid, mycolor1}]
  table[row sep=crcr]{%
47	0.000477170989092612\\
65	7.30195089694852e-05\\
91	1.74312241689178e-05\\
129	1.15348851836838e-06\\
183	1.09669387238221e-06\\
257	1.23927456946454e-06\\
365	1.12965145337993e-06\\
515	1.17558922130279e-06\\
727	1.21367466859985e-06\\
1029	1.19067263182714e-06\\
1455	1.19218143640467e-06\\
};
\addlegendentry{${b}_{2}$}
\node at (rel axis cs:0,1) [anchor=north west] {(d)};
\end{axis}
\end{tikzpicture}%

%% file: Graphs/pert_5sol_b.tex
\definecolor{mycolor1}{rgb}{0,0,0}%
\definecolor{mycolor2}{rgb}{0.00000,1.00000,0.40000}%
\definecolor{mycolor3}{rgb}{0.00000,0.40000,1.00000}%
\definecolor{mycolor4}{rgb}{0.80000,0.00000,1.00000}%
\begin{tikzpicture}
\begin{axis}[%
width=0.951\figurewidth,
height=\figureheight,
at={(0\figurewidth,0\figureheight)},
scale only axis,
xmode=log,
xmin=100,
xmax=2000,
xminorticks=true,
xlabel style={font=\color{white!15!black}},
xlabel={Number of time samples, $M$},
ymode=log,
ymin=1e-08,
ymax=1,
ytick={1,1e-1,1e-2,1e-3,1e-4,1e-5,1e-6,1e-7,1e-8},
yticklabels={0,-1,-2,-3,-4,-5,-6,-7,-8},
yminorticks=true,
ylabel style={font=\color{white!15!black}},
axis background/.style={fill=white},
axis x line*=bottom,
axis y line*=left,
legend style={inner sep=0cm, legend cell align=left, align=left, draw=white!15!black, /tikz/every even column/.append style={column sep=1em}}
]
\addlegendimage{color=gray, densely dotted, semithick, mark=diamond*, mark options={scale=0.5, solid}}
\addlegendentry{(filled markers) FB-CFQM}
\addlegendimage{color=gray, dashed, mark=diamond, mark options={solid}}
\addlegendentry{(empty markers) FC}
\addplot[color=red, densely dotted, semithick, mark=*, mark options={solid, red, scale=0.5, fill=red}, forget plot]
  table[row sep=crcr]{%
115	5.35326619969357e-08\\
161	1.19208840525455e-07\\
229	3.0682683372705e-08\\
323	2.3011439277317e-08\\
455	2.39707389710291e-08\\
643	2.42002815377061e-08\\
909	2.4355957076854e-08\\
1285	2.45471129016276e-08\\
1817	2.46105099364817e-08\\
};
%\addlegendentry{$\text{b}_\text{1}\text{, FB}$}
\addplot [color=red, dashed, mark=o, mark options={solid, red}, forget plot]
  table[row sep=crcr]{%
115	0.0019326235980697\\
161	0.00174133958881484\\
229	0.00169250350773316\\
323	0.00166864067155031\\
455	0.00165678449781195\\
643	0.00165087951103527\\
909	0.00164793423537532\\
1285	0.00164646392301568\\
1817	0.00164572956414368\\
};
%\addlegendentry{$\text{b}_\text{1}\text{, FC}$}
\addplot [color=mycolor1, densely dotted, semithick, mark=square*, mark options={solid, mycolor1, scale=0.5, fill=mycolor1}, forget plot]
  table[row sep=crcr]{%
115	1.72442228040313e-06\\
161	3.55851277447343e-08\\
229	2.51646874100586e-08\\
323	2.89244254333119e-08\\
455	3.17474635892505e-08\\
643	3.22411387622951e-08\\
909	3.24716074203443e-08\\
1285	3.27291447589925e-08\\
1817	3.28136968783263e-08\\
};
%\addlegendentry{$\text{b}_\text{2}\text{, FB}$}
\addplot [color=mycolor1, dashed, mark=square, mark options={solid, mycolor1}, forget plot]
  table[row sep=crcr]{%
115	0.000223980796764198\\
161	0.000197510246477131\\
229	0.00019234810435937\\
323	0.000189883958592002\\
455	0.0001886611965986\\
643	0.000188052540345358\\
909	0.000187749016118661\\
1285	0.000187597475826423\\
1817	0.000187521880157923\\
};
%\addlegendentry{$\text{b}_\text{2}\text{, FC}$}
\addplot [color=mycolor2, densely dotted, semithick, mark=triangle*, mark options={rotate=90, solid, mycolor2, scale=0.5, fill=mycolor2}, forget plot]
  table[row sep=crcr]{%
115	2.83833494042352e-06\\
161	2.69436679792851e-07\\
229	4.43568771539082e-08\\
323	4.25493850465994e-08\\
455	4.74957977071167e-08\\
643	4.83461484789418e-08\\
909	4.8705122950413e-08\\
1285	4.90926604170969e-08\\
1817	4.92195393647964e-08\\
};
%\addlegendentry{$\text{b}_\text{3}\text{, FB}$}
\addplot [color=mycolor2, dashed, mark=triangle, mark options={rotate=90, solid, mycolor2}, forget plot]
  table[row sep=crcr]{%
115	2.25475314712247e-05\\
161	5.16069209115443e-06\\
229	4.92290433211192e-06\\
323	5.00052141274566e-06\\
455	5.19226855868797e-06\\
643	5.22859221373117e-06\\
909	5.24810259314633e-06\\
1285	5.27242039183491e-06\\
1817	5.28077140088835e-06\\
};
%\addlegendentry{$\text{b}_\text{3}\text{, FC}$}
\addplot [color=mycolor3, densely dotted, semithick, mark=triangle*, mark options={solid, rotate=180,  scale=0.5, fill=mycolor3, mycolor3}, forget plot]
  table[row sep=crcr]{%
115	1.79778491110697e-06\\
161	2.78620110451416e-07\\
229	9.61665007535175e-08\\
323	8.90188710107038e-08\\
455	9.5497001215586e-08\\
643	9.67284625789311e-08\\
909	9.73854454483944e-08\\
1285	9.81518755037692e-08\\
1817	9.84040222386493e-08\\
};
%\addlegendentry{$\text{b}_\text{4}\text{, FB}$}
\addplot [color=mycolor3, dashed, mark=triangle, mark options={solid, rotate=180, mycolor3}, forget plot]
  table[row sep=crcr]{%
115	0.000192586405697579\\
161	0.000223370754149449\\
229	0.000212979295370764\\
323	0.000222643939968017\\
455	0.00023526490602523\\
643	0.000238793368437134\\
909	0.000240911497022358\\
1285	0.000242811137119065\\
1817	0.000243760322292562\\
};
%\addlegendentry{$\text{b}_\text{4}\text{, FC}$}
\addplot [color=mycolor4, densely dotted, semithick, mark=pentagon*, mark options={solid, mycolor4, scale=0.5, fill=mycolor4}, forget plot]
  table[row sep=crcr]{%
115	9.46558729362405e-07\\
161	7.14632729490502e-07\\
229	6.36353367506323e-07\\
323	6.69212097700196e-07\\
455	7.22197005074417e-07\\
643	7.31699471096378e-07\\
909	7.36427348681543e-07\\
1285	7.41843113838409e-07\\
1817	7.43621285582528e-07\\
};
%\addlegendentry{$\text{b}_\text{5}\text{, FB}$}
\addplot [color=mycolor4, dashed, mark=pentagon, mark options={solid, mycolor4}, forget plot]
  table[row sep=crcr]{%
115	2.33184256850174e-06\\
161	1.37638692761197e-06\\
229	1.20345910017011e-06\\
323	1.25185768381643e-06\\
455	1.34093564171689e-06\\
643	1.35241456341434e-06\\
909	1.35701842130217e-06\\
1285	1.3641361886223e-06\\
1817	1.36546177183436e-06\\
};
%\addlegendentry{$\text{b}_\text{5}\text{, FC}$}
\node at (rel axis cs:0, 1) [anchor=north west] {(e)};
\end{axis}
%\begin{axis}[%
%width=1.227\figurewidth,
%height=1.227\figureheight,
%at={(-0.16\figurewidth,-0.135\figureheight)},
%scale only axis,
%xmin=0,
%xmax=1,
%ymin=0,
%ymax=1,
%axis line style={draw=none},
%ticks=none,
%axis x line*=bottom,
%axis y line*=left,
%legend style={legend cell align=left, align=left, draw=white!15!black}
%]
%\end{axis}
\end{tikzpicture}%

%% file: Graphs/pert_sech_b.tex
\definecolor{mycolor1}{rgb}{0,0,0}%
\begin{tikzpicture}
\begin{axis}[%
width=0.951\figurewidth,
height=\figureheight,
at={(0\figurewidth,0\figureheight)},
scale only axis,
xmode=log,
xmin=33,
xmax=1025,
xminorticks=true,
xlabel style={font=\color{white!15!black}},
xlabel={Number of time samples, $M$},
ymode=log,
ymin=1e-16,
ymax=0.01,
ytick={1,1e-1,1e-2,1e-3,1e-4,1e-5,1e-6,1e-7,1e-8,1e-9,1e-10,1e-11,1e-12,1e-13,1e-14,1e-15,1e-16},
yticklabels={0,,-2,,-4,,-6,,-8,,-10,,-12,,-14,,-16},
yminorticks=true,
axis background/.style={fill=white},
axis x line*=bottom,
axis y line*=left,
legend columns=3,
legend style={at={(1, 1)}, anchor=north east, inner sep=0cm,legend cell align=right, align=right, draw=white!15!black, /tikz/every even column/.append style={column sep=0.8em}}
]
\addlegendimage{empty legend}
\addlegendentry{FB-CFQM:}
\addplot [color=red, densely dotted, semithick, mark=*, mark options={solid, scale=0.5, fill=red, red}]
  table[row sep=crcr]{%
33	1.33351203773871e-15\\
47	8.88186412139027e-16\\
65	5.55431636553075e-16\\
91	6.80829046444916e-16\\
129	2.22052905207049e-16\\
181	6.67374741623461e-16\\
257	2.31743848828119e-16\\
363	1.11023863241579e-15\\
513	2.23517623151686e-16\\
725	1.88758935606809e-15\\
1025	7.7732800378447e-16\\
};
\addlegendentry{$b_1$}
\addplot [color=mycolor1, densely dotted, semithick, mark=square*, mark options={solid, scale=0.5, fill=mycolor1, mycolor1}]
  table[row sep=crcr]{%
33	1.51620434199308e-11\\
47	1.9243795828751e-11\\
65	4.05281553347096e-11\\
91	9.33280125221742e-12\\
129	2.32266547697754e-11\\
181	2.7929016350626e-11\\
257	1.18676130101818e-10\\
363	1.65007894152778e-10\\
513	3.29945145650001e-10\\
725	1.1819014360427e-10\\
1025	3.9336614402315e-10\\
};
\addlegendentry{$b_2$}
\addlegendimage{empty legend}
\addlegendentry{FC:}
\addplot [color=red, dashed, mark=o, mark options={solid, red}]
table[row sep=crcr]{%
33	2.109284473865e-15\\
47	7.34517394901762e-15\\
65	1.87880416528438e-15\\
91	5.33144810062401e-15\\
129	5.44010942867064e-15\\
181	1.11028752839981e-15\\
257	1.33367209752635e-15\\
363	3.10865868513853e-15\\
513	2.6659242590061e-15\\
725	2.55404545618455e-15\\
1025	1.3322737049889e-15\\
};
\addlegendentry{$b_1$}
\addplot [color=mycolor1, dashed, mark=square, mark options={solid, mycolor1}]
  table[row sep=crcr]{%
33	0.00128779297317172\\
47	1.83409281177632e-13\\
65	7.35914544071817e-13\\
91	5.62256869110772e-13\\
129	7.52129598986895e-13\\
181	2.77308485608164e-12\\
257	2.80007796721955e-12\\
363	1.01274253962913e-12\\
513	9.20618461836564e-12\\
725	9.98968235456082e-12\\
1025	1.77765071895043e-11\\
};
\addlegendentry{$b_2$}
\node at (rel axis cs:0, 1) [anchor=north west] {(f)};
\end{axis}
\end{tikzpicture}%

%% file: Graphs/noise_2sol_phi_l.tex
\definecolor{mycolor1}{rgb}{0,0,0}%
\begin{tikzpicture}
\begin{axis}[%
width=0.951\figurewidth,
height=\figureheight,
at={(0\figurewidth,0\figureheight)},
scale only axis,
xmin=0,
xmax=1,
ymode=log,
ymin=1e-05,
ymax=0.01,
yminorticks=true,
% mark repeat=2,
ylabel style={font=\color{white!15!black}},
ylabel={$\log_{10}|\textrm{variance}|$},
ytick={1e-2,1e-3,1e-4,1e-5},
yticklabels={-2,-3,-4,-5},
axis background/.style={fill=white},
title style={font=\bfseries},
title={2-soliton},
axis x line*=bottom,
axis y line*=left,
legend style={at={(0.5, 1)}, anchor=north, inner sep=0, legend cell align=left, align=left, draw=white!15!black}
]
\addlegendimage{empty legend}
\addlegendentry{\tikz\draw[color=red, scale=0.5, fill=red] plot[mark=*] (0,0); \tikz\draw[color=red, fill=red] plot[mark=o] (0,0); $\sigma_{\lambda_1}^2$ \quad \tikz\draw[color=black, scale=0.5, fill=black] plot[mark=square*] (0,0); \tikz\draw[color=black, fill=black] plot[mark=square] (0,0); $\sigma_{\lambda_2}^2$}
\addlegendimage{empty legend}
\addlegendentry{\tikz\draw[color=black!50!green, scale=0.5, fill=black!50!green] plot[mark=triangle*] (0,0); \tikz\draw[color=black!50!green, fill=black!50!green] plot[mark=triangle] (0,0);$|\Real{\sigma_{\lambda_1\lambda_2}}|$}
%\addlegendimage{dotted, gray, mark=diamond*, mark options={solid, gray, fill=gray}}
%\addlegendentry{FB-CFQM}
%\addlegendimage{dashed, gray, mark=diamond, mark options={solid, gray}}
%\addlegendentry{FC}
%\addlegendimage{solid, gray}
%\addlegendentry{Analytic}
\addplot [color=red, only marks, semithick, mark=*, mark options={solid, red, fill=red, scale=0.5}, forget plot]
  table[row sep=crcr]{%
  	0	0.00232813688408916\\
  	0.0625	0.00219538005388133\\
  	0.125	0.000781603357781872\\
  	0.1875	0.000430410692497072\\
  	0.25	0.000251796715926442\\
  	0.3125	0.000176511177224401\\
  	0.375	0.000148472621012486\\
  	0.4375	0.000124403374618363\\
  	0.5	0.000125943254787086\\
  	0.5625	0.000123724565758504\\
  	0.625	0.000149238036367302\\
  	0.6875	0.000183797884763466\\
  	0.75	0.000259473488327077\\
  	0.8125	0.000445988112782129\\
  	0.875	0.000768773994547669\\
  	0.9375	0.00156746808945373\\
};
%\addlegendentry{$\text{Var(}\lambda{}_{\text{1}}\text{) (FB)}$}
\addplot [color=red, only marks, mark=o, mark options={solid, red}, forget plot]
  table[row sep=crcr]{%
  	0	0.00200063044105978\\
  	0.0625	0.00148430461752533\\
  	0.125	0.000795736763577642\\
  	0.1875	0.000431095211479906\\
  	0.25	0.000263349908313051\\
  	0.3125	0.000182940113759554\\
  	0.375	0.000142722807549001\\
  	0.4375	0.000123572075044135\\
  	0.5	0.000117872287133043\\
  	0.5625	0.000123572075044136\\
  	0.625	0.000142722807549003\\
  	0.6875	0.000182940113759558\\
  	0.75	0.000263349908313059\\
  	0.8125	0.000431095211479922\\
  	0.875	0.000795736763577672\\
  	0.9375	0.00148430461752537\\
};
%\addlegendentry{$\text{Var(}\lambda{}_{\text{1}}\text{) (FC)}$}
\addplot [color=red, solid, forget plot]
  table[row sep=crcr]{%
0	0.00200063044105978\\
0.0625	0.00148430461752533\\
0.125	0.000795736763577642\\
0.1875	0.000431095211479906\\
0.25	0.000263349908313051\\
0.3125	0.000182940113759554\\
0.375	0.000142722807549001\\
0.4375	0.000123572075044135\\
0.5	0.000117872287133043\\
0.5625	0.000123572075044136\\
0.625	0.000142722807549003\\
0.6875	0.000182940113759558\\
0.75	0.000263349908313059\\
0.8125	0.000431095211479922\\
0.875	0.000795736763577672\\
0.9375	0.00148430461752537\\
};
%\addlegendentry{$\text{Var(}\lambda{}_{\text{1}}\text{) (analytic)}$}
\addplot [color=red, densely dotted, thick, forget plot]
  table[row sep=crcr]{%
0	0.00200063044105973\\
0.0625	0.00148430461752531\\
0.125	0.00079573676357763\\
0.1875	0.000431095211479892\\
0.25	0.000263349908313037\\
0.3125	0.000182940113759539\\
0.375	0.000142722807548986\\
0.4375	0.00012357207504412\\
0.5	0.000117872287133028\\
0.5625	0.000123572075044119\\
0.625	0.000142722807548985\\
0.6875	0.000182940113759539\\
0.75	0.000263349908313037\\
0.8125	0.000431095211479892\\
0.875	0.00079573676357763\\
0.9375	0.00148430461752531\\
};
% \addlegendentry{$\text{Var(}\lambda{}_{\text{1}}\text{) (TD)}$}
\addplot [color=mycolor1, only marks, semithick, mark=square*, mark options={solid, mycolor1, scale=0.5, fill=mycolor1}, forget plot]
  table[row sep=crcr]{%
  	0	0.000890301557938004\\
  	0.0625	0.00067572022935127\\
  	0.125	0.000440005252476749\\
  	0.1875	0.000333381053369635\\
  	0.25	0.00028335301498718\\
  	0.3125	0.000248708322535408\\
  	0.375	0.000244172676393902\\
  	0.4375	0.000230176911407246\\
  	0.5	0.000227894341792359\\
  	0.5625	0.000228105800610826\\
  	0.625	0.000234490608174257\\
  	0.6875	0.000268515210151512\\
  	0.75	0.00028899782554203\\
  	0.8125	0.000343568211229577\\
  	0.875	0.000443152055736591\\
  	0.9375	0.000679582457565236\\
};
%\addlegendentry{$\text{Var(}\lambda{}_{\text{2}}\text{) (FB)}$}
\addplot [color=mycolor1, only marks, mark=square, mark options={solid, mycolor1}, forget plot]
  table[row sep=crcr]{%
  	0	0.000806040630801486\\
  	0.0625	0.0006786897236676\\
  	0.125	0.000442044608636552\\
  	0.1875	0.000335010876458486\\
  	0.25	0.000284811013004317\\
  	0.3125	0.000250017596539798\\
  	0.375	0.000245480375422989\\
  	0.4375	0.000231401575480581\\
  	0.5	0.000229139499270074\\
  	0.5625	0.00022935990943432\\
  	0.625	0.000235740380333672\\
  	0.6875	0.000269937900433789\\
  	0.75	0.00029051204895954\\
  	0.8125	0.000345238301911486\\
  	0.875	0.000445264950694586\\
  	0.9375	0.00068253424019242\\
};
%\addlegendentry{$\text{Var(}\lambda{}_{\text{2}}\text{) (FC)}$}
\addplot [color=mycolor1, forget plot]
  table[row sep=crcr]{%
0	0.000774813110433148\\
0.0625	0.000624106341775922\\
0.125	0.000436494017231507\\
0.1875	0.000338086241302552\\
0.25	0.000286322822645095\\
0.3125	0.00025581882423653\\
0.375	0.000237341811017754\\
0.4375	0.000227256939515027\\
0.5	0.000224039490530124\\
0.5625	0.000227256939515025\\
0.625	0.000237341811017751\\
0.6875	0.000255818824236524\\
0.75	0.000286322822645086\\
0.8125	0.00033808624130254\\
0.875	0.000436494017231491\\
0.9375	0.000624106341775905\\
};
%\addlegendentry{$\text{Var(}\lambda{}_{\text{2}}\text{) (analytic)}$}
\addplot [color=mycolor1, densely dotted, thick, forget plot]
  table[row sep=crcr]{%
0	0.000774813109303277\\
0.0625	0.000624106340789158\\
0.125	0.000436494016424004\\
0.1875	0.000338086240564867\\
0.25	0.000286322821916639\\
0.3125	0.000255818823497758\\
0.375	0.000237341810266079\\
0.4375	0.000227256938754145\\
0.5	0.000224039489765986\\
0.5625	0.000227256938754145\\
0.625	0.000237341810266079\\
0.6875	0.000255818823497758\\
0.75	0.000286322821916639\\
0.8125	0.000338086240564868\\
0.875	0.000436494016424004\\
0.9375	0.000624106340789158\\
};
% \addlegendentry{$\text{Var(}\lambda{}_{\text{2}}\text{) (TD)}$}
\addplot [color=black!50!green, only marks, semithick, mark=triangle*, mark options={solid, black!50!green, scale=0.5, fill=black!50!green}, forget plot]
  table[row sep=crcr]{%
  	0	0.00107820811112655\\
  	0.0625	0.000809373957065175\\
  	0.125	0.000326524821356953\\
  	0.1875	0.000118226405977023\\
  	0.25	2.33218168824316e-05\\
  	0.3125	3.76584807724159e-05\\
  	0.375	7.84732198624161e-05\\
  	0.4375	9.09362339359871e-05\\
  	0.5	9.95823529553728e-05\\
  	0.5625	8.77800108686763e-05\\
  	0.625	7.76985181724266e-05\\
  	0.6875	4.4846835424497e-05\\
  	0.75	1.33074254686548e-05\\
  	0.8125	0.000129405128990606\\
  	0.875	0.000331885024826821\\
  	0.9375	0.000810389383519128\\
};
%\addlegendentry{$\sigma{}_{\lambda{}_\text{1}\lambda{}_\text{2}}\text{(FB)}$}
\addplot [color=black!50!green, only marks, mark=triangle, mark options={solid, black!50!green}, forget plot]
  table[row sep=crcr]{%
  	0	0.00108700301144324\\
  	0.0625	0.000788058605530158\\
  	0.125	0.000327171622509215\\
  	0.1875	0.00011844624620568\\
  	0.25	2.3330057326672e-05\\
  	0.3125	3.77897407684273e-05\\
  	0.375	7.87103568600933e-05\\
  	0.4375	9.11896865820161e-05\\
  	0.5	9.98696789836758e-05\\
  	0.5625	8.80343721689459e-05\\
  	0.625	7.79302364441283e-05\\
  	0.6875	4.50070471168292e-05\\
  	0.75	1.32925091734866e-05\\
  	0.8125	0.000129643455083546\\
  	0.875	0.000332555511890418\\
  	0.9375	0.000812042246995242\\
};
%\addlegendentry{$\sigma{}_{\lambda{}_\text{1}\lambda{}_\text{2}}\text{(FC)}$}
\addplot [color=black!50!green, forget plot]
  table[row sep=crcr]{%
0	0.00106653408592105\\
0.0625	0.000749214696321237\\
0.125	0.000335165880190363\\
0.1875	0.000116253056815642\\
0.25	1.08102375188742e-05\\
0.3125	4.37984725603799e-05\\
0.375	7.33906733748082e-05\\
0.4375	8.83878954135466e-05\\
0.5	9.30031140123058e-05\\
0.5625	8.8387895413547e-05\\
0.625	7.33906733748093e-05\\
0.6875	4.37984725603819e-05\\
0.75	1.08102375188706e-05\\
0.8125	0.000116253056815636\\
0.875	0.000335165880190352\\
0.9375	0.000749214696321224\\
};
%\addlegendentry{$\sigma{}_{\lambda{}_\text{1}\lambda{}_\text{2}}\text{(analytic)}$}
\node at (rel axis cs:0, 1) [anchor=north west] {(a)};
\end{axis}
\draw [dashed, thick, black] (\figurewidth, -1.5em) -- (\figurewidth, \figureheight+1.8em);
\end{tikzpicture}%

%% file: Graphs/noise_5sol_l.tex
\definecolor{mycolor1}{rgb}{0,0,0}%
\definecolor{mycolor2}{rgb}{0.00000,1.00000,0.40000}%
\definecolor{mycolor3}{rgb}{0.00000,0.40000,1.00000}%
\definecolor{mycolor4}{rgb}{0.80000,0.00000,1.00000}%
\begin{tikzpicture}
\begin{axis}[%
width=0.951\figurewidth,
height=\figureheight,
at={(0\figurewidth,0\figureheight)},
scale only axis,
xmin=2,
xmax=20,
xtick=data,
ymode=log,
ymin=1e-05,
ymax=0.1,
ytick={1e-1,1e-2,1e-3,1e-4,1e-5},
yticklabels={-1,-2,-3,-4,-5},
yminorticks=true,
ylabel style={font=\color{white!15!black}},
axis background/.style={fill=white},
title style={font=\bfseries},
title={5-soliton},
axis x line*=bottom,
axis y line*=left,
legend style={legend cell align=left, align=left, draw=white!15!black}
]
\addplot [color=red, only marks, semithick, mark=*, mark options={solid, red, scale=0.5, fill=red}, forget plot]
  table[row sep=crcr]{%
  	2	0.00219383535104397\\
  	5	0.00103030984146614\\
  	8	0.000517649002046992\\
  	11	0.000274568522914567\\
  	14	0.000134917839104768\\
  	17	6.93126175670431e-05\\
};
%\addlegendentry{$\text{Var(}\lambda{}_{\text{1}}\text{) (FB)}$}
\addplot [color=red, only marks, mark=o, mark options={solid, red}, forget plot]
  table[row sep=crcr]{%
  	2	0.00219380238072488\\
  	5	0.00103029197252361\\
  	8	0.000517645351369968\\
  	11	0.00027456809468682\\
  	14	0.000134917516768258\\
  	17	6.93126200491455e-05\\
};
%\addlegendentry{$\text{Var(}\lambda{}_{\text{1}}\text{) (FC)}$}
\addplot [color=red, forget plot]
  table[row sep=crcr]{%
2	0.00211261504117495\\
5	0.00105881568820584\\
8	0.000530664905693041\\
11	0.000265962476067372\\
14	0.000133296997628866\\
17	6.68067534924325e-05\\
};
%\addlegendentry{$\text{Var(}\lambda{}_{\text{1}}\text{) (analytic)}$}
\addplot [color=red, densely dotted, thick, forget plot]
  table[row sep=crcr]{%
2	0.00211261504117496\\
5	0.00105881568820584\\
8	0.000530664905693043\\
11	0.000265962476067374\\
14	0.000133296997628867\\
17	6.68067534924327e-05\\
};
% \addlegendentry{$\text{Var(}\lambda{}_{\text{1}}\text{) (TD)}$}
\addplot [color=mycolor1, only marks, semithick, mark=square*, mark options={solid, mycolor1, scale=0.5, fill=mycolor1}, forget plot]
  table[row sep=crcr]{%
  	2	0.00175792410746798\\
  	5	0.000914376297834305\\
  	8	0.000447925623181924\\
  	11	0.000228001968608024\\
  	14	0.000109430311465986\\
  	17	5.83223260446686e-05\\
};
%\addlegendentry{$\text{Var(}\lambda{}_{\text{2}}\text{) (FB)}$}
\addplot [color=mycolor1, only marks, mark=square, mark options={solid, mycolor1}, forget plot]
  table[row sep=crcr]{%
  	2	0.0017578815247033\\
  	5	0.000914358377188519\\
  	8	0.00044792195888588\\
  	11	0.000228002720254426\\
  	14	0.000109430183019085\\
  	17	5.83223615905794e-05\\
};
%\addlegendentry{$\text{Var(}\lambda{}_{\text{2}}\text{) (FC)}$}
\addplot [color=mycolor1, forget plot]
  table[row sep=crcr]{%
2	0.00178490029955749\\
5	0.000894569243435709\\
8	0.000448346684405585\\
11	0.000224705634463195\\
14	0.00011261959531707\\
17	5.6443503429185e-05\\
};
%\addlegendentry{$\text{Var(}\lambda{}_{\text{2}}\text{) (analytic)}$}
\addplot [color=mycolor1, densely dotted, thick, forget plot]
  table[row sep=crcr]{%
2	0.00178490029955749\\
5	0.00089456924343571\\
8	0.000448346684405586\\
11	0.000224705634463195\\
14	0.00011261959531707\\
17	5.64435034291852e-05\\
};
% \addlegendentry{$\text{Var(}\lambda{}_{\text{2}}\text{) (TD)}$}
\addplot [color=mycolor2, only marks, semithick, mark=triangle*, mark options={rotate=90, solid, mycolor2, scale=0.5, fill=mycolor2}, forget plot]
  table[row sep=crcr]{%
  	2	0.0044328071306357\\
  	5	0.00219303204575186\\
  	8	0.00113495060119123\\
  	11	0.000572396410169417\\
  	14	0.000281641746225928\\
  	17	0.000135746439404091\\
};
%\addlegendentry{$\text{Var(}\lambda{}_{\text{3}}\text{) (FB)}$}
\addplot [color=mycolor2, only marks, mark=triangle, mark options={rotate=90, solid, mycolor2}, forget plot]
  table[row sep=crcr]{%
  	2	0.00461947086776083\\
  	5	0.00219278100226581\\
  	8	0.00113487840240364\\
  	11	0.000572376709414568\\
  	14	0.000281637311723822\\
  	17	0.000135744838849242\\
};
%\addlegendentry{$\text{Var(}\lambda{}_{\text{3}}\text{) (FC)}$}
\addplot [color=mycolor2, forget plot]
  table[row sep=crcr]{%
2	0.00447471356739239\\
5	0.00224266931411581\\
8	0.00112399722948248\\
11	0.000563333062049041\\
14	0.000282335338979139\\
17	0.000141502867498173\\
};
%\addlegendentry{$\text{Var(}\lambda{}_{\text{3}}\text{) (analytic)}$}
\addplot [color=mycolor2, densely dotted, thick, forget plot]
  table[row sep=crcr]{%
2	0.0044747135673924\\
5	0.00224266931411582\\
8	0.00112399722948248\\
11	0.000563333062049042\\
14	0.00028233533897914\\
17	0.000141502867498173\\
};
% \addlegendentry{$\text{Var(}\lambda{}_{\text{3}}\text{) (TD)}$}
\addplot [color=mycolor3, only marks, semithick, mark=triangle*, mark options={solid, rotate=180,  scale=0.5, fill=mycolor3, mycolor3}, forget plot]
  table[row sep=crcr]{%
  	2	0.0212607552533909\\
  	5	0.0110423897378609\\
  	8	0.0057717899090857\\
  	11	0.00303515852879816\\
  	14	0.00149466446589393\\
  	17	0.000746967631965374\\
};
%\addlegendentry{$\text{Var(}\lambda{}_{\text{4}}\text{) (FB)}$}
\addplot [color=mycolor3, only marks, mark=triangle, mark options={solid, rotate=180, mycolor3}, forget plot]
  table[row sep=crcr]{%
  	2	0.0191639061436051\\
  	5	0.010627677677\\
  	8	0.00578415542175085\\
  	11	0.00303542326348436\\
  	14	0.00149473676823113\\
  	17	0.000746988795167148\\
};
%\addlegendentry{$\text{Var(}\lambda{}_{\text{4}}\text{) (FC)}$}
\addplot [color=mycolor3, forget plot]
  table[row sep=crcr]{%
2	0.0233141428490913\\
5	0.0116847507589271\\
8	0.00585624790849086\\
11	0.00293507668869203\\
14	0.00147102296608945\\
17	0.000737257930976558\\
};
%\addlegendentry{$\text{Var(}\lambda{}_{\text{4}}\text{) (analytic)}$}
\addplot [color=mycolor3, densely dotted, thick, forget plot]
  table[row sep=crcr]{%
2	0.02331414284909\\
5	0.0116847507589265\\
8	0.00585624790849054\\
11	0.00293507668869187\\
14	0.00147102296608937\\
17	0.000737257930976518\\
};
% \addlegendentry{$\text{Var(}\lambda{}_{\text{4}}\text{) (TD)}$}
\addplot [color=mycolor4, only marks, semithick, mark=pentagon*, mark options={solid, mycolor4, scale=0.5, fill=mycolor4}, forget plot]
  table[row sep=crcr]{%
  	2	0.0134877280318258\\
  	5	0.00676276002954651\\
  	8	0.00345291417873755\\
  	11	0.00177549758155647\\
  	14	0.000859253014336932\\
  	17	0.000419994086122254\\
};
%\addlegendentry{$\text{Var(}\lambda{}_{\text{5}}\text{) (FB)}$}
\addplot [color=mycolor4, only marks, mark=pentagon, mark options={solid, mycolor4}, forget plot]
  table[row sep=crcr]{%
  	2	0.0109663901022737\\
  	5	0.00643683557883025\\
  	8	0.0034023391194455\\
  	11	0.00174685435054563\\
  	14	0.000845232887338644\\
  	17	0.000412714038592159\\
};
%\addlegendentry{$\text{Var(}\lambda{}_{\text{5}}\text{) (FC)}$}
\addplot [color=mycolor4, forget plot]
  table[row sep=crcr]{%
2	0.0128620690526036\\
5	0.00644630480719735\\
8	0.00323080567343743\\
11	0.0016192385578574\\
14	0.000811541693395165\\
17	0.000406734336285915\\
};
%\addlegendentry{$\text{Var(}\lambda{}_{\text{5}}\text{) (analytic)}$}
\addplot [color=mycolor4, densely dotted, thick, forget plot]
  table[row sep=crcr]{%
2	0.0128620690428864\\
5	0.00644630480232719\\
8	0.00323080567099657\\
11	0.00161923855663407\\
14	0.000811541692782048\\
17	0.000406734335978628\\
};
% \addlegendentry{$\text{Var(}\lambda{}_{\text{5}}\text{) (TD)}$}
\node at (rel axis cs:0.5, 1) [anchor=north] {(b)};
\node at (17, 6.68067534924325e-05) [anchor=west] {\footnotesize ${\lambda_1}$};
\node at (17, 3.6443503429185e-05) [anchor=west] {\footnotesize ${\lambda_2}$};
\node at (17, 0.000141502867498173) [anchor=west] {\footnotesize ${\lambda_3}$};
\node (top) at (17, 0.000837257930976558) [anchor=west] {\footnotesize ${\lambda_4}$};
\node at (17, 0.000406734336285915) [anchor=west] {\footnotesize ${\lambda_5}$};
\node at (18, 0.002) [anchor=south, align=center] {Var\\of};
\end{axis}
\end{tikzpicture}%

%% file: Graphs/noise_sech_l.tex
\definecolor{mycolor1}{rgb}{0,0,0}%
\begin{tikzpicture}
\begin{axis}[%
width=0.951\figurewidth,
height=\figureheight,
at={(0\figurewidth,0\figureheight)},
scale only axis,
xmin=2,
xmax=17,
xtick=data,
ymode=log,
ymin=7e-5,
ymax=0.04,
ytick={1e-1,1e-2,1e-3,1e-4},
yticklabels={-1,-2,-3,-4},
yminorticks=true,
axis background/.style={fill=white},
title style={font=\bfseries},
title={$\mathbf{2.2sech(\boldsymbol{t})}$},
axis x line*=bottom,
axis y line*=left,
%legend columns=2,
%transpose legend,
legend style={font=\footnotesize, at={(1,1)}, anchor=north east, inner sep=0cm, legend cell align=left, align=left, draw=white!15!black}
]
\addlegendimage{empty legend}
\addlegendentry{\tikz\draw[color=red, scale=0.5, fill=red] plot[mark=*] (0,0); \tikz\draw[color=red, fill=red] plot[mark=o] (0,0); $\sigma_{\lambda_1}^2$ \quad \tikz\draw[color=black, scale=0.5, fill=black] plot[mark=square*] (0,0); \tikz\draw[color=black, fill=black] plot[mark=square] (0,0); $\sigma_{\lambda_2}^2$}
\addlegendimage{empty legend}
\addlegendentry{\tikz\draw[color=black!50!green, scale=0.5, fill=black!50!green] plot[mark=triangle*] (0,0); \tikz\draw[color=black!50!green, fill=black!50!green] plot[mark=triangle] (0,0);$|\Real{\sigma_{\lambda_1\lambda_2}}|$}
%\addlegendimage{mark=*, only marks, color=red}
%\addlegendentry{$\sigma_{\lambda_1}^2$}
%\addlegendimage{mark=square*, only marks, color=mycolor1}
%\addlegendentry{$\sigma_{\lambda_2}^2$}
%\addlegendimage{mark=o, only marks, color=black!50!green}
\addplot [color=red, only marks, semithick, mark=*, mark options={solid, red, scale=0.5, fill=red}, forget plot]
table[row sep=crcr]{%
2	0.0127992218977028\\
5	0.00326972217135292\\
8	0.00169781477361194\\
11	0.000786159875224445\\
14	0.000402623452138206\\
17	0.000198057941029054\\
};
%\addlegendentry{$\text{Var(}\lambda{}_{\text{1}}\text{) (FB)}$}
\addplot [color=red, only marks, mark=o, mark options={solid, red}, forget plot]
table[row sep=crcr]{%
2	0.00714962283825512\\
5	0.0032696723886619\\
8	0.00169777527387226\\
11	0.000786156232489749\\
14	0.000402623633841531\\
17	0.000198058906020957\\
};
%\addlegendentry{$\text{Var(}\lambda{}_{\text{1}}\text{) (FC)}$}
\addplot [color=red, forget plot]
table[row sep=crcr]{%
2	0.00647891694764261\\
5	0.00324715046188985\\
8	0.00162743035716609\\
11	0.000815647318629118\\
14	0.00040879202323923\\
17	0.000204881343256165\\
};
%\addlegendentry{$\text{Var(}\lambda{}_{\text{1}}\text{) (analytic)}$}
\addplot [color=red, densely dotted, thick, forget plot]
  table[row sep=crcr]{%
2	0.00641977780037397\\
5	0.00321751067627121\\
8	0.00161257527500658\\
11	0.000808202141096285\\
14	0.000405060595307686\\
17	0.000203011199213675\\
};
% \addlegendentry{$\text{Var(}\lambda{}_{\text{1}}\text{) (TD)}$}
\addplot [color=mycolor1, only marks, semithick, mark=square*, mark options={solid, mycolor1, scale=0.5, fill=mycolor1}, forget plot]
table[row sep=crcr]{%
2	0.00948254626919034\\
5	0.00379072925103486\\
8	0.00181877995476209\\
11	0.000905359955866076\\
14	0.000451196345617369\\
17	0.000217231432932982\\
};
%\addlegendentry{$\text{Var(}\lambda{}_{\text{2}}\text{) (FB)}$}
\addplot [color=mycolor1, only marks, mark=square, mark options={solid, mycolor1}, forget plot]
table[row sep=crcr]{%
2	0.00769042830847126\\
5	0.0037907975663317\\
8	0.00181882799280015\\
11	0.000905348425595839\\
14	0.000451202781791426\\
17	0.00021723720075277\\
};
%\addlegendentry{$\text{Var(}\lambda{}_{\text{2}}\text{) (FC)}$}
\addplot [color=mycolor1, forget plot]
table[row sep=crcr]{%
2	0.00720922962609385\\
5	0.00361317385288575\\
8	0.0018108766079422\\
11	0.000907588237574891\\
14	0.000454871638062811\\
17	0.000227975857936206\\
};
%\addlegendentry{$\text{Var(}\lambda{}_{\text{2}}\text{) (analytic)}$}
\addplot [color=mycolor1, densely dotted, thick, forget plot]
  table[row sep=crcr]{%
2	0.00719062254348728\\
5	0.00360384822062829\\
8	0.00180620272010926\\
11	0.000905245744661614\\
14	0.000453697610519815\\
17	0.00022738745031973\\
};
% \addlegendentry{$\text{Var(}\lambda{}_{\text{2}}\text{) (TD)}$}
\addplot [color=black!50!green, only marks, semithick, mark=triangle*, mark options={solid, black!50!green, scale=0.5, fill=black!50!green}, forget plot]
table[row sep=crcr]{%
2	0.00514466892079926\\
5	0.00173073451082142\\
8	0.000828269002119982\\
11	0.000392475843365217\\
14	0.000200793215454568\\
17	9.88886240023721e-05\\
};
%\addlegendentry{$\sigma{}_{\lambda{}_\text{1}\lambda{}_\text{2}}\text{(FB)}$}
\addplot [color=black!50!green, only marks, mark=triangle, mark options={solid, black!50!green}, forget plot]
table[row sep=crcr]{%
2	0.00388708444369168\\
5	0.00173071361656787\\
8	0.000828236465708283\\
11	0.000392468207981683\\
14	0.000200794830508728\\
17	9.88911216814672e-05\\
};
%\addlegendentry{$\sigma{}_{\lambda{}_\text{1}\lambda{}_\text{2}}\text{(FC)}$}
\addplot [color=black!50!green, forget plot]
table[row sep=crcr]{%
2	0.00323948439115551\\
5	0.00162358822038196\\
8	0.00081372168872306\\
11	0.000407826922113623\\
14	0.000204397646892852\\
17	0.000102441491206152\\
};
%\addlegendentry{$\sigma{}_{\lambda{}_\text{1}\lambda{}_\text{2}}\text{(analytic)}$}
\node at (rel axis cs:0, 1) [anchor=north west] {(c)};
%\node at (11, 0.000575858504386041) [anchor=north east] {\footnotesize ${\Real{\sigma_{\lambda_1\lambda_2}}}$};
%\node at (8, 0.00255698789054221) [anchor=south west] {\footnotesize ${\sigma_{\lambda_2}^2}$};
%\node at (17, 0.000209295973953984) [anchor=west] {\footnotesize ${\sigma_{\lambda_1}^2}$};
\end{axis}
\end{tikzpicture}%

%% file: Graphs/noise_2sol_phi_b.tex
\definecolor{mycolor1}{rgb}{0,0,0}%
\begin{tikzpicture}
\begin{axis}[%
width=0.951\figurewidth,
height=\figureheight,
at={(0\figurewidth,0\figureheight)},
scale only axis,
xmin=0,
xmax=1,
xlabel={$\text{Phase difference }\alpha\ (\times 2\pi)$},
ymode=log,
ymin=4e-07,
ymax=0.1,
ytick={1e-2,1e-3,1e-4,1e-5,1e-6,1e-7,1e-8},
yticklabels={-2,-3,-4,-5,-6,-7},
yminorticks=true,
% mark repeat=2,
ylabel style={font=\color{white!15!black}},
yminorticks=true,
ylabel={$\log_{10}|\textrm{variance}|$},
axis background/.style={fill=white},
axis x line*=bottom,
axis y line*=left,
legend style={font=\footnotesize, at={(1,1)}, anchor=north east, inner sep=0cm, legend cell align=left, align=left, draw=white!15!black}
]
\addlegendimage{empty legend}
\addlegendentry{\tikz\draw[color=red, scale=0.5, fill=red] plot[mark=*] (0,0); \tikz\draw[color=red, fill=red] plot[mark=o] (0,0); $\sigma_{b_1}^2$ \quad \tikz\draw[color=black, scale=0.5, fill=black] plot[mark=square*] (0,0); \tikz\draw[color=black, fill=black] plot[mark=square] (0,0); $\sigma_{b_2}^2$}
\addlegendimage{empty legend}
\addlegendentry{\tikz\draw[color=black!50!green, scale=0.5, fill=black!50!green] plot[mark=triangle*] (0,0); \tikz\draw[color=black!50!green, fill=black!50!green] plot[mark=triangle] (0,0);$|\Real{\sigma_{b_1b_2}}|$}
\addplot [color=red, only marks, semithick, mark=*, mark options={solid, red, scale=0.5, fill=red}, forget plot]
  table[row sep=crcr]{%
  	0	0.000321714503774262\\
  	0.0625	0.0005528174852354\\
  	0.125	0.000508239774168385\\
  	0.1875	0.000711009901465572\\
  	0.25	0.00082547876975345\\
  	0.3125	0.000911091309246268\\
  	0.375	0.000999107338766428\\
  	0.4375	0.00100479972744529\\
  	0.5	0.0010744532409461\\
  	0.5625	0.00103732388211612\\
  	0.625	0.000988776254447773\\
  	0.6875	0.000947069336886711\\
  	0.75	0.000839547438351379\\
  	0.8125	0.000768987971119583\\
  	0.875	0.000579099662986391\\
  	0.9375	0.000368756125100425\\
};
%\addlegendentry{$\text{Var(b}_{\text{1}}\text{) (FB)}$}
\addplot [color=red, only marks, mark=o, mark options={solid, red}, forget plot]
  table[row sep=crcr]{%
  	0	0.000224362508724483\\
  	0.0625	0.000345055578194214\\
  	0.125	0.000514797409048204\\
  	0.1875	0.000720345792202261\\
  	0.25	0.000827428810737102\\
  	0.3125	0.000914029605126136\\
  	0.375	0.00100061001865847\\
  	0.4375	0.00100875117270823\\
  	0.5	0.00107567871138402\\
  	0.5625	0.0010392175525756\\
  	0.625	0.000995191617033396\\
  	0.6875	0.000953406837942324\\
  	0.75	0.000847917855249314\\
  	0.8125	0.000783033112458934\\
  	0.875	0.000587508257584584\\
  	0.9375	0.000379884509062692\\
};
%\addlegendentry{$\text{Var(b}_{\text{1}}\text{) (FC)}$}
\addplot [color=red, forget plot]
  table[row sep=crcr]{%
0	0.000225060109610881\\
0.0625	0.000345859708879571\\
0.125	0.000546639745213546\\
0.1875	0.00071466315129981\\
0.25	0.000845053245614057\\
0.3125	0.000942403786244783\\
0.375	0.00100980254717086\\
0.4375	0.00104937678653023\\
0.5	0.00106242258137212\\
0.5625	0.00104937678654567\\
0.625	0.0010098025470985\\
0.6875	0.000942403786345187\\
0.75	0.000845053245736046\\
0.8125	0.000714663151332206\\
0.875	0.000546639745207163\\
0.9375	0.000345859708946146\\
};
%\addlegendentry{$\text{Var(b}_{\text{1}}\text{) (analytic)}$}
\addplot [color=mycolor1, only marks, semithick, mark=square*, mark options={solid, mycolor1, scale=0.5, fill=mycolor1}, forget plot]
  table[row sep=crcr]{%
  	0	0.000680175555794857\\
  	0.0625	0.000790765265425097\\
  	0.125	0.000937206918232257\\
  	0.1875	0.000862582871447443\\
  	0.25	0.000805192839703384\\
  	0.3125	0.000756727178051225\\
  	0.375	0.000702926223413988\\
  	0.4375	0.000708754781270107\\
  	0.5	0.000702580010498437\\
  	0.5625	0.000678928679071013\\
  	0.625	0.000687788879140803\\
  	0.6875	0.000730225349388155\\
  	0.75	0.000804085543178053\\
  	0.8125	0.000848792189668175\\
  	0.875	0.000900146680644247\\
  	0.9375	0.000782006518012735\\
};
%\addlegendentry{$\text{Var(b}_{\text{2}}\text{) (FB)}$}
\addplot [color=mycolor1, only marks, mark=square, mark options={solid, mycolor1}, forget plot]
  table[row sep=crcr]{%
  	0	0.000666964461846124\\
  	0.0625	0.000888097472489087\\
  	0.125	0.0009977446819735\\
  	0.1875	0.000961812521878307\\
  	0.25	0.000893483264007859\\
  	0.3125	0.00087333370190322\\
  	0.375	0.000788296525959756\\
  	0.4375	0.000797672972476397\\
  	0.5	0.000801156690810637\\
  	0.5625	0.000774010292741906\\
  	0.625	0.000770283746066039\\
  	0.6875	0.000831062897236088\\
  	0.75	0.000882057144783453\\
  	0.8125	0.000940002837845056\\
  	0.875	0.000978775822332502\\
  	0.9375	0.000871335867070916\\
};
%\addlegendentry{$\text{Var(b}_{\text{2}}\text{) (FC)}$}
\addplot [color=mycolor1, forget plot]
  table[row sep=crcr]{%
0	0.000653020532721252\\
0.0625	0.000848984419028669\\
0.125	0.000986499032595992\\
0.1875	0.000956750608612213\\
0.25	0.000890932359309426\\
0.3125	0.000833808788277759\\
0.375	0.000793133066915318\\
0.4375	0.000769147409624129\\
0.5	0.000761238552569545\\
0.5625	0.000769147409624202\\
0.625	0.000793133066914319\\
0.6875	0.000833808788278716\\
0.75	0.000890932359310765\\
0.8125	0.000956750608612731\\
0.875	0.000986499032592846\\
0.9375	0.000848984419029735\\
};
%\addlegendentry{$\text{Var(b}_{\text{2}}\text{) (analytic)}$}
\addplot [color=black!50!green, only marks, semithick, mark=triangle*, mark options={solid, black!50!green, scale=0.5, fill=black!50!green}, forget plot]
  table[row sep=crcr]{%
  	0	0.000311575383011376\\
  	0.0625	0.000180288657522964\\
  	0.125	1.19525800711079e-05\\
  	0.1875	4.53621776383344e-05\\
  	0.25	9.83926106131742e-05\\
  	0.3125	2.64778678404364e-05\\
  	0.375	2.61193785983153e-05\\
  	0.4375	1.9659648905052e-05\\
  	0.5	3.25555566199154e-05\\
  	0.5625	1.67498836001505e-05\\
  	0.625	7.16932250266532e-06\\
  	0.6875	8.38599642385017e-05\\
  	0.75	7.27999320495848e-05\\
  	0.8125	3.98809984937682e-05\\
  	0.875	3.72842718007254e-05\\
  	0.9375	0.000184465019877963\\
};
%\addlegendentry{$\sigma{}_{\text{b}_\text{1}\text{b}_\text{2}}\text{(FB)}$}
\addplot [color=black!50!green, only marks, mark=triangle, mark options={solid, black!50!green}, forget plot]
  table[row sep=crcr]{%
  	0	0.000326971269631809\\
  	0.0625	0.000192308211158544\\
  	0.125	1.67800831086687e-05\\
  	0.1875	2.38037156229709e-05\\
  	0.25	8.0927765318537e-05\\
  	0.3125	2.71484770907633e-06\\
  	0.375	3.01314602944582e-05\\
  	0.4375	3.89000638526756e-05\\
  	0.5	4.08375698293962e-05\\
  	0.5625	4.82960650713017e-07\\
  	0.625	1.34331497892294e-05\\
  	0.6875	6.69866673818732e-05\\
  	0.75	5.2825733178489e-05\\
  	0.8125	9.78808672595137e-06\\
  	0.875	5.56538060452893e-05\\
  	0.9375	0.000211421558184086\\
};
%\addlegendentry{$\sigma{}_{\text{b}_\text{1}\text{b}_\text{2}}\text{(FC)}$}
\addplot [color=black!50!green, forget plot]
  table[row sep=crcr]{%
0	0.000335502612843102\\
0.0625	0.000200766596879957\\
0.125	3.97550770819177e-05\\
0.1875	3.10658834767421e-05\\
0.25	4.51127826071062e-05\\
0.3125	2.93370720027184e-05\\
0.375	3.2160118169244e-06\\
0.4375	1.85748240464502e-05\\
0.5	2.68716645106786e-05\\
0.5625	1.85748240416415e-05\\
0.625	3.21601184836698e-06\\
0.6875	2.93370719376258e-05\\
0.75	4.51127825977157e-05\\
0.8125	3.1065883546117e-05\\
0.875	3.97550769606859e-05\\
0.9375	0.00020076659695832\\
};
%\addlegendentry{$\sigma{}_{\text{b}_\text{1}\text{b}_\text{2}}\text{(analytic)}$}
\node at (rel axis cs:0, 1) [anchor=north west] {(d)};
\end{axis}
\draw [dashed, thick, black] (\figurewidth, -2.5em) -- (\figurewidth, \figureheight);
\end{tikzpicture}%

%% file: Graphs/noise_5sol_b.tex
\definecolor{mycolor1}{rgb}{0,0,0}%
\definecolor{mycolor2}{rgb}{0.00000,1.00000,0.40000}%
\definecolor{mycolor3}{rgb}{0.00000,0.40000,1.00000}%
\definecolor{mycolor4}{rgb}{0.80000,0.00000,1.00000}%
\begin{tikzpicture}
\begin{axis}[%
width=0.951\figurewidth,
height=\figureheight,
at={(0\figurewidth,0\figureheight)},
scale only axis,
xmin=2,
xmax=20,
xtick=data,
xlabel style={font=\color{white!15!black}},
xlabel={SNR (dB)},
ymode=log,
ymin=0.001,
ymax=1,
ytick={1e-1,1e-2,1e-3,1e-4,1e-5},
yticklabels={-1,-2,-3,-4,-5},
yminorticks=true,
ylabel style={font=\color{white!15!black}},
axis background/.style={fill=white},
axis x line*=bottom,
axis y line*=left,
legend style={legend cell align=left, align=left, draw=white!15!black}
]
\addplot [color=red, only marks, semithick, mark=*, mark options={solid, red, scale=0.5, fill=red}, forget plot]
  table[row sep=crcr]{%
  	2	0.182234234125358\\
  	5	0.0864756563383983\\
  	8	0.041517794374413\\
  	11	0.0194167741583234\\
  	14	0.00988586042653986\\
  	17	0.00448380619115067\\
};
%\addlegendentry{$\text{Var(b}_{\text{1}}\text{) (FB)}$}
\addplot [color=red, only marks, mark=o, mark options={solid, red}, forget plot]
  table[row sep=crcr]{%
  	2	0.184022714219697\\
  	5	0.0868184379602661\\
  	8	0.0416478444039245\\
  	11	0.0194459588585983\\
  	14	0.00989545347570694\\
  	17	0.00449256328006551\\
};
%\addlegendentry{$\text{Var(b}_{\text{1}}\text{) (FC)}$}
\addplot [color=red, forget plot]
  table[row sep=crcr]{%
2	0.148733644099135\\
5	0.0745434036333485\\
8	0.0373602022521591\\
11	0.018724456414515\\
14	0.00938445851156522\\
17	0.00470337080050128\\
};
%\addlegendentry{$\text{Var(b}_{\text{1}}\text{) (analytic)}$}
\addplot [color=mycolor1, only marks, semithick, mark=square*, mark options={solid, mycolor1, scale=0.5, fill=mycolor1}, forget plot]
  table[row sep=crcr]{%
  	2	0.586622145858169\\
  	5	0.286118600192227\\
  	8	0.134372215514889\\
  	11	0.0616115157470355\\
  	14	0.0304621758691393\\
  	17	0.0146350697544885\\
};
%\addlegendentry{$\text{Var(b}_{\text{2}}\text{) (FB)}$}
\addplot [color=mycolor1, only marks, mark=square, mark options={solid, mycolor1}, forget plot]
  table[row sep=crcr]{%
  	2	0.586878691708413\\
  	5	0.285952169805071\\
  	8	0.134239098111943\\
  	11	0.0616291634497562\\
  	14	0.0304447127348255\\
  	17	0.0146377310477736\\
};
%\addlegendentry{$\text{Var(b}_{\text{2}}\text{) (FC)}$}
\addplot [color=mycolor1, forget plot]
  table[row sep=crcr]{%
2	0.464094746153437\\
5	0.232598361965592\\
8	0.11657532957977\\
11	0.0584260669412726\\
14	0.0292823988620182\\
17	0.0146759644796253\\
};
%\addlegendentry{$\text{Var(b}_{\text{2}}\text{) (analytic)}$}
\addplot [color=mycolor2, only marks, semithick, mark=triangle*, mark options={solid, mycolor2, rotate=90, scale=0.5, fill=mycolor2}, forget plot]
  table[row sep=crcr]{%
  	2	0.294887357334178\\
  	5	0.155841163306558\\
  	8	0.073282781398008\\
  	11	0.036744898075847\\
  	14	0.0183534185694892\\
  	17	0.00865073090798753\\
};
%\addlegendentry{$\text{Var(b}_{\text{3}}\text{) (FB)}$}
\addplot [color=mycolor2, only marks, mark=triangle, mark options={solid, rotate=90, mycolor2}, forget plot]
  table[row sep=crcr]{%
  	2	0.299326035447143\\
  	5	0.155918035358369\\
  	8	0.0732670032394039\\
  	11	0.0367807657974457\\
  	14	0.0183441550323491\\
  	17	0.0086512681669538\\
};
%\addlegendentry{$\text{Var(b}_{\text{3}}\text{) (FC)}$}
\addplot [color=mycolor2, forget plot]
  table[row sep=crcr]{%
2	0.285287037842079\\
5	0.14298222128579\\
8	0.0716608639441078\\
11	0.0359155101594877\\
14	0.0180003951811459\\
17	0.00902156826503617\\
};
%\addlegendentry{$\text{Var(b}_{\text{3}}\text{) (analytic)}$}
\addplot [color=mycolor3, only marks, semithick, mark=triangle*, mark options={solid, rotate=180, mycolor3, scale=0.5, fill=mycolor3}, forget plot]
  table[row sep=crcr]{%
  	2	0.109981581696544\\
  	5	0.0449426788046893\\
  	8	0.0208678848064767\\
  	11	0.00993770422037372\\
  	14	0.00463293631766126\\
  	17	0.00237426847220765\\
};
%\addlegendentry{$\text{Var(b}_{\text{4}}\text{) (FB)}$}
\addplot [color=mycolor3, only marks, mark=triangle, mark options={solid, rotate=180, mycolor3}, forget plot]
  table[row sep=crcr]{%
  	2	0.119752178805373\\
  	5	0.0433060269551409\\
  	8	0.0209132301088571\\
  	11	0.00994883412615883\\
  	14	0.00463494275379925\\
  	17	0.00238476731102411\\
};
%\addlegendentry{$\text{Var(b}_{\text{4}}\text{) (FC)}$}
\addplot [color=mycolor3, forget plot]
  table[row sep=crcr]{%
2	0.0699479501888007\\
5	0.0350570196530233\\
8	0.0175701306991157\\
11	0.00880592519955939\\
14	0.00441341729029586\\
17	0.00221194840256615\\
};
%\addlegendentry{$\text{Var(b}_{\text{4}}\text{) (analytic)}$}
\addplot [color=mycolor4, only marks, semithick, mark=pentagon*, mark options={solid, mycolor4, scale=0.5, fill=mycolor4}, forget plot]
  table[row sep=crcr]{%
  	2	0.122405847411139\\
  	5	0.0508081371176394\\
  	8	0.0257458219753479\\
  	11	0.0132412203479704\\
  	14	0.00634655788071731\\
  	17	0.00323244111813474\\
};
%\addlegendentry{$\text{Var(b}_{\text{5}}\text{) (FB)}$}
\addplot [color=mycolor4, only marks, mark=pentagon, mark options={solid, mycolor4}, forget plot]
  table[row sep=crcr]{%
  	2	0.0966989648169794\\
  	5	0.0485656924604706\\
  	8	0.0258512264711232\\
  	11	0.0131512262104282\\
  	14	0.00643933301505652\\
  	17	0.00326815884118052\\
};
%\addlegendentry{$\text{Var(b}_{\text{5}}\text{) (FC)}$}
\addplot [color=mycolor4, forget plot]
  table[row sep=crcr]{%
2	0.10284879913707\\
5	0.0515465051213952\\
8	0.0258344503049461\\
11	0.0129478966806172\\
14	0.00648932051865026\\
17	0.00325236459886302\\
};
%\addlegendentry{$\text{Var(b}_{\text{5}}\text{) (analytic)}$}
\node at (rel axis cs:0.5, 1) [anchor=north] {(e)};
\node at (17, 0.00500337080050128) [anchor=west] {\footnotesize ${b_1}$};
\node (top) at (17, 0.0186759644796253) [anchor=west] {\footnotesize ${b_2}$};
\node at (17, 0.00902156826503617) [anchor=west] {\footnotesize ${b_3}$};
\node at (17, 0.00181194840256615) [anchor=west] {\footnotesize ${b_4}$};
\node at (17, 0.00325236459886302) [anchor=west] {\footnotesize ${b_5}$};
\node at (18, 0.03) [anchor=south, align=center] {Var\\of};
\end{axis}
\end{tikzpicture}%

%% file: Graphs/noise_sech_b.tex
\definecolor{mycolor1}{rgb}{0,0,0}%
\begin{tikzpicture}
\begin{axis}[%
width=0.951\figurewidth,
height=\figureheight,
at={(0\figurewidth,0\figureheight)},
scale only axis,
xmin=2,
xmax=17,
xtick=data,
xlabel style={font=\color{white!15!black}},
xlabel={SNR (dB)},
ymode=log,
ymin=8e-5,
ymax=1,
yminorticks=true,
ytick={1e-0,1e-1,1e-2,1e-3,1e-4},
yticklabels={0,-1,-2,-3,-4},
ytick={1, 1e-1, 1e-2, 1e-3, 1e-4},
axis background/.style={fill=white},
axis x line*=bottom,
axis y line*=left,
legend style={font=\footnotesize, at={(1,1)}, anchor=north east, inner sep=0cm, legend cell align=left, align=left, draw=white!15!black}
]
\addlegendimage{empty legend}
\addlegendentry{\tikz\draw[color=red, scale=0.5, fill=red] plot[mark=*] (0,0); \tikz\draw[color=red, fill=red] plot[mark=o] (0,0); $\sigma_{b_1}^2$ \quad \tikz\draw[color=black, scale=0.5, fill=black] plot[mark=square*] (0,0); \tikz\draw[color=black, fill=black] plot[mark=square] (0,0); $\sigma_{b_2}^2$}
\addlegendimage{empty legend}
\addlegendentry{\tikz\draw[color=black!50!green, scale=0.5, fill=black!50!green] plot[mark=triangle*] (0,0); \tikz\draw[color=black!50!green, fill=black!50!green] plot[mark=triangle] (0,0);$|\Real{\sigma_{b_1b_2}}|$}
\addplot [color=red, only marks, semithick, mark=*, mark options={solid, red, scale=0.5, fill=red}, forget plot]
table[row sep=crcr]{%
2	0.0471825965888476\\
5	0.0216630831366008\\
8	0.0110351048792599\\
11	0.00548337239995644\\
14	0.00259617381969024\\
17	0.0012782535723079\\
};
%\addlegendentry{$\text{Var(b}_{\text{1}}\text{) (FB)}$}
\addplot [color=red, only marks, mark=o, mark options={solid, red}, forget plot]
table[row sep=crcr]{%
2	0.165527469529481\\
5	0.0436829629918157\\
8	0.0162547589222057\\
11	0.00645088206238543\\
14	0.00287619152733384\\
17	0.00136380383562415\\
};
%\addlegendentry{$\text{Var(b}_{\text{1}}\text{) (FC)}$}
\addplot [color=red, forget plot]
table[row sep=crcr]{%
2	0.0414250209851634\\
5	0.0207616916705057\\
8	0.0104054948137632\\
11	0.00521510116023288\\
14	0.0026137421235835\\
17	0.00130997418433389\\
};
%\addlegendentry{$\text{Var(b}_{\text{1}}\text{) (analytic)}$}
\addplot [color=mycolor1, only marks, semithick, mark=square*, mark options={solid, mycolor1, scale=0.5, fill=mycolor1}, forget plot]
table[row sep=crcr]{%
2	0.0402143138050757\\
5	0.0172612292112312\\
8	0.00854093052413681\\
11	0.00436151795545162\\
14	0.00205349866638274\\
17	0.00107410456551736\\
};
%\addlegendentry{$\text{Var(b}_{\text{2}}\text{) (FB)}$}
\addplot [color=mycolor1, only marks, mark=square, mark options={solid, mycolor1}, forget plot]
table[row sep=crcr]{%
2	0.0362221828985733\\
5	0.0172531851410668\\
8	0.00870413046331708\\
11	0.00441361128079098\\
14	0.00207533721894606\\
17	0.00108517213120571\\
};
%\addlegendentry{$\text{Var(b}_{\text{2}}\text{) (FC)}$}
\addplot [color=mycolor1, forget plot]
table[row sep=crcr]{%
2	0.0332697841003588\\
5	0.0166743910566355\\
8	0.00835699192609445\\
11	0.00418841766488473\\
14	0.00209918146273918\\
17	0.0010520829501919\\
};
%\addlegendentry{$\text{Var(b}_{\text{2}}\text{) (analytic)}$}
\addplot [color=black!50!green, only marks, semithick, mark=triangle*, mark options={solid, black!50!green, scale=0.5, fill=black!50!green}, forget plot]
table[row sep=crcr]{%
2	0.00692582868128535\\
5	0.00225353579092406\\
8	0.000994783826108619\\
11	0.000596158539551751\\
14	0.000231664486733657\\
17	0.000102019487221337\\
};
%\addlegendentry{$\sigma{}_{\text{b}_\text{1}\text{b}_\text{2}}\text{(FB)}$}
\addplot [color=black!50!green, only marks, mark=triangle, mark options={solid, black!50!green}, forget plot]
table[row sep=crcr]{%
2	0.00722328300269955\\
5	0.00198606009193092\\
8	0.00099091986001809\\
11	0.000581401182643371\\
14	0.000226871361040128\\
17	9.67729332790611e-05\\
};
%\addlegendentry{$\sigma{}_{\text{b}_\text{1}\text{b}_\text{2}}\text{(FC)}$}
\addplot [color=black!50!green, forget plot]
table[row sep=crcr]{%
2	0.00383591699181119\\
5	0.0019225126255497\\
8	0.000963538784412757\\
11	0.000482913337852414\\
14	0.000242029999879964\\
17	0.000121302346094648\\
};
%\addlegendentry{$\sigma{}_{\text{b}_\text{1}\text{b}_\text{2}}\text{(analytic)}$}
\node at (rel axis cs:0, 1) [anchor=north west] {(f)};
\end{axis}
\end{tikzpicture}%

%% file: Graphs/C_big.tex
\begin{tikzpicture}

\begin{axis}[%
width=\figurewidth,
height=\figureheight,
at={(0\figurewidth,0\figureheight)},
enlargelimits=false,
axis x line*=top,
axis y line*=right,
scale only axis,
xmin=-0.5,
xmax=7.5,
xtick={0,1,2,3,4,5,6,7},
xticklabels={$\Re\lambda_1$,$\Re\lambda_2$,$\Im\lambda_1$,$\Im\lambda_2$,$\Re b_1$,$\Re b_2$,$\Im b_1$,$\Im b_2$},
xtick style={draw=none},
y dir=reverse,
ymin=-0.5,
ymax=7.5,
ytick={0,1,2,3,4,5,6,7},
yticklabels={$\Re\lambda_1$,$\Re\lambda_2$,$\Im\lambda_1$,$\Im\lambda_2$,$\Re b_1$,$\Re b_2$,$\Im b_1$,$\Im b_2$},
axis background/.style={fill=white},
ytick style={draw=none},
]
\end{axis}

\begin{axis}[%
width=\figurewidth,
height=\figureheight,
at={(0\figurewidth,0\figureheight)},
enlargelimits=false,
scale only axis,
axis on top,
xmin=-0.5,
xmax=7.5,
xtick={0,1,2,3,4,5,6,7},
xticklabels={$\Re\lambda_1$,$\Re\lambda_2$,$\Im\lambda_1$,$\Im\lambda_2$,$\Re b_1$,$\Re b_2$,$\Im b_1$,$\Im b_2$},
xtick style={draw=none},
y dir=reverse,
ymin=-0.5,
ymax=7.5,
ytick={0,1,2,3,4,5,6,7},
yticklabels={$\Re\lambda_1$,$\Re\lambda_2$,$\Im\lambda_1$,$\Im\lambda_2$,$\Re b_1$,$\Re b_2$,$\Im b_1$,$\Im b_2$},
axis background/.style={fill=white},
ytick style={draw=none},
legend style={legend cell align=left, align=left, draw=white!15!black},
colormap name=viridis,
% colormap={custom}{rgb(0cm)=(0.4,0.4,1); rgb(1cm)=(1,1,0.6)},
colorbar,
colorbar left,
colorbar style={ytick={-4,-5,-6,-7,-8,-9,-10,-11,-12,-13,-14,-15,-16,-17,-18}, ylabel={\normalsize $\log_{10} |\text{covariance}|$}}
]
\addplot[matrix plot, mesh/cols=8, point meta=explicit] %, nodes near coords*=\n, visualization depends on={\thisrow{n} \as \n}]
table
[meta=C] {
x y n C
0 0 1.550990e-05 -4.809391
0 1 3.817154e-06 -5.420272
0 2 4.771106e-07 -6.321933
0 3 -3.689371e-07 -6.431528
0 4 3.153282e-05 -4.498460
0 5 3.758135e-05 -4.411338
0 6 1.638509e-06 -5.804365
0 7 -3.884218e-07 -5.965497
1 0 3.817154e-06 -5.420272
1 1 6.295059e-05 -4.203471
1 2 6.246284e-06 -5.205294
1 3 2.871338e-06 -5.542575
1 4 1.418448e-06 -5.934781
1 5 7.350605e-05 -4.142524
1 6 1.433068e-05 -4.872833
1 7 4.698400e-06 -6.394210
2 0 4.771106e-07 -6.321933
2 1 6.246284e-06 -5.205294
2 2 1.104331e-04 -3.956900
2 3 9.605253e-05 -4.018712
2 4 1.750668e-06 -5.604913
2 5 6.870847e-06 -5.285379
2 6 5.955473e-05 -4.229604
2 7 1.237405e-04 -3.914794
3 0 -3.689371e-07 -6.431528
3 1 2.871338e-06 -5.542575
3 2 9.605253e-05 -4.018712
3 3 1.661889e-04 -3.781725
3 4 1.305776e-06 -5.564463
3 5 -6.661038e-06 -5.423780
3 6 -1.011774e-04 -3.990338
3 7 2.046180e-04 -3.704560
4 0 3.153282e-05 -4.498460
4 1 1.418448e-06 -5.934781
4 2 1.750668e-06 -5.604913
4 3 1.305776e-06 -5.564463
4 4 2.110516e-04 -3.676471
4 5 1.211176e-04 -3.929483
4 6 6.663893e-06 -5.219104
4 7 6.825170e-08 -5.533356
5 0 3.758135e-05 -4.411338
5 1 7.350605e-05 -4.142524
5 2 6.870847e-06 -5.285379
5 3 -6.661038e-06 -5.423780
5 4 1.211176e-04 -3.929483
5 5 3.658076e-04 -3.490745
5 6 1.530377e-05 -4.996011
5 7 2.624865e-06 -5.366189
6 0 1.638509e-06 -5.804365
6 1 1.433068e-05 -4.872833
6 2 5.955473e-05 -4.229604
6 3 -1.011774e-04 -3.990338
6 4 6.663893e-06 -5.219104
6 5 1.530377e-05 -4.996011
6 6 8.646271e-04 -3.063577
6 7 -8.028003e-05 -4.070202
7 0 -3.884218e-07 -5.965497
7 1 4.698400e-06 -6.394210
7 2 1.237405e-04 -3.914794
7 3 2.046180e-04 -3.704560
7 4 6.825170e-08 -5.533356
7 5 2.624865e-06 -5.366189
7 6 -8.028003e-05 -4.070202
7 7 4.353491e-04 -3.420742
};
% \addplot [forget plot] graphics [xmin=0.5, xmax=8.5, ymin=0.5, ymax=8.5] {Graphs/C_big-1.png};

\end{axis}
\draw [ultra thick] (-0.15\figurewidth, 0.5\figureheight) -- (1.15\figurewidth, 0.5\figureheight);
\draw [ultra thick] (0.5\figurewidth, -0.25\figureheight) -- (0.5\figurewidth, 1.25\figureheight);

%\draw [{*-}] (0.25\figurewidth, 0.75\figureheight) -- (0.1\figurewidth, 1.1\figureheight) node [anchor=south] {\large $\sigma^2\mathbf{C}_{\tilde{\mathbf{\lambda}}}$};
\node at (0.1\figurewidth, 1.1\figureheight) [anchor=south] {\large $\sigma^2\mathbf{C}_{\tilde{\mathbf{\lambda}}}$};
\node at (0.9\figurewidth, 1.1\figureheight) [anchor=south] {\large $\sigma^2\mathbf{C}_{\tilde{\mathbf{\lambda}}\tilde{\mathbf{b}}}$};
\node at (0.1\figurewidth, -0.1\figureheight) [anchor=north] {\large $\sigma^2\mathbf{C}_{\tilde{\mathbf{\lambda}}\tilde{\mathbf{b}}}^{\mathrm{T}}$};
\node at (0.9\figurewidth, -0.1\figureheight) [anchor=north] {\large $\sigma^2\mathbf{C}_{\tilde{\mathbf{b}}}$};

\end{tikzpicture}%